\documentclass{ruthesis}
\figurespage
\usepackage{amsmath}
\usepackage{amssymb}
\usepackage{amsthm}
\usepackage{rafael_thesis}
\usepackage{picture}
\usepackage{graphicx}
\newtheorem{theorem}{Theorem}[section]
\newtheorem{lemma}[theorem]{Lemma}
\newtheorem{corollary}[theorem]{Corollary}
\newtheorem{definition}[theorem]{Definition}
\newtheorem{proposition}[theorem]{Proposition}
\newtheorem{assumption}[theorem]{Assumption}

\begin{document}
\phd\title{Effects of Quenched Randomness on Classical and Quantum Phase Transitions}
\author{Rafael L. Greenblatt}
\director{Joel L. Lebowitz}
\program{Physics and Astronomy}
\approvals{4}
\submissionyear{2010}
\submissionmonth{October}
%\email{rafaelgr@physics.rutgers.edu}
%\affiliation{Department of Physics and Astronomy, Rutgers University, Piscataway NJ 08854-8019, USA}
%\date{\today}
%\maketitle
%\tableofcontents
\abstract{
This dissertation describes the effect of quenched randomness on first order phase transitions in lattice systems, classical and quantum.  It is proven that a large class of quantum lattice systems in low dimension ($d \le 2$ or, with suitable continuous symmetry, $d \le 4$) cannot exhibit first-order phase transitions in the presence of suitable (``direct'') quenched disorder.
}
\beforepreface

\acknowledgements{
Nothing is ever truly the work of a single person.  Much of the content of this dissertation was shaped in the course of conversations with colleagues and many distinguished scientists, among them Royce Zia, Moshe Schechter, Enza Orlandi, Vieri Mastropietro, Tobias Kuna, Michael Kiessling, Sheldon Goldstein, Alessandro Giuliani, Giovanni Gallavotti, Diana David-Rus, Eric Carlen, Gabe Bouch, and especially Michael Aizenman.

Anything else achieved here would have been impossible without the support and encouragement of my friends; besides those I have already listed above, I owe special thanks to Joe Cleffie, Anand Gopal, Alicia Graham, Sarah Grey, and Pankaj Mehta.  I am also grateful to my brothers and sisters in the Rutgers Council of AAUP Chapters-AFT for all they have done to improve the working lives of graduate employees like myself.

Ron Ransome, the director of the physics graduate program for most of my time here, has been extremely helpful at many points, particularly in making it possible for me to resume my studies after an interruption.

Above all, I thank my advisor and mentor, Joel Lebowitz, who has helped me (like so many other students) in more ways than I can express here, not least by the remarkable example he has set for all of us.
}

\dedication{
\begin{center}
{\em To my first math teacher and my first physics teacher:\\
my mother, Susan Selvig\\
and my father, Richard Greenblatt} \\
\end{center}
}

%\figurespage
%\tablespage

\afterpreface

\chapter{Introduction}\label{Background.section}

%\section{The physics of quenched disorder}

One of the basic techniques of condensed matter physics is the effective description of solids as a combination of a static portion and a rapidly-moving part: in a simple description of metals, the nuclei and tightly-bound electrons remain fixed and form an effective potential background for the conduction electrons.  Although the simplest description involves a uniform or periodic background (i.e.\ a perfect crystal), this is hardly a natural assumption: completely pure samples are anything but common or easily prepared.  There are many situations in which disorder has only a minor effect, but there are cases of fundamental importance where this is not the case.  The best-established illustration is the description of electrical conduction in metals: with the application of quantum mechanics in this context it became clear that irregularly-placed scatterers were necessary to account for finite conductivity.  Anderson localization provides a further way in which disorder produces a qualitative difference in the behavior of a physical system.

These examples concern transport properties, which are harder to fit into a comprehensive framework than equilibrium properties.  The core of the work described here is a similarly qualitative effect at the level of equilibrium thermodynamics, the rounding effect predicted by Imry and Ma~\cite{YM} and described in detail in the next section.

It is misleading in a way to talk about equilibrium in this context.  As noted already in the paper which introduced the mathematical framework now known as quenched randomness~\cite{Brout}, it is important to consider a background which is a metastable configuration, and is not typical of the equilibrium state of the full system.  Although disorder is still present in systems which are genuinely in equilibrium (this is what is known as annealed disorder), annealed systems cannot exhibit behavior which is fundamentally different from that of ordered systems.

There is still much that remains to be understood about classical models of quenched randomness, but the situation for quantum systems has been even more obscure.  The main results presented in this dissertation, Propositions~\ref{main_prop} and~\ref{continuous_prop}, are an extension of the proof of the Imry-Ma rounding effect to quantum systems.  The present chapter will discuss the previous state of understanding and attempt to provide context for the result.  %Chapter~\ref{Background.section} discusses the problem and the systems under consideration, summarizing a variety of rigorous and nonrigorous results and some of the methods used to obtain them.  Appendix~\ref{numerics.chapter} supplements this with a description of commonly used methods for numerical studies and a proposed new method for Monte Carlo simulations of certain continuous-spin models.
Chapters~\ref{Begin.proof} to~\ref{End.proof} comprise the proof of these results.  Chapter~\ref{Begin.proof} describes the formalism used, establishes several preliminary results, and states the main propositions to be established.  Chapter~\ref{CLT_section} contains a proof of a nonlinear central limit theorem (based on an earlier result of Aizenman and Wehr~\cite{AW.CMP}) which may be of some independent interest.  Chapter~\ref{G_section} completes the proof of the main results with an analysis of the free energy effects of the quenched randomness.  %These chapters are a draft of a paper being prepared with M. Aizenman and J. L. Lebowitz;
This work was announced in a publication by the author with M. Aizenman and J. L. Lebowitz, which provides a summary of the argument as is therefore attached as~\cite{QIMLetter}.  Additionally, Appendix~\ref{prob_appendix} reviews some probabilistic terminology and results used in the previous chapters which may be unfamiliar to some readers.

Appendix~\ref{GZRP} (written with J.L. Lebowitz, and published as~\cite{GZRP}) describes earlier work by the author, with J. L. Lebowitz, on nonequilibrium stochastic dynamics.

%Section describing Imry-Ma and early attempts to make rigorous; disagreement with Parisi-Sourlas

\section{The rounding effect for classical systems}

A 1975 paper by Imry and Ma contains an important insight into phase transitions in disordered systems based on an analysis of the energy of the ordered phase, using considerations similar to those applied more rigorously in Peierls' proof of long range order in the Ising model~\cite{Peierls} and later in Pfister's proof of the Mermin-Wagner theorem for classical systems~\cite{Pfister81MW}.  The context is $O(N)$ models, that is lattice models where configurations consist of a specification of an $N$-dimensional unit vector (a classical spin) $\vec{\sigma}_x$ at each site $x$, with equilibrium states determined by the Hamiltonian
\begin{equation}\label{ON.simple.Hamiltonian}
 \ham= -J \sum \vec{\sigma}_x \cdot \vec{\sigma}_y - \sum \vec{h}_x \cdot \vec{\sigma_x}.
\end{equation}

Since the paramagnetic phase of this system has higher entropy, for the ordered phase to be stable requires that the energy cost involved in forming a domain where the spins inside are aligned in a different direction from those outside grow with the size of such a domain.  In the Ising ($N=1$) case the cost for a domain of diameter $L$ is of order $L^{d-1}$, and in continuous ($N \ge 2$) versions spin-wave analysis~\cite{HerringKittel} suggests a cost on the order of $L^{d-2}$.  In the absence of a random field this suggests (correctly) that ferromagnetism does not exist in these systems at finite temperature for $d=1$ and $d \le 2$ respectively. Since ferromagnetism appears at any higher dimension\footnote{With the possible exception of the symmetric quantum case, where a ferromagnetic phase has yet to be rigorously shown to exist in three dimensions.}, we may speculate that it is sufficient for the energy cost to grow with $L$, a contention that is supported by estimates of the number of genuinely independent contours of given size~\cite{Chalker,FFS}.   If the random field has typical strength $H$ and we neglect correlations between the field at different locations, then the total random field in a domain of volume $L^d$ will typically have magnitude $H L^{d/2}$.  Then when $d \le 2$ for Ising models and $d\le 4$ for continuous models, and given any direction, there will be a large number of large domains for which flipping into that direction is energetically favored. On this basis, Imry and Ma predicted that there would be no long range order at low temperature for the random field Ising model in two dimensions and for similar continuous models in $d \le 4$.  They also suggested that ferromagnetism would persist in higher dimensions.

Another way of looking at ferromagnetic order in this system is as a first order transition, where the equilibrium magnetization $\state{\vec{\sigma}}$ changes discontinuously as the external field $\vec{h}$ is changed through zero.  The disappearance of ferromagnetic order corresponds to a ``rounding'' of this discontinuity, leaving a continuous transition.  We shall see that this ``rounding effect'' occurs in a large number of systems in the presence of quenched randomness.

In 1976 Aharony, Imry and Ma established a detailed connection between random field $O(N)$ models with continuous spin in $4<d<6$ dimensions and the field-free versions in $d-2$, finding an exact correspondence between the most divergent Feynman diagrams of all orders for the two models~\cite{AYM}; among other things this provided strong support (which had previously been lacking) for the prediction that the random field models had ferromagnetic order for $d>4$.  However by expressing the Lagrangian of the model in a supersymmetric form, Parisi and Sourlas were able to extend this perturbative correspondence to all dimensions and to $n=1$, suggesting that the random field Ising model was not, in fact, ferromagnetic in three dimensions but only for four dimensions or more~\cite{ParisiSourlas}, or in other words that its lower critical dimension $d_l$ was 3.  A number of attempts to study the formation of domain walls more carefully than Imry and Ma seemed at first to agree on $d_l=3$~\cite{Pytte81RFIM,Kogon81RFIM,Binder81RFIM}, but before long other domain-wall studies appeared to return to $d_l=2$~\cite{GrinsteinMa,Binder83}, along with other theoretical~\cite{Villain} and experimental work~\cite{Belanger83RFIM}; in particular Chalker~\cite{Chalker} and Fisher, Fr\"{o}hlich and Spencer~\cite{FFS} provided strong (but not conclusive) arguments for $d_l=2$ based on a rigorous treatment of the ``no contours within contours'' approximation.  However further arguments emerged for $d_l=3$~\cite{Krey}, and the debate was only resolved with rigorous proofs of long range order for the 3 dimensional random field Ising model by Imbrie~\cite{Imbrie.PRL,Imbrie.CMP} (for zero temperature) and Bricmont and Kupiainen~\cite{BK.PRL,BK.CMP} (for low temperature), based on intricate examinations of the scaling behavior of the contour representations of the model.

This did not yet completely vindicate Imry and Ma's argument; this was done by Aizenman and Wehr, who proved that first order transitions could not exist for a large variety of classical systems in the presence of disorder~\cite{AW.PRL,AW.CMP}.  They were able to do this by first constructing a suitable description of the equilibrium states of the infinite system (metastates), which allowed the construction of a quantity describing the free energy fluctuations due to the random term in the Hamiltonian.  The estimates of domain energies in the Imry-Ma argument correspond to rigorous bounds on this quantity, and by examining only hypercubic domains it is possible to show that a first order transition would cause a contradiction between these bounds in the dimensions which Imry-Ma predicted a rounding effect, that is always in  $d \le 2$, and for systems with continuous symmetries $d \le 4$.

The precise conditions are somewhat cumbersome to state precisely.  They are exactly the same as those of Propositions~\ref{main_prop} and~\ref{continuous_prop} below, so for the moment I will confine myself to some general remarks.  The main one is on the relationship between the quenched disorder and the order parameter.  If the phase transition or long range order under examination is described by averages of some local quantity $\kappa_x$, then a rounding effect can be expected only when the Hamiltonian can be written in the form
\begin{equation}
  \mathcal{H} = \mathcal{H}_0 - \sum_x \left( h + \epsilon \eta_{x} \right) \lop_{x};
\end{equation}
following Hui and Berker~\cite{HuiBerker.PRL}, we can refer to this as ``direct randomness''\label{Direct.label}. In the Ising model, a random field is direct with respect to the spins, and so an arbitrarily weak random field eliminates ferromagnetism in two dimensions; bond randomness is not direct, so it does not (at least not when it is sufficiently weak).  It should be noted that this notion is relative to a particular phase transition: by way of illustration we may consider the random-bond Potts model in two dimensions.  In the nonrandom version of this model with sufficiently many colors, the order-disorder transition involves a nonzero latent heat, which in this case means that the equilibrium bond energy density is discontinuous with respect to the bond strength.  This is a first order transition for which bond randomness is direct, and therefore the latent heat vanishes whenever it is present.  On the other hand this randomness does not couple to the color, and so as in the Ising model long range order remains.

\subsection{Ising models}\label{background.lowDIsing}

More can be said about the random field Ising model (henceforth RFIM) by bringing a variety of specialized techniques to bear, leading to more insight into the scope and significance of the rounding effect.  In particular, some insight can perhaps be obtained into what replaces the ferromagnetic phase in this situation, and into the more subtle effects of the random field in 3 dimensions.

Besides computational efficiency, one major issue in simulating disordered systems is the presence of additional finite size effects, especially at the lower critical dimension.  To make this clearer, let us revisit the Imry-Ma~\cite{YM} analysis of the random field Ising model with random fields of typical strength $H$.  Flipping a typical domain of linear size $L$ will involve a bond energy of the order $JL^{d-1}$ and field energy on the order of $HL^{d/2}$; in one dimension, the field energy will dominate once length scales on the order of $(J/H)^2$ come into play, but when smaller systems are analyzed they will appear to be ferromagnetic.  In two dimensions the competing energies are both proportional to $L$; when $H$ is small compared to $J$ the rounding effect occurs only because of fluctuations in the random field, which makes its effect stronger in particular regions.  This can be studied by means of extreme value statistics, and this approach~\cite{Binder83} gives a breakup length scale on the order of
\begin{equation}
  L_b = \exp\left[ A (J/H)^2 \right],
\end{equation}
with a constant $A$ of order $1$.  This has been backed up by numerical studies, which found $A=2.1 \pm .2$ for a Gaussian distribution of the random fields and $1.9 \pm .2$ for a bimodal distribution~\cite{RFIM.breakup}.  For weak values of the random field, this distance can easily be hundreds or even thousands of sites - nowhere near macroscopic, but potentially very difficult to reach in simulations.

\subsubsection{Nearest neighbor Ising chain at zero temperature}\label{1dIsing.sec}
One requirement of the Aizenman-Wehr proof of the rounding effect are assumptions which must be made on the distribution of the random parameter in certain contexts.  Some limitations may be purely technical (see Section~\ref{restrictions_sec} below), but not all.  We can see this thanks to studies of the one-dimensional Ising model by Bleher et. al.~\cite{Bleher}.  Examining the case of a ``dichotomous'' random field, i.e.\ one taking only the two values $\pm H$ and those with equal probability, they found that the ground state configuration of the spin at any site $x$ could be deduced from the random field in some finite but undetermined neighborhood as follows.

Let us write the Hamiltonian of the system as
\begin{equation}
  \ham=-J \sum_x \sigma_x \sigma_{x+1} + \sum_x \eta_x \sigma_x.
\end{equation}
We can recursively define two position-dependent functions of the random fields by
\begin{gather}
  u_x = \piecewise{
  u_{x-1}+h_x, & |u_{x-1}+h_{x-1}| \le J \\
  J, & u_{x-1}+h_{x-1} > J \\
  -J, & u_{x-1}+h_{x-1} < -J }\label{ux_recursion}\\
  v_x = \piecewise{
  v_{x+1}+h_x, & |v_{x+1}+h_{x+1}| \le J \\
  J, & u_{x+1}+h_{x+1} > J \\
  -J, & u_{x+1}+h_{x+1} < -J }.
\end{gather}
$u_x$ (respectively $v_x$) can be thought of as representing the effect of $x$'s neighbors to the left (resp. right) on flipping it out of the ground state - the lowest-energy flip may involve a number of sites depending on the magnetic field they experience.
These quantities always exist, and are almost always uniquely specified since there will eventually be a large block of sites where all the magnetic fields point in the same direction.  Any site $x$ for which $u_x+v_x+h_x$ is positive (resp. negative) will necessarily have $\sigma_x = 1$ (resp. $-1$) in any ground state, and if $u_x+v_x+h_x=0$ there will be ground states with $\sigma_x = 1$ and $\sigma_x = -1$.  All this is proven (for $\eta_x= \pm H$) in~\cite{Bleher}, but readers should be able to convince themselves by considering the minimum energy cost of flipping a block of sites containing $x$ out of the resulting configuration; and Appendix~\ref{numerics.chapter} contains a derivation for arbitrary fields.

It is not difficult to numerically estimate the probability distribution of $u_0$ from the recursion relationship~\eqref{ux_recursion} (the distribution of $v_0$ is identical and independent), and from this calculate the average value of the ground state magnetization.  Figure~\ref{chain.fig.first} shows a plot resulting from such a calculation
\begin{figure}
  \includegraphics[width=\textwidth]{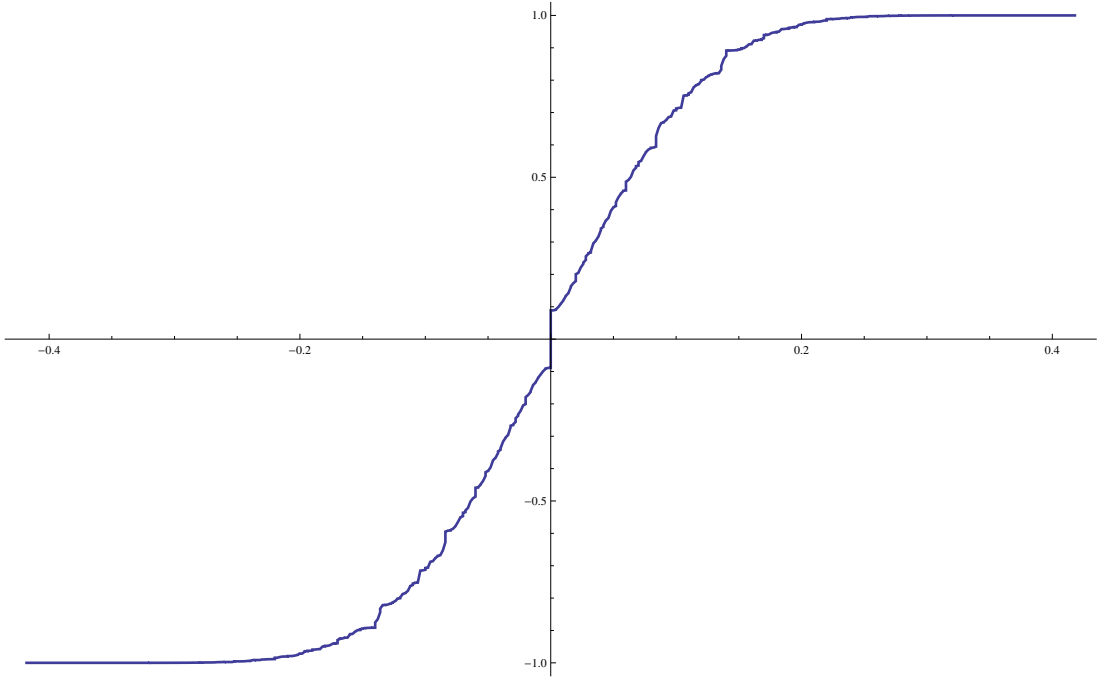}
  \caption{Plot of ground state magnetization as a function of mean magnetic field $h_0$ for a random field Ising chain with magnetic field distribution $\tfrac{1}{2}\delta_{\bar{h}-H}+\tfrac{1}{2}\delta_{\bar{h}+H}$, $J=1$, $H=0.42$ .\label{chain.fig.first}}
\end{figure}
%\begin{fullpagefigure}
%  \includegraphics[width=\textwidth]{zigg1}
%  \caption{Another Ising chain magnetization - heaviside distribution.\label{chain.fig.firstCT}}
%\end{fullpagefigure}
%\begin{fullpagefigure}
%%  \includegraphics[width=\textwidth]{zigg1}
%  \caption{Another Ising chain magnetization - triangle distribution?.\label{chain.fig.last}}
%\end{fullpagefigure}

As is apparent from Figure~\ref{chain.fig.first}, the magnetization in the presence of a dichotomous random field has a number of discontinuities, in fact an infinite number occurring wherever $\bar{h}/H$ is rational.  It is interesting to note that these first order transitions \emph{do not} correspond to any long range order: there are a finite density of isolated regions which can be flipped independently with no change in energy, resulting in a finite residual entropy; this situation was called ``Perestroika'' when first described in 1989~\cite{RFIM.residual}.\label{Perestroika}

Nonetheless, this illustrates that one of the restrictions on the proof of the rounding effect by Aizenman and Wehr~\cite{AW.CMP}, the requirement of an absolutely continuous distribution of the random field at zero temperature, is indeed necessary.

%As can be seen in Figures~\ref{chain.fig.firstCT} to~\ref{chain.fig.last},
These first order transitions do not appear for absolutely continuous distributions of the random field.  The curves do not, however, appear to be always analytic, and the character of the singularities (that is, the order of phase transitions present) appears to depend on the corresponding properties of the random field distribution in a way that remains to be investigated more carefully.

\subsubsection{Long range interactions in one dimension}

The one dimensional Ising model can exhibit long range order at finite temperature if interactions are sufficiently long range~\cite{Dyson69}.  Let us consider the variant of the RFIM with the following Hamiltonian:
\begin{equation}
  \ham= -J_0 \sum_{x<y} \frac{\sigma_x \sigma_y}{|y-x|^\alpha} - \sum_x h_x \sigma_x.
\end{equation}
The bond energy associated with flipping a block of $L$ spins is on the order of $L^{2-\alpha}$, so the Imry-Ma argument indicates that rounding should occur for $\alpha \ge 3/2$.  This was confirmed by Aizenman and Wehr~\cite{AW.CMP}, but the question of what happens for even longer ranged interactions remained unanswered until recent work by Cassandro, Orlandi and Picco~\cite{Cassandro}, who showed that long range order persists in the presence of weak random fields for $1-\ln(3/2) < \alpha <3/2$.  This means that the estimate provided by the Imry-Ma argument is also sharp in this respect, and suggests that the restrictions on long range interactions used below may in some sense be sharp as well.

\subsubsection{Higher dimensions}

There is a complication in the ground state behavior of the RFIM which has not been well studied: ``Perestroika'' (see p.~\pageref{Perestroika}) occurs in the ground state for dichotomous random fields in all finite dimensions, because there will be a finite density of regions where the magnetic field has a pattern like that shown in Figure~\ref{perestroika.fig}.
\begin{figure}
\begin{center}
\begin{picture}(200,200)
  \multiput(0,0)(0,20){11}{\line(1,0){200}}
  \multiput(0,0)(20,0){11}{\line(0,1){200}}
  \multiput(0,0)(20,0){4}{\multiput(5,5)(0,20){10}{$+$}}
  \multiput(120,0)(20,0){4}{\multiput(5,5)(0,20){10}{$-$}}
  \multiput(85,5)(0,20){4}{$+$}
  \multiput(105,5)(0,20){6}{$+$}
  \multiput(85,85)(0,20){6}{$-$}
  \multiput(105,125)(0,20){4}{$-$}
  \linethickness{1mm}
  \put(80,80){\line(0,1){120}}
  \put(80,80){\line(1,0){40}}
  \put(120,0){\line(0,1){120}}
  \put(80,120){\line(1,0){40}}
\end{picture}
\end{center}
  \caption{A two-dimensional example of a dichotomous random field configuration leading to ``Perestroika'' in the ground state. The bold lines divide separate regions where the spins are always $+$, indeterminate, and always $-$ in any ground states. \label{perestroika.fig}}
\end{figure}
In two dimensions, this means that a dichotomous RFIM will exhibit a first order transition at zero temperature, which has probability zero~\cite{AW.CMP} either for the same model at finite temperature or for any absolutely continuous distribution of the random fields.  In three dimensions there is also a more conventional degeneracy in the ground state due to long range order~\cite{Imbrie.CMP}, but numerical studies still indicate that dichotomous and absolutely continuous distributions show strikingly different behavior, not even lying in the same universality class at zero temperature~\cite{HartmannNowak,Sourlas99}.  If Perestroika is the main cause of the difference between the two cases then it should disappear at finite temperature, with all low temperature systems behaving like the zero temperature system with an absolutely continuous field distribution.  This appears to be supported by comparing finite-temperature Monte Carlo studies with dichotomous~\cite{Fytas08} and Gaussian~\cite{WuMachta} random fields, as well as comparing the latter to ground state studies with Gaussian fields~\cite{WuMachta}.  It is possible that renormalization studies which include parameters differentiating between the different distributions could shed light on the situation; then the scenario described above would involve additional dichotomous-field fixed points, all with this new parameter as an unstable direction.

Whatever the details, it is clear that different types of random field distribution can result in profoundly different types of behavior, especially in the ground state.

\subsection{The 3 dimensional XY model}\label{3DXY.section}

A claim has arisen recently that the prediction of a rounding effect for three and four dimensional systems with continuous symmetry (where dimensional reduction, the Imry-Ma argument, and the proof of Aizenman and Wehr are all in agreement) is either incorrect or misunderstood.  The controversy has been specifically about what is probably the simplest such model, the three dimensional random field XY model described by the Hamiltonian
\begin{equation}\label{RFXY_ham}
  \beta \ham = -J\sum_{<x,y>} \vec{\sigma}_x \cdot \vec{\sigma}_y - \sum_x \vec{h}_x\cdot \vec{\sigma}_x
\end{equation}
where $\sigma_x$ are unit vectors in $\R^2$, and $\vec{h}_x$ are i.i.d.\ random vectors in $\R^2$; this is the $N=2$ case of the $O(N)$ model discussed above.  In the following discussion, we can assume that $\vec{h}_x$ are chosen uniformly from some circle of specified radius $H$, as is done in most of the numerical studies we will discuss.

In 2007, Fisch~\cite{Fisch07} published the results of Monte Carlo simulations on the random field clock model, which shares the Hamiltonian~\eqref{RFXY_ham}, but where $\vec{\sigma}$ and $\vec{h}$ are now restricted to a set of $q$ evenly-spaced directions (in~\cite{Fisch07} $q=12$ is used).  Although this model does not have the continuous $U(1)$ symmetry of the XY model, there is certainly a relationship between the properties of the two models~\cite{Fisch07,Z12.Jose,Tobochnik}, and the clock model can be simulated very efficiently.

We can be more concrete in considering the nonrandom ($H=0$) model in two dimensions.  Here we know that the XY model has no long range order (i.e.\ no ferromagnetic phase) at finite temperature thanks to the Mermin-Wagner theorem~\cite{MerminWagner66,Mermin67,Pfister81MW}.  The clock model, on the other hand, has a finite number of ground states, which are related by a symmetry group with minimum interface energy of $(1-\cos 2\pi/q)J$ per bond, and so by Pirogov-Sinai theory~\cite{PirogovSinai1,PirogovSinai2} have ferromagnetic long range order for sufficiently small temperatures.

An early Monte Carlo study of two dimensional clock models~\cite{Tobochnik} noted that outside the ferromagnetic phase the clock model behaved similarly to the XY model, in particular showing evidence of a Kosterlitz-Thouless phase of quasi-long-range order.  It is not surprising that the relationship between the two systems should depend significantly on the temperature, since the relatively high energy excitations will be similar between the two systems.

It was claimed in~\cite{Fisch07} that the temperatures under consideration were high enough that differences between the clock and XY models would not come into play, but there are reasons to doubt that this is a reasonable line of argument.  Ferromagnetic order can always be disrupted by a system's lowest energy excitations; their effects only become less relevant with increasing temperature insofar as they are overwhelmed by other, more entropically favorable, excitations.  If the QLRO phase begins at a nonzero temperature, it is because it is only then that the associated modes begin to play a dominant role, and it is at exactly this temperature that the absence of sufficiently low-energy excitations in the clock model makes itself felt.

It is also worth noting that some studies of the clock model~\cite{Fisch10} involve some sites with zero magnetic field; although ferromagnetism is still ruled out at finite temperature (see Section~\ref{continuous_def} below), this does open the possibility of a ferromagnetic ground state or some other phenomenon similar to perestroika which could have subtle effects at finite temperature.
%This picture is supported by Fisch's earlier work which found that the ferromagnetic region for the twelve-state clock model was significant smaller than that of the six-state model, leading to the suggestion that it vanished entirely in the XY model~\cite{Fisch97}, and also that a different approximation with 256 states per spin showed no sign of long range order~\cite{Fisch00}.

%The problem of disorder-related finite size effects is also crucial in interpreting these results.  It seems reasonable to suppose that flipping a domain of diameter $L$ should involve a bond energy of about $3JL$, which is the spin wave energy associated with a surface area of $6L^2$, while the field energy should be about $HL^{3/2}$; then ordered domains would start to break up at a length scale of about $(3J/H)^2$.  The values of $J$ used for simulations in~\cite{Fisch07,Fisch10} are not reported, but the discussion of the lookup table method in~\cite{Fisch07} appears to indicate $J=1

It appears to be possible to virtually eliminate discretization effects with another scheme which at least provides substantial improvements over rejection sampling.  The idea, described in Appendix~\ref{numerics.chapter}, is based on the Ziggurat algorithm~\cite{Ziggurat}, a method which has proven to be highly efficient in sampling the normal distribution~\cite{Gaussian.RNG.review}.

Preliminary tests of this method have been very promising.  Using a C++ program on a desktop computer with a 3.2 GHz Pentium 4 processor and 2 GB RAM, I have been able to achieve an update rate of $8.6\times 10^5$ to $2.1 \times 10^6$ sites per second\footnote{The efficiency depends on the system size, which may indicate that further optimization is possible.} on a parameter range $k \in [0,14]$, compared to $6.8\times 10^5$ to $1.3 \times 10^6$ for a comparable lookup-table implementation of the 12-state clock model.  It seems very likely, then, that it will be possible in the near future to conduct a detailed study comparing the two models.

\section{The rounding effect for quantum systems}

\subsection{Transverse field Ising models: direct and orthogonal randomness}

The simplest quantum lattice spin system\footnote{That is, the simplest lattice spin system which involves nontrivial commutation relationships.  The Ising model, for example, has no classical dynamics, and is in a certain trivial sense a quantum system - the DLR conditions~\cite{Ruelle} which define its equilibrium state are equivalent to the quantum KMS conditions defined by Heisenberg evolution\cite{Simon}.  This is unlike off-lattice systems where the classical and quantum KMS conditions do not coincide~\cite{ClassicalKMS}.} is the transverse field Ising model, defined by the Hamiltonian
\begin{equation}
  \ham = - \sum J_{xy} \sigma_{3,x} \sigma_{3,y} - \sum \lambda_x \sigma_{1,x} - \sum h_x \sigma_{3,x},
\end{equation}
where $\sigma_{i,x}$ denotes the $i$ component of a $\tfrac{1}{2}$-spin at site $x$.  The properties of this system (and a number of variants) are relatively well known, in large part due to the fact that its path integral representation is the continuum limit of an Ising model with an additional dimension (sometimes called the ``space-time Ising model'')~\cite{AKN,AizenmanNachter}.  Its behavior is consequently very close to that of a classical Ising model in many respects; the nonrandom version is ferromagnetic when $\lambda$ is small, for example.  One difference is that the phase diagram of the system at zero temperature is more complicated than that of its classical counterpart; the system has a ferromagnetic-paramagnetic transition at zero temperature at a critical value of $\lambda$, which provides a paradigmatic example of a quantum critical point~\cite{Sachdev}.

Among the ways of introducing quenched randomness to this system, the most straightforward are to add randomness in the transverse field $\lambda_x$ or the longitudinal field $h_x$.  We can hardly expect quantum effects to be very striking in this system, but it is still worth clarifying where it fits into the picture I have been discussing.

The random transverse field case is a good example of orthogonal randomness, in that ferromagnetic order remains as long as the transverse field is not too strong~\cite{campanino1991lgs}; the nature of the ferromagnetic-paramagnetic transition can be changed significantly~\cite{Fisher}, but the details of this are beyond the scope of the present work.

A random longitudinal field, on the other hand, couples to the magnetization, and so should be direct randomness.  It has been expected~\cite{Senthil} that the outcome should be similar to the (classical) random field Ising model, to which it reduces for $\lambda_x \equiv 0$.

\subsection{The quantum Ashkin-Teller chain: an exception?}

The possibility of additional complications in the quantum case have been raised in the context of the quantum $N$-color Ashkin-Teller model.  In this system, each lattice site $x$ contains $N$ $\tfrac{1}{2}$-spins, described by operators $\ATSpin{i}{x}{\alpha}$ for the $i$th component of the $\alpha$ spin at site $x$.  In one dimension, the Hamiltonian is given by
\begin{equation}
\begin{split}
  \ham = &- \sum_{\alpha=1}^N \sum_x \left( J_x \ATSpin{3}{x}{\alpha} \ATSpin{3}{x+1}{\alpha} + h_x \ATSpin{1}{x}{\alpha} \right)
  \\&-\epsilon \sum_{\alpha < \beta}^N \sum_x \left( J_x \ATSpin{3}{x}{\alpha} \ATSpin{3}{x+1}{\alpha} \ATSpin{3}{x}{\beta} \ATSpin{3}{x+1}{\beta} + h_x \ATSpin{1}{x}{\alpha}  \ATSpin{1}{x}{\beta} \right).
  \end{split}
\end{equation}
To begin with, we examine the nonrandom version of the system, where $J_x=J$ and $h_x=h$ are constant.  At sufficiently low temperature (zero temperature in one dimension) and when $J$ is large compared to $h$, the system is in an ordered ``Baxter phase'', with the spins of each color exhibiting long range order, with no simple correlation between the different colors, while for large $h$ the system is paramagnetic.  For $N \ge 3$ and $\epsilon >0$, the transition between these states is of first order~\cite{Goswami}, characterized for example by a discontinuity in $\state{\ATSpin{1}{x}{\alpha}}$ (which is independent of $x$ and $\alpha$).  Randomness in $h_x$ is clearly direct with respect to this transition, and so is randomness in $J_x$ - the two can be shown to be equivalent by a duality transformation~\cite{Goswami}.  Therefore we should expect a system with such randomness to round the first order transition (at least provided it has an absolutely continuous distribution, cf. Section~\ref{1dIsing.sec} above).

A renormalization group analysis by Goswami, Schwab and Chakravarty~\cite{Goswami} suggested that this might not be the case.  They found that when $\epsilon$ was below a certain nonzero value $\epsilon_c(N)$, the flow of the system was similar to that of the random transverse-field Ising model and that there was no first order transition; however above this value their scaling analysis broke down in a way that led them to suggest that a first order transition might persist.  As we shall see this can be rigorously ruled out, but we are not yet in a position to say exactly what is happening.

\chapter{Proof of the rounding effect: overview and preliminaries}\label{Begin.proof}

%\section{Introduction}
We now embark on the proof of the rounding effect.  The basic framework of the argument is the same as~\cite{AW.CMP}, which in turn uses reasoning based on that of Imry and Ma~\cite{YM}.  One constructs a random variable $G_L$ which represents the free energy effect of the random field on a scale $L$.  We then show that it has a strict upper bound of the form
\begin{equation}\label{simpleUpperBound}
|G_L| \le C L^{d-1} + C' L^{d/2}
\end{equation}
or in more restricted cases
\begin{equation}\label{continuousUpperBound}
|G_L| \le C L^{d-2} + C' L^{d/2}.
\end{equation}
At the same time, we show that when the system is at a first order transition, it has asymptotic fluctuations described by a normal distribution,
 \begin{equation}\label{normalLimit}
   G_L \approx \mathcal{N}(0,L^{d/2})
 \end{equation}
on the scale $L^{d/2}$, which means that it will violate the above bounds in sufficiently low dimension.

The behavior indicated in Equation~\eqref{normalLimit} is akin to a central limit theorem, but instead of a sum of random variables it concerns a suitably continuous function of a large number of random variables.  In Chapter~\ref{CLT_section} we present a suitable nonlinear central limit theorem.  This result is a slight modification of one found in~\cite{AW.CMP}.  Although the result is phrased in what we hope will be a more useful form for some readers, the proof is substantially the same, apart from a correction due to Bovier~\cite{Bovier}.

The upper bound~\eqref{simpleUpperBound} is quite easy to show for finite systems, however it is not trivial to show that an infinite-system limit exists.  This problem was resolved for classical systems by defining $G_L$ as expectation values with respect to metastates, which are random probability measures related to the random Gibbs states of a disordered classical system~\cite{Bovier}.  The notion of metastate has been generalized to one suitable to quantum systems (that is, one based on the operator analysis notion of KMS state rather than the measure-theoretical notion of Gibbs state) by Barreto and Fidaleo~\cite{Barreto,Fidaleo}, but while this is promising for many other problems in disordered systems it is of little use to us.  Instead, we have formulated an argument which remains almost exclusively at a thermodynamic level.  This has the additional merit of producing a proof which is considerably more accessible from both a physical and a mathematical point of view.  Chapter~\ref{G_section} begins with the construction of an object satisfying the upper bound~\eqref{simpleUpperBound} and the conditions of the nonlinear central limit theorem proven in Chapter~\ref{CLT_section}, completes the proof of our first main result, and then provides the additional estimates needed to obtain the bound~\eqref{continuousUpperBound} under suitable conditions and obtain a stronger result for systems with continuous symmetry.

Before embarking on the proof, we establish definitions and a number of preliminary results which establish the context, and state our two main results.

\section{Notation and systems under consideration} \label{defs_section}
%Say precisely what a quenched quantum spin system is, provide examples, noting that classical case is basically included.
We consider systems on a lattice (we take this to be the simple cubic lattice $\Zd$ for simplicity, but many other cases can be reduced to this), where the possible configurations of each site are described by a finite-dimensional Hilbert space, with time evolution affected by a static background described by means of its statistical properties.

To make this more mathematically precise, we suppose that we are given a dimensionality $d$ and a finite-dimensional $C^*$-algebra\footnote{A $C^*$-algebra is a collection of operators with addition, multiplication, conjugation, which is closed under all of these operations (e.g.\ the product of two operators is another operator in the same algebra) and with a norm which defines limits, convergent series, etc.  This is a common way of formalizing the notion of the set of operators describing a quantum system. \cite{Ruelle,BR1,BR2}} $\alg_0$.  We introduce a copy $\alg_x$ of this algebra for each lattice site $x \in \Zd$, and take everything which can be obtained by tensor products, sums, and limits: this is the quasi-local $C^*$-algebra $\alg$~\cite{BR1,BR2,Simon}, and we will take the conventional point of view that this allows us to describe all physical observables.  We let $\finiteDom$ be the finite subsets of $\Zd$, and for any $\Lambda \in \finiteDom$ we let $\alg_\Lambda$ be the local $C^*$-algebra on $\Lambda$.

To specify the background referred to above, we will make use of the following concepts:
\begin{definition}
A field (on \Zd) is a map from $\Zd$ to the real numbers.  The set of all fields is denoted by \Fields.
\end{definition}

\begin{definition} \label{RandomField.def}
A random field (on \Zd) is a collection of random variables indexed by the elements of $\Zd$.  A random field is i.i.d.\ if the random variables it consists of are independently and identically distributed. %A random field is said to be independent if said random variables are independent, and to be Lyapunov if they are all Lyapunov.
\end{definition}
If we consider the space $\Fields$ to have the cylinder-Borel sigma algebra (the conventional choice), then an i.i.d.\ random field is also a $\Fields$-valued random variable (by Kolmogorov's extension theorem).

For Chapter~\ref{CLT_section}, as in many other generalizations of the central limit theorem, we need a restriction on the moments of the individual random variables:
\begin{definition}\label{Lyapunov.def}
A random variable $X$ is Lyapunov if there is a $\delta > 2$ such that $\Av |X|^\delta$ is finite.  A random field $\eta$ is Lyapunov iff each $\eta_x$ is Lyapunov.
\end{definition}
Note that an independent, Lyapunov random field defines an array (by restrictions to subsets of \Zd) which satisfies the usual Lyapunov condition, hence my appropriation of that name.

In what follows, strictly separate symbols will be used to denote random fields and specific values.  $\eta$ will be a random field (consisting of the individual real random variables $\eta_x$), while $\zeta$ is a specified (nonrandom) element of $\Fields$.  It is very convenient to have a compact way of referring to the random field within a specified subset of $\Zd$; to do so we use the symbol $\eta_\Lambda$ to refer to the collection of $\eta_x$ with $x \in \Lambda$; the meaning of expressions like $\zeta_\Lambda = 0$ should be clear.  This allows a convention we will use for conditional expectations: by expressions of the form
\begin{equation}
\condAv{f(\eta)}{\eta_\Lambda = \zeta_\Lambda}
\end{equation}
we mean a conditional expectation of the random variable $f(\eta)$ on the sigma-algebra generated by specifications of $\eta_\Lambda$, understood as a function of $\zeta$.

Dynamics (and equilibrium states) on such a structure are defined by way of the concept of an interaction, basically a rule for assigning Hamiltonians to families of systems defined on different finite regions.  Formally, a nonrandom interaction is a function $\inter_0:\finiteDom \to \alg$ satisfying $\inter_0(X) \in \alg_X$.

We wish to consider interactions depending on one or more random fields.  For the matter at hand, we do not need to talk about arbitrary random interactions; it is enough to talk about systems where the Hamiltonian on a finite region $\Gamma \in \finiteDom$ with free boundary conditions is
\begin{equation}
H_{\Gamma,0}^{h,\zeta,\ul{\omega}}=\sum_{X \subset \Gamma} \inter_0(X) + \sum_{x:T_x A_0 \subset \Gamma} (h+\zeta_x)\lop_x + \sum_{\alpha=1}^{N_\alpha} \sum_{x:T_x A_\alpha \subset \Gamma} \omega_{\alpha x} \gamma_{\alpha x},
\end{equation}
where $T_x$ denotes translation by $x$, and $\inter_0$ is assumed to be translation invariant ($\inter_0(T_x X) = T_x \inter_0 (X)$).  We define other boundary conditions as follows:
\begin{definition}
A boundary condition is a linear map $B:\finiteDom \times \alg \to \alg$, $(\Gamma,A) \mapsto B_\Gamma (A)$ satisfying
\begin{enumerate}
\item $\|B_\Gamma(A)\| \le \|A\|$ for all $A \in \alg$
\item $B_\Gamma(A) \in \alg_\Gamma$ for all $A \in \alg$
\item $B_\Gamma(A)=A$ for all $A \in \alg_\Gamma$
\item $B_\Gamma(A)=0$ for all $A \in \alg_{\Gamma^C}$
\end{enumerate}
\end{definition}
This is a fairly generous notion of boundary conditions, and in particular includes fixed and periodic boundary conditions.  We denote the Hamiltonian with boundary condition $B$ by
\begin{equation} \label{HDef}
H_{\Gamma,B}^{h,\zeta,\ul{\omega}} = \sum_X B_\Gamma(\inter_0(X)) + \sum_{x\in\partial_0 \Gamma} (h+\zeta_x)B_\Gamma(\lop_x) + \sum_{\alpha=1}^{N_\alpha} \sum_{x\in \partial_\alpha \Gamma} \omega_{\alpha x} B_\Gamma (\gamma_{\alpha x}).
\end{equation}
where $\partial_\alpha \Gamma$ denotes the set of $x \in \Zd$ for which $T_x A_\alpha$ contains members of both $\Gamma$ and $\Gamma^C$.

This allows us to define partition functions by
\begin{equation} \label{Zdef}
Z^h_{\Gamma,\bc}(\zeta,\ul{\omega}) := \Tr \exp (- \beta H_{\Gamma,\bc}^{h,\zeta,\ul{\omega}}),
\end{equation}
free energy by
\begin{equation} \label{Fdef}
F^h_{\Gamma,\bc}(\zeta,\ul{\omega}) := -\frac{1}{\beta} \log Z^h_{\Gamma,\bc}(\zeta,\ul{\omega})
\end{equation}
and Gibbs states by
\begin{equation} \label{stateDef}
\state{\cdot}^h_\Gamma (\zeta,\ul{\omega}) := \frac{\Tr \cdot e^{-\beta H_{\Gamma,\bc}^{h,\zeta,\ul{\omega}}}}{ Z^h_{\Gamma,\bc}(\zeta,\ul{\omega}) }.
\end{equation}
To avoid a profusion of subscripts, we omit a label for boundary conditions when periodic boundary conditions should be understood; and when an integer $L$ appears instead of the finite set $\Gamma$ it should be understood to represent the (hyper)cubic subset of $\Zd$ of side length $L$ approximately centered at the origin, i.e.\ \begin{equation}
\Gamma_L:=\left[-\frac{-L+1/2}{2},\frac{L+1/2}{2}\right]^d \cap \Zd.
\end{equation}

It is helpful to observe that the $\beta \to \infty$ limit of the free energy and Gibbs states (for the time being, we consider these limits with all other parameters fixed) exist.  Indeed, when the ground state is nondegenerate, the free energy converges to the ground state energy and the Gibbs state converges to the (unique) ground state.  Even in the presence of degeneracy, these limits provide an equally useful description of the system, and we can establish many results simultaneously for finite and zero temperature by taking advantage of this.  We will therefore take the free energy and Gibbs states to be defined for all $\beta \in [0,\infty]$, with the values at $\beta=\infty$ being the above limits.

The free energy, as defined in Equation~\eqref{Fdef}, has the following well-known property we will use repeatedly in what follows:
\begin{lemma}[\cite{Ruelle}]\label{Ruelle_lemma}
For any Hermitian matrices $A,B$ of the same size,
\begin{equation}
\left| \log \Tr e^A - \log \Tr e^B \right| \le \|A-B\|
\end{equation}
\end{lemma}

The terms of the interaction connecting a finite region to the rest of the system play an important role in the arguments of the present work.  We denote these by
\begin{equation} \label{Vdef}
%\begin{split}
V_L^{\zeta,\ul{\omega}} := \sum_{X: \overlappingClass{X}{\Gamma_L}} \inter_0(X) + \sum_{x\in \partial_0 \Gamma_L} (h + \zeta_x) \lop_x + \sum_{\alpha=1}^{N_\alpha} \sum_{x \in \partial_\alpha \Gamma_L} \omega_{\alpha x} \gamma_{\alpha x},
%\end{split}
\end{equation}
Note that $\| B_\Gamma (V^{\zeta,\ul{\omega}}_L) \| \le \| V^{\zeta,\ul{\omega}}_L \|$ for all boundary conditions, so bounds on the norm of the infinite-system operator above give considerable information about finite systems as well.

Our main result will be restricted to systems which are short range in the following sense:
\begin{assumption}\label{short_range_assumption}
There are constants $0 \le C, C' < \infty$ such that
\begin{equation}
\Av \left\| V_L^{\zeta,\ul{\omega}} \right\| \le C (1+|h|) L^{d-1} + C' L^{d/2}
\end{equation}
\end{assumption}
This may not be very transparent, so we note the  following results which provide sufficient conditions under which Assumption~\ref{short_range_assumption} is satisfied.
\begin{lemma}
If $\eta$ and $\ul{\upsilon}$ are i.i.d.\ and mutually independent with $N_\alpha$ finite, then there is a constant $0 \le c_1 < \infty$ such that
\begin{equation}
\Av \left\| \sum_{x:\overlappingClass{T_x A_0}{\Gamma_L}} (h + \eta_x) \lop_x + \sum_{\alpha=1}^{N_\alpha} \sum_{x:\overlappingClass{T_x A_\alpha}{\Gamma_L}} \upsilon_{\alpha x} \gamma_{\alpha x} \right\| \le c_1 L^{d-1}.
\end{equation}
\end{lemma}
\begin{proof}
The quantity whose norm is being bounded consists of $N+1$ sums, each with no more than $2d |A_\alpha| L^{d-1} $ terms, each bounded in norm by $1$ or $|h|$.
\end{proof}
When this holds, it means that Assumption~\ref{short_range_assumption} is satisfied iff the following condition on $\inter_0$ is satisfied:
\begin{equation} \label{inter0_boundary_bound}
\|V_L^{0,0}\| = \left\| \sum_{X: \overlappingClass{X}{\Gamma_L}} \inter_0(X) \right\| \le C L^{d-1} + C' L^{d/2}
\end{equation}
This is clearly the case when $\inter_0$ is of finite range, but also allows some scope for infinite range interactions.  A convenient condition~\cite{AW.CMP} is
\begin{lemma}\label{AW_LR_lemma}
If
\begin{equation}
\sum_{\substack{ X \ni 0 \\ \diam X \le L }} \diam X \frac{|\partial X|}{|X|} \|\inter_0(X)\| \le c' L^{(2-d)/2}
\end{equation}
for all $L$, then Inequality~\ref{inter0_boundary_bound} is true.
\end{lemma}
\begin{proof}
By the triangle inequality
\begin{equation}
\|V_L^{0,0}\| \le \sum_{X: \overlappingClass{X}{\Gamma_L}}  \left\| \inter_0(X) \right\|;
\end{equation}
and in this sum the terms with diameter $L$ or less contribute, at most,
\begin{equation}
\sum_{\substack{ X \ni 0 \\ \diam X \le L}} 2 d L^{d-1} \frac{\diam X}{|X|} \le 2 d c' L^{d/2},
\end{equation}
and the remaining portion is bounded by
\begin{equation}
\sum_{\substack{X \ni 0 \\ \diam X \le L}} L^d \frac{|\partial X |}{|X|} \| \inter_0 (X) \| \le L^{d-1} \sum_{\substack{X \ni 0 \\ \diam X \le L}} \diam X \frac{|\partial X|}{|X|} \| \inter_0 (X) \| \le c' L^{d/2},
\end{equation}
%\comment{Counting the sets involves bounding by approximately cubic case; might want to spell this out somewhere}
Putting the two parts back together we have Inequality~\ref{inter0_boundary_bound} with $C' = (2 d + 1) c'$.
\end{proof}

For pair interactions, the bound in Lemma~\ref{AW_LR_lemma} is satisfied in $d=1$ for interactions decaying like $(\textup{distance})^{-3/2}$ or faster; a result of Cassandro, Orlandi and Picco~\cite{Cassandro} shows that Proposition~\ref{main_prop} is false for a system with slightly longer range interactions, which suggests that Assumption~\ref{short_range_assumption} may in some sense be optimal.  This may be of some practical interest, since for pair interactions in $d=2$ we need the interactions to decay strictly faster than $(\textup{distance})^{-3}$ for Lemma~\ref{AW_LR_lemma} to apply, and inverse cube interactions seem to be quite common~\cite{SchSt.JPCM}.

Finally, we give a similar statement which provides some control over the case of infinite-range random interactions:
\begin{lemma}
Let $\ul{\upsilon}$ be i.i.d., with $N_\alpha=\infty$ and
\begin{equation}
\sum_{\substack{\alpha \ge 1 \\ \diam A_\alpha \le L}} \diam A_\alpha |\partial A_\alpha| \Av | \upsilon_{\alpha,0} | \le c L^{(2-d)/2}.
\end{equation}
Then
\begin{equation}
\Av \left\| \sum_{\alpha=1}^\infty \sum_{x:\overlappingClass{T_x A_\alpha}{\Gamma_L}} \upsilon_{\alpha x} \gamma_{\alpha x} \right\| \le c' L^{d/2}.
\end{equation}
\end{lemma}
\begin{proof}
The contribution of terms with $\diam A_\alpha \le L$ is bounded by %\comment{Uses $|A| \le |\partial A| \diam A$; I think this is true for all finite $A \subset \Zd$, but should make sure}
\begin{equation}
\begin{split}
& \sum_{\substack{ \alpha \ge 1 \\ \diam A_\alpha \le L}} 2d (L+ \diam A_\alpha)^{d-1} \diam A_\alpha \Av |\upsilon_{\alpha 0} | \\
& \; \le \sum_{\substack{ \alpha \ge 1 \\ \diam A_\alpha \le L}} d 2^d L^{d-1} \diam A_\alpha |\partial A_\alpha| \Av |\upsilon_{\alpha 0} | \le d 2^d c L^{d/2},
\end{split}
\end{equation}
while the remaining terms are bounded by
\begin{equation}
\sum_{\substack{ \alpha \ge 1 \\ \diam A_\alpha > L}} L^d |\partial A_\alpha| \Av |\upsilon_{\alpha 0} |
\le L^{d-1} \sum_{\substack{ \alpha \ge 1 \\ \diam A_\alpha > L}} \diam A_\alpha |\partial A_\alpha| \Av |\upsilon_{\alpha 0} |
\le c L^{d/2},
\end{equation}
and the conclusion follows with $c' = (d 2^d +1)c$.
\end{proof}

\subsection{Systems with continuous symmetries} \label{continuous_def}

Imry and Ma's initial work~\cite{YM} mainly concerned systems with continuous symmetries.  In this context the Mermin-Wagner theorem~\cite{MerminWagner66,Mermin67} already precludes long range order without randomness in two dimensions, so the rounding effect would be of little consequence except that it extends to four dimensions, but only so long as the randomness preserves the symmetry ``on average'' in a sense the following passage should make clear.

First, we assume that the single-site algebra $\alg_0$ contains a subalgebra isomorphic to the rotations $SO(N)$ for some $N \ge 2$.  For each rotation $R \in SO(N)$, let $R_x$ be the corresponding element of $\alg_x$.  We will say that an interaction $\inter$ is invariant iff
\begin{equation}
\inter(X) = \left( \prod_{x \in X} R^{-1}_x \right) \inter(X) \left( \prod_{x \in X} R_x \right)
\end{equation}
for all $X \in \finiteDom$ and all $R \in SO(N)$.

Intuitively, for a random system to be (stochastically) invariant under rotations, the field and the quantity it couples to should both transform as dual representations of $SO(N)$.  The vector representation is the only case we are aware of which includes any cases of intrinsic interest (this case, in particular, includes Heisenberg models in a random magnetic field), so we will focus on this.  The fields are then elements of $\Fields^N$, or equivalently maps $\vec{\zeta}:\Zd \to \R^N$, and we let
\begin{definition}
A random vector field is a collection of $\R^N$-valued random variables indexed by the elements of $\Zd$.
\end{definition}

A random vector field is i.i.d.\ iff these random variables are independent and identically distributed.  The components of a random vector field in a particular direction are a random field in the sense of Definition~\ref{RandomField.def}, a fact which we will use frequently.  We will say that a random vector field satisfies the Lyapunov condition if each of its components does in the sense of Definition~\ref{Lyapunov.def}.

A random vector field $\vec{\eta}$ is {\emph isotropically distributed} iff for each $x \in \Zd$ and $R \in SO(N)$ the distribution of $\vec{\eta}_x$ is the same as the distribution of $R \vec{\eta}_x$.  Among other things, this implies that the component $\hat{e} \cdot \vec{\eta}_x$ in an arbitrary direction will have an absolutely continuous distribution so long as $\vec{\eta} \ne 0$ with probability one, and will have no isolated point masses (see the statement of Proposition~\ref{main_prop} below) provided that $\vec{\eta} \ne 0$ with nonzero probability.

We then define systems by the quenched local Hamiltonians
\begin{equation} \label{vecHam}
H_{\Gamma}^{h,\vec{\zeta},\ul{\vec{\omega}}} = \sum_X B_\Gamma(\inter_0(X)) + \sum_{x \in \Gamma} (\vec{h}+\vec{\zeta}_x) \cdot B_\Gamma(\vec{\lop}_x) %+ \sum_{\alpha=1}^{N_\alpha} \sum_{x \in \Gamma} \omega_{\alpha x} \cdot P_\Gamma (\vec{\gamma}_{\alpha x})
,
\end{equation}
where each $\vec{\lop}_x$ %and $\vec{\gamma}_{\alpha x}$ are, for each $\alpha$ and $x$,
is a vector operator, that is a collection of $N$ operators satisfying
\begin{equation}\label{vector_operator_transformation}
R \vec{\lop}_x = %\left( \prod_{x \in T_x A_0} R_x \right) \vec{\lop}_x,
R_x^{-1} \vec{\lop}_x R_x,
\end{equation}
and we also assume that the components of $\vec \lop _x$ are in $\alg_x$. %\comment{might be nice to weaken this, but not obviously easy.}
Other local Hamiltonians, free energies, etc. are defined in the same terms.  Then if $\vec{\eta}$ is isotropically distributed and $\inter_0$ is invariant, we will say that the system described by $H_{\Gamma}^{h,\vec{\eta}} $ is isotropic.

We will have need of a restriction on long range interactions similar to Assumption~\ref{short_range_assumption} to extract additional results for these systems.  The assumption (employed in the proof of Lemma~\ref{little_g_bound_lemma}) is as follows:
\begin{assumption}\label{continuous_short_range_assum}
The sum
\begin{equation}
\sum_{X \ni 0} (\diam X)^2 |X| \|\inter_0(X)\|
\end{equation}
is finite.
\end{assumption}
For pair interactions, this reduces to the statement
\begin{equation}
\sum_{x \in \Zd} \|\inter_0 (\{0,x\}) \|x\|_\infty^2 < \infty
\end{equation}
found (in slightly different notation) in~\cite{QIMLetter}.

\section{Thermodynamic limit and notions of long range order}
%Sketch proof of thermodynamic limit and relationship of free energy to macroscopic quenched variables; explain notions of (long and short) long range order and prove relationship to differentiability of free energy

The first requirement in talking rigorously about the thermodynamics of an infinite system is to prove the existence of the thermodynamic limit of some basic quantity.  For lattice systems one conventionally uses either the pressure (as in~\cite{Simon,Ruelle}) or the free energy density (as in~\cite{AW.CMP,Bovier}) - they are related by $P = -\beta f$, so for most purposes they are interchangeable.  We will employ the free energy density, since it has the considerable advantage of having a well-defined behavior at $\beta=\infty$ (zero temperature) where in the absence of residual entropy it coincides with the ground state energy density.

We define the free energy density for a finite system in the more or less obvious manner, as
\begin{equation}
f^h_{\Gamma,\bc}(\zeta,\ul{\omega}) := \frac{F^h_{\Gamma,\bc}(\zeta,\ul{\omega})}{|\Gamma|},
\end{equation}
where $|\Gamma|$ is the number of points in $\Gamma$.  As the notation suggests, this depends on the choice of boundary conditions and of the disorder variables.  In the thermodynamic limit, however, the dependence on boundary conditions disappears and the dependence on the disorder variables becomes trivial, as the following theorem will show.  Essentially the same statement was first proven by Vuillermot in 1977~\cite{vuillermot1977tqr}; the version given here is more suited to the present work.

\begin{theorem}[\cite{AW.CMP}]\label{Therm.Limit.Thm}
Let Assumption~\ref{short_range_assumption} be satisfied.  For any $h$, any i.i.d.\ random fields $\eta,\ul{\upsilon}$ with finite variance, any $\beta \in [0,\infty]$, there is a set $\mathcal{N} \in \Fields^{N+1}$ such that $\Prob{(\eta,\ul{\upsilon}) \in \mathcal{N}}=1$ so that the limit
\begin{equation}
\F(\beta,h) := \lim_{L \to \infty} f^h_{\Gamma_L,\bc}(\zeta,\ul{\omega})
\end{equation}
exists for all $(\zeta,\ul{\omega}) \in \mathcal{N}$, $h \in \R$, and all $\bc$, and is independent of $\zeta$, $\ul{\omega}$, and $\bc$.

Furthermore,
\begin{equation}
\lim_{L \to \infty} \frac{\| V^{\zeta,\ul{\omega}}_L \|}{L^d} = 0
\end{equation}
for all $(\zeta,\ul{\omega}) \in \mathcal{N}$.
\end{theorem}
%\begin{proof}
%Since $\zeta$ does not play a special role here, we denote $\zeta_x$ by $\omega_{0,x}$ for brevity.
%
%The proof is easier for bounded random fields.  We reduce to this case by examining the effect of removing the fields exceeding a certain strength.  Let $\phi^{h,K}_{\Gamma_L,B}(\zeta,\ul{\omega})$ be the value of $f$ with all fields of absolute value larger than $K$ replaced with 0.  Then
%\begin{equation}
%\left| \phi^{h,K}_{\Gamma_L,B}(\zeta,\ul{\omega}) - f^{h}_{\Gamma_L,B}(\zeta,\ul{\omega}) \right|
%\le \frac{1}{\Gamma_L} \sum_\alpha=0^{N_\alpha} \sum_{x \in \Gamma_L^0} |\omega_{\alpha,x} | I\left[ |\omega_{\alpha,x} | \ge K \right] =: d_{K,L}(\ul \omega)
%\end{equation}
%where $\Gamma^\alpha_L$ is the set of sites $x$ for which $T_x A_\alpha$ intersects $\Gamma_L$, which includes all the fields affecting $f$.  Then the assumption of independence and finite variance lets us use the law of large numbers to note that there is ???
%
%
%\end{proof}
This theorem was stated for classical systems, but the proof depends only on some properties of $f$ - in particular Lemma~\ref{Ruelle_lemma} - which also hold for quantum systems.

The fact that the limiting free energy is almost certainly independent of the random field provides the following conclusion:
\begin{corollary}[Brout's prescription\cite{Brout,vuillermot1977tqr}]
  \begin{equation}
    \lim_{L \to \infty} \Av f^h_{\Gamma_L,\bc}(\zeta,\ul{\omega}) = \Av \lim_{L \to \infty} f^h_{\Gamma_L,\bc}(\zeta,\ul{\omega}).
  \end{equation}
\end{corollary}
In other words, one can take the average over the randomness before or after the thermodynamic limit without changing the free energy.

Since $\F$ is a limit of convex functions, the following useful fact (also noted in~\cite{AW.CMP}) follows immediately from Theorem~\ref{Therm.Limit.Thm}:
\begin{corollary}
$\F(\beta,h)$ is convex as a function of $\beta$ and concave as a function of $h$.
\end{corollary}

This allows us to prove some handy results which extend the relationship between the derivatives of the free energy to expectation values of certain observables from finite to infinite systems.  To begin with, note that
\begin{equation}
\frac{\partial f^h_{\Gamma,\bc}(\zeta,\ul{\omega})}{\partial h} = \frac{1}{|\Gamma|} \sum_{x \in \Gamma} \state{\lop_x}^h_{\Gamma,\bc}(\zeta,\ul{\omega}).
\end{equation}
The convexity of $\F$ does not imply that the above derivative always converges in the thermodynamic limit, but it does imply something almost as good:
\begin{corollary} \label{LongLongRange}
\begin{equation}
\LIM_{L \to \infty} \frac{1}{|\Gamma_L|} \sum_{x \in \Gamma_L} \state{\lop_x}^h_{\Gamma_L,\bc}(\zeta,\ul{\omega})
\in \left[ \frac{\partial \F}{\partial h-}, \frac{\partial \F}{\partial h+} \right],
\end{equation}
where $\LIM$ denotes the set of accumulation points, and $\tfrac{\partial}{\partial h \pm}$ denote directional derivatives with respect to $h$.
\end{corollary}

The above statement is about the average of $\state{\lop}$ over the whole system, or in other words it is a statement about ``long long range order''.  It is also possible to make a similar statement relating to ``short long range order'':
\begin{theorem}\label{ShortLongRange}
\begin{equation}
\LIM_{L \to \infty} \LIM_{M \to \infty} \frac{1}{|\Gamma_L|} \sum_{x \in \Gamma_L} \state{\lop_x}^h_{\Gamma_M,\bc}(\zeta,\ul{\omega})
\in \left[ \frac{\partial \F}{\partial h-}, \frac{\partial \F}{\partial h+} \right].
\end{equation}
\end{theorem}
\begin{proof}
Let $F^{h,\delta,\Lambda}_{\Gamma,B}$ denote the free energy with the fixed field within $\Lambda$ changed by $\delta$, so that
\begin{equation}
\frac{1}{|\Gamma_L|} \sum_{x \in \Gamma_L} \state{\lop_x}^h_{\Gamma_M,\bc}(\zeta,\ul{\omega}) = \left. \frac{1}{|\Gamma_L|}\frac{\partial F^{h,\delta,\Lambda}_{\Gamma_M,B}}{\partial \delta} \right|_{\delta=0}
\end{equation}
Now from Lemma~\ref{Ruelle_lemma} we see that
\begin{equation}
\frac{1}{|\Gamma_L|} \left( F^{h,\delta,\Gamma_L}_{\Gamma_M,B} -  F^{h,\delta,\Gamma_L}_{\Gamma_M,B} \right)
= f^{h+\delta}_{\Gamma_L,0} - f^{h}_{\Gamma_L,0} + O \left( \frac{\|V^{\zeta,\ul{\omega}}\|}{L^d} \right)
\end{equation}
uniformly in $M$.  Then for $(\zeta,\ul{\omega}) \in \mathcal{N}$, this implies that
\begin{equation}
\lim_{L \to \infty} \frac{1}{|\Gamma_L|} \left( F^{h,\delta,\Gamma_L}_{\Gamma_M,B} -  F^{h,\delta,\Gamma_L}_{\Gamma_M,B} \right)
=\F(\beta,h+\delta)-\F(\beta,h)
\end{equation}
and the conclusion follows by standard convexity arguments.\qed
\end{proof}
%\comment{Still need that last SLRO statement}

Choosing a positive sequence $\delta_i \to 0$ such that $\F$ is differentiable at all $h\pm \delta_i$, we have also
\begin{equation}
\lim_{i \to \infty} \lim_{L \to \infty} \lim_{M \to \infty} \frac{1}{|\Gamma_L|} \sum_{x \in \Gamma_L} \state{\lop_x}^{h\pm \delta_i}_{\Gamma_L,\bc}(\zeta,\ul{\omega})
=\lim_{i \to \infty} \left. \frac{\partial \F}{\partial h} \right|_{h\pm \delta_i} = \frac{\partial \F}{\partial h \pm},
\end{equation}
which together with the individual ergodic theorem (applicable since the random fields are i.i.d) this implies
\begin{corollary}\label{mag_ergodicity}
\begin{equation}
\lim_{i \to \infty} \lim_{L \to \infty} \Av \state{\lop_x}^{h\pm \delta_i}_{\Gamma_L,\bc}(\zeta,\ul{\omega})
= \frac{\partial \F}{\partial h \pm},
\end{equation}
\end{corollary}

\section{Statement of main results}
%State Theorems 1 and 2 of PRL version; outline of the method of proof (conflicting upper and lower bounds).  Note that we will go through the proof first for a single random term (possibly vectorial) and then show that the adaptation to the general case is just a matter of averaging everything over the other random fields.

The first main result of the following chapters is:
\begin{proposition} \label{main_prop}
In dimensions $d \le 2$, any system of the type described in in Section~\ref{defs_section}, with $\eta$ an i.i.d.\ Lyapunov random field and $\ul{\gamma}$ i.i.d, has $\F$ differentiable in $h$ for all $h$, provided any of the following hold:
\begin{itemize}
\item The system satisfies the weak FKG property with respect to $\lop$, $\beta < \infty$, and the distribution of $\eta$ is nontrivial
\item $\beta < \infty$, and the distribution of $\eta_0$ has no isolated point masses (i.e.\ there are no real numbers $x$ and $\delta>0$ such that $\Prob{|\eta_0-x| \le \delta} = \Prob{\eta_0=x}>0$)
\item The distribution of $\eta_0$ is absolutely continuous with respect to the Lebesgue measure
\end{itemize}
\end{proposition}

We note that the phrasing of the result relates to the way we have arranged the Hamiltonians of the systems under consideration, so that the result has something to say only when the random field can be expressed as part of the source field %\comment{Source field is a CMT term for conjugate field}
for the order parameter, in other words when the randomness is direct in the sense used on p.~\pageref{Direct.label} above.

%The proof of this proposition occupies most of the next several sections, culminating in section~\ref{wrapup}.  The basic strategy is the same as that of Aizenman and Wehr~\cite{AW.CMP}, which was inspired in turn by Imry and Ma~\cite{YM}. It is to identify a quantity $G_L$, which should represents the effect effect of the random field $\eta$ in the region $\Gamma_L$ on the free energy difference among macroscopically distinct states, and which asymptotically satisfies
%\begin{equation}\label{loose_CLT_statement}
%G_L(\eta) \to L^{d/2} N(0,b),
%\end{equation}
%with $b>0$ when $m_+ \ne m_-$, and
%\begin{equation}
%|G_L(\eta)| \le L^{d-1}.
%\end{equation}
%When $b>0$ and $d \le 2$, these two conditions are incompatible.  This is obvious for $d=1$, where~\eqref{loose_CLT_statement} implies $G_L(\eta) ~ b L^{1/2}$; for $d=2$ it is not so obvious; we will provide a proof, but it should be clear that the statement that $G_L(\eta)$ is asymptotically normally distributed means that it will occasionally exhibit arbitrarily large fluctuations on the scale of $b L$, and so cannot be strictly bounded on that scale.

We also establish
\begin{proposition} \label{continuous_prop}
In dimensions $d \le 4$, any isotropic system of the type described in Section~\ref{continuous_def} satisfying Assumption~\ref{continuous_short_range_assum} has $\nabla_{\vec h} \F$ continuous at $0$, provided the distribution of $\vec{\eta}$ is isotropic and one of the following holds:
\begin{itemize}
 \item $|\vec{\eta}_0|>0$ with probability 1, or
 \item $\beta < \infty$, and the distribution of %$\hat{e} \cdot \vec{\eta}_0$ has no isolated point masses
     $\vec{\eta}$ is not concentrated at a single point.
 \end{itemize}
\end{proposition}
We note that an apparently weaker condition on the distribution of $\vec{\eta}$ is adequate because the what will ultimately be important is the distribution of a particular component.  With the assumption of an isotropic distribution for the vector, the components satisfy the stronger conditions used in Proposition~\ref{main_prop} or Theorem~\ref{AW_theta}, as discussed above.
%by a similar method, culminating in Section~\ref{continuous_wrapup}, with an improved bound of the form
%\begin{equation}
%|G_L(\eta)| \le L^{d-2}.
%\end{equation}

\chapter{A nonlinear central limit theorem}\label{CLT_section}
\section{Background}
The ``classical'' central limit theorem, long one of the central elements of probability theory, states that a sum of independent random variables with finite variance converges, in distribution and on an appropriate rescaling, to a normally distributed random variable.  There are a number of generalizations, the best-known due to Lindeberg~\cite{Lindeberg}, which generalize this notion by replacing the i.i.d assumption with a weaker assumption, including the possibility that the distribution of the variables, as well as their cardinality, changes as the limit is taken.

We will present a clarified version of a result due to Aizenman and Wehr~\cite{AW.CMP} which builds on results of that kind to replace the customary sum with a member of a much larger class of functions, which however have certain properties (a partial symmetry with respect to permutations of arguments, and a fairly strong continuity) in common with it.  This exposition also incorporates a necessary correction pointed out by Bovier~\cite{Bovier}.

We should note that the statement that a certain sequence of random variables, described as a related collection of Lipschitz continuous functions of a family of independent random variables, converges in distribution to a normal random variable, is closely related to the concentration of measure phenomenon~\cite{Ledoux2001concentration}.  Among its many other facets, this involves upper bounds on the probability with which certain classes of random variables described as functions of a family of $N$ independent random variables.  A central limit theorem involves an estimate of a similar form.  The result we will discuss is stronger than a concentration estimate in that it provides a lower bound as well as an upper bound (the latter would not be useful for our main result); however it is purely asymptotic, whereas concentration of measure techniques provide speed of convergence information as well.  The assumptions are in some ways stronger and in some ways weaker than those involved in concentration of measure:
\begin{enumerate}
\item The result below uses a form of Lipschitz continuity based on the $\ell_1$ norm, which is stronger than the $\ell_2$ notion used in concentration of measure (see below) but more suited to functions of infinitely many variables.
\item The result below assumes translation covariance (a weak form of exchangability), but no assumption is made on its level sets.  Gaussian random variables do not play a distinguished role.
\end{enumerate}

\section{Definitions}

As well as the notions related to random fields introduced in the previous chapter, we will make use of the following:
\begin{definition}\label{TranslationCovariantDef}
A function $\tau:\Zd \times \Fields \to \R$ (equivalently, a collection of functions of fields indexed by elements of \Zd) is translation covariant if $\tau_x(\eta) = \tau_{x-y}(T_y \eta)$ for all $x,y,\eta$, where $T_y$ denotes translation.
\end{definition}

We will have occasion to frequently use the $\ell^1$ norm on \Fields,
\begin{equation}
\pnorm{\zeta}{1} := \sum_{x \in \Zd} |\zeta_x|;
\end{equation}
in particular this defines a Lipschitz seminorm on functions $f:\Fields \to \R$ by
\begin{equation} \label{lip1def}
\lip{f} :=
\sup_{\substack{\zeta,\zeta' \in \Fields \\ 0<\pnorm{\zeta-\zeta'}{1} < \infty}} \frac{|f(\zeta)-f(\zeta')|}{\pnorm{\zeta-\zeta'}{1}}.
\end{equation}
It is worth spending a moment on the comparison of this norm with the similar quantity based on the $\ell^2$ norm which appears frequently in the concentration of measure literature.  The $\ell^2$ norm in this context is defined by
\begin{equation}
\pnorm{\zeta}{2} := \left( \sum_{x \in \Zd} \zeta_x^2 \right)^{1/2},
\end{equation}
%satisfies
%\begin{equation}
%\pnorm{\zeta}{1} \ge \pnorm{\zeta}{2}
%\end{equation}
%for all $\zeta \in \Fields$ by the triangle inequality, and therefore the related Lipschitz seminorm satisfies
%\begin{equation}
%\lip{f}_2 := \sup_{\substack{\zeta,\zeta' \in \Fields \\ 0<\pnorm{\zeta-\zeta'}{2} < \infty}} \frac{|f(\zeta)-f(\zeta')|}{\pnorm{\zeta-\zeta'}{2}} \ge \lip{f}.
%\end{equation}
%Consequently any uniform bound on $\lip{\cdot}_2$ implies a similar bound on $\lip{\cdot}$.
and the related Lipschitz seminorm by
\begin{equation} \label{lip2def}
  \lip{f}_2 := \sup_{\substack{\zeta,\zeta' \in \Fields \\ 0<\pnorm{\zeta-\zeta'}{2} < \infty}} \frac{|f(\zeta)-f(\zeta')|}{\pnorm{\zeta-\zeta'}{2}}
\end{equation}
In concentration of measure one is concerned with functions of $N$ variables, for which the supremum in the above expression is attained with $\zeta$ and $\zeta'$ differing only in the corresponding $N$ elements, whence
\begin{equation}
 \pnorm{\zeta-\zeta'}{1} \le \sqrt{N} \pnorm{\zeta-\zeta'}{2}
\end{equation}
by Young's inequality.  Using this to compare Equations~\eqref{lip1def} and~\eqref{lip2def}, we see that
\begin{equation}
  \lip{f_N}_2 \le \sqrt{N} \lip{f_N},
\end{equation}
so that if we have $\lip{f_N} \le 1$ (as below), this implies $\lip{f_N}_2 \le \sqrt{N}$.

\section{The proposition}

\begin{proposition}\label{CLT_prop}
Let $\eta$ be an independent, Lyapunov random field, and let $G_L:\Fields \to \R$ be a family of functions indexed by $L \in \mathbb{N}$, each with the following properties:
\begin{enumerate}
\item $G_L$ depends only on the values of the field for sites in $\Gamma_L$ \label{local_hypothesis}
\item $\lip{G_L} \le 1$ \label{uLif}
\item $\Av G_L(\eta) = 0$
%\item $G_1$ has a distributional derivative $G_1'$ with $\Av G_1'= \epsilon M$ for some $M \ge 0$\footnote{Note that $G_1'$ is a function of only one variable.  The existence of a distributional derivative, i.e.\ a Lebesgue-integrable function satisfying $\int_y^z G_1'(x) dx = G_1(y)-G_1(x)$ is guaranteed by item~\ref{uLif}, which also implies that $\|G_1'\|_\infty \le \epsilon$ and therefore also that $\Av G_1'$ exists.  It may be, however, that this derivative is not unique, and when the distribution of $\eta_0$ is not absolutely continuous with respect to the Lebesgue measure and $B=\infty$ it is possible that this could allow more than one valid choice of $M$, although this is immaterial for the application we have in mind.} \label{M_hypothesis}
\item $\condAv{G_L(\eta)}{\eta_\Lambda=\zeta_\Lambda} = G_{L'}\circ T_x$ whenever $T_{-x} \Gamma_{L'} = \Lambda \subset \Gamma_L$ \label{consistency_assumption}
\end{enumerate}
Then
\begin{equation}
G_L(\eta)/L^{d/2} \to N(0,b^2)
\end{equation}
in distribution as $L \to \infty$, for some $b$ satisfying
\begin{equation} \label{bBounds}
\Av G_1^2 \le b^2 \le 2 (\Av |\eta_0|)^2.
\end{equation}
\end{proposition}

In order to use this result, we will need to establish some control over the conditions under which $\Av{G_1^2} > 0$.  To do this, we employ the following theorem (proven in Appendix III of~\cite{AW.CMP}):

\begin{theorem} \label{AW_theta}
Let $\nu$ be a Borel probability measure on $\R$, and
\begin{equation}
\mathcal{V}_{1,\beta}:= \left\{ \begin{array}{ll}
\left\{ g \in C^1(\mathbb{R}) \middle| \lip{g} \le 1, \lip{g'} \le \beta \right\}, & \beta < \infty \\
\left\{ g \in C(\R) \middle| \lip{g} \le 1 \right\}, & \beta = \infty
\end{array} \right. ,
\end{equation}
and also
\begin{gather}
\theta_\nu(M,\beta) = \inf \left\{ \left[ \int g(x)^2 \nu(dx) \right]^{1/2} \middle| g \in \mathcal{V}_{1,\beta} , \int g'(x) \nu(dx) = M \right\} \\
\gamma_\nu(M,\beta) = \inf \left\{ \left[ \int g(x)^2 \nu(dx) \right]^{1/2} \middle| g \in \mathcal{V}_{1,\beta} , g'(\cdot) \ge 0,  \int g'(x) \nu(dx) = M \right\}.
\end{gather}
Then:
\begin{enumerate}
\item $\theta_\nu(0,\beta) \equiv \gamma_\nu(0,\beta) \equiv 0$
\item $\theta_\nu(M,0)$ is nonzero for all $M>0$ iff $\nu$ is absolutely continuous with respect to the Lebesgue measure
\item For finite $\beta$, $\theta_\nu(M,\beta)$ is nonzero for all $M>0$ iff $\nu$ has no isolated point masses
\item For finite $\beta$, $\gamma_\nu(M,\beta)$ is nonzero for all $M>0$ iff $\nu$ is not concentrated at a single point
\end{enumerate}
\end{theorem}
To employ this, we note that $G_1 \in \mathcal{V}_{1,B}$ for
\begin{equation}
B := \lip{\frac{\partial G_1}{\partial\zeta_0}}
\end{equation}
if the derivative on the right hand side exists everywhere, and $B=\infty$ otherwise.  Then when $G_1$ has a distributional derivative $G_1'$ with $\Av G_1'= M$ for some $M \ge 0$\footnote{Note that $G_1'$ is a function of only one variable.  The existence of a distributional derivative, i.e.\ a Lebesgue-integrable function satisfying $\int_y^z G_1'(x) dx = G_1(y)-G_1(x)$ is guaranteed by item~\ref{uLif}, which also implies that $\|G_1'\|_\infty \le 1$ and therefore also that $\Av G_1'$ exists.  It may be, however, that this derivative is not unique, and when the distribution of $\eta_0$ is not absolutely continuous with respect to the Lebesgue measure and $B=\infty$ it is possible that this could allow more than one valid choice of $M$, although this is immaterial for the application we have in mind.} we have $\Av G_1^2 \ge \theta_\nu^2 (M,B)$ (for monotone $G_1$, $\Av G_1^2 \ge \gamma_\nu^2 (M,B)$).

We will not provide a proof of Theorem~\ref{AW_theta}, but since it is the source of a perplexing limitation in our result (as in the classical case) some commentary seems to be warranted, and it is possible to provide some insight into the situation and its prospects.  This will be done in Section~\ref{restrictions_sec} below.

Before going on to the proof, we should clarify the relationship to the formulation of the corresponding result, Proposition 6.1 of \cite{AW.CMP}.  Much of the difference is due to the fact that I have separated the main result from the positivity criteria embodied in Theorem~\ref{AW_theta}, but there is a remaining difference in language in which the results are framed, as the following lemma, which also brings the abstract objects of Proposition~\ref{CLT_prop} into a form more closely related to their use in a thermodynamic context, should clarify:

\begin{lemma} \label{tau_into_CLT}
Let $\eta$ be an independent Lyapunov random field, and let $G_L:\Fields \to \R$ be a family of functions indexed by $L \in \mathbb{N}$, and $\tau_x:\Fields \to \R$ a translation covariant family of functions satisfying
\begin{enumerate}
\item
\begin{equation}
\frac{\partial G_L (\zeta)}{\partial \zeta_x} = \piecewise{  \condAv{\tau_x(\eta)}{\eta_{\Gamma_L}=\zeta_{\Gamma_L}}, & x \in \Gamma_L \\ 0, & x \notin \Gamma_L}
\end{equation}
\item $\Av \tau_x = M$ for some $M \ge 0$ \label{tau_M_assumption}
\item $|\tau_x(\zeta)|\le 1$ and $\left| \frac{\partial \tau_x}{\partial \zeta_x}\right|\le B'$ for all $x, \zeta$
\item $\Av G_L(\eta) = 0$
\end{enumerate}

Then $\eta$ and $G_L$ satisfy the hypotheses of Proposition~\ref{CLT_prop} (with the same $M$).  Furthermore:
\begin{enumerate}
\item If $\tau_x$ is nonnegative, then $G_1$ is nondecreasing.
\item $G_1 \in \mathcal{V}_{1,B'}$
\item $G_1$ has a distributional derivative $G_1'$ with $\Av G_1'= M$
\end{enumerate}
\end{lemma}
\begin{proof}
The first three conditions in Proposition~\ref{CLT_prop} are trivially satisfied since $\lip{G_L} \le \sup_\zeta |\tau_x(\zeta)| \le $ and $G_1' = \condAv{\tau_0}{\eta_0}$.  The enumerated properties of $G_1$ are equally trivial.

For the last point, we note that $\condAv{G_L}{\eta_\Lambda}$ and $G_{L'} \circ T_x$ have derivatives given by identical expressions in terms of $\tau_x$, and so can only differ by a constant; but both have zero mean, and so that constant must be zero.
\end{proof}

\section{Proof of Proposition~\ref{CLT_prop}}

Let $G_L$ and $\eta$ satisfy the conditions of Proposition~\ref{CLT_prop}.  Order the elements of $\Zd $ lexicographically, and let $F(L,k)$ be the set consisting of the first $k$ elements of $\Gamma_L$ (of course $0 \le k \le L^d$).  Then we define
\begin{equation}
Y_{L,k} := \cFLk{G_L}{k}-\cFLk{G_L}{k-1} \label{Ydef}
\end{equation}
so that
\begin{equation}
G_L(\zeta) = \sum_{k=1}^{L^d} Y_{L,k}(\zeta)
\end{equation}
for all $\zeta \in \Fields$.  The definition~\ref{Ydef} of $Y$ makes it a martingale array and we will ultimately obtain a proof by showing that it satisfies the conditions of an existing central limit theorem for such objects~\cite{HallHeyde}.  As is the case with other central limit theorems for non-i.i.d.\ arrays, the conditions of this theorem are basically the existence of a limit of the average variance (Lemma~\ref{YtobConvLemma} below) and the vanishing of fluctuations on a larger scale (Lemma~\ref{NoLargeFluct} below).

The following result tells us that fluctuations in $Y_{L,k}(\eta)$ (a random variable) are basically no worse than those of the field at a single site.  From here on, we will let $x_k$ denote the $k$th element in $\Gamma_L$, when the value of $L$ is clear from the context, and $\zeta_k = \zeta_{x_k}$ etc.
\begin{lemma}\label{YBoundLemma}
For all $\zeta \in \Fields$,
\begin{equation}
|Y_{L,k}(\zeta)|\le (|\zeta_{k}|+\Av |\eta_0|).
\end{equation}
\end{lemma}
\begin{proof}
We can write
\begin{equation}
Y_{L,k}(\zeta)=\Av \left( G_L (\eta_{\Gamma_L^C},\zeta_{1,\ldots,k},\eta_{k+1,\ldots,L^d}) -  G_L (\eta_{\Gamma_L^C},\zeta_{1,\ldots,k},\eta_{k+1,\ldots,L^d}) \right);
\end{equation}
then the assumption that $\lip{G_L} \leq 1$ means that the quantity being averaged above has absolute value no more than $|\zeta_k-\eta_k|$.  Thus
\begin{equation}
|Y_{L,k}(\zeta)|\le \Av |\zeta_k-\eta_k| \le (|\zeta_{k}|+\Av |\eta_0|).
\end{equation}
\end{proof}

This is a uniform bound in absolute value by a square-integrable function (since the Lyapunov condition implies in particular that $\eta_x$ has finite variance), and will allow us to apply a number of general convergence theorems.  In particular it makes any collection of the functions $Y_{L,k}$ uniformly integrable, and we will take advantage of this to show that the asymptotics of $Y$ are described by a translation covariant (from another perspective, stationary or exchangeable) object $W$.

For $x \in \Gamma_L$, let $Y_{L,x}$ denote $Y_{L,k}$ for $k$ such that $x_k=x$, and let $\mathcal{F}_L=\mathcal{F}_{L,L^d}$.  Then it is evident from the consistency condition on $G_L$ in Proposition~\ref{CLT_prop} that
\begin{equation}
Y_{L,x}=\condAv{Y_{L',x}}{\eta_{\Gamma_L}=\zeta_{\Gamma_L}},
\end{equation}
which is to say that for fixed $x$, the sequence $Y_{L,x}(\eta)$ forms a martingale with respect to $\mathcal{F}$, and applying the uniformly integrable martingale convergence theorem we have
\begin{corollary} \label{W_exists}
For each $x \in \Zd$, the $\mathcal{L}^1$ limit $W_x = \lim_{L \to \infty} Y_{L,x}$ exists with
\begin{equation}
Y_{L,x} = \condAv{W_x}{\eta_{\Gamma_L}=\zeta_{\Gamma_L}}
\end{equation}
whenever $\Gamma_L \ni x$.
\end{corollary}

\begin{lemma}
$W_x$ form a translation-covariant family.
\end{lemma}
\begin{proof}
Recalling the definition of translation covariance (\ref{TranslationCovariantDef}), we examine
\begin{equation}
\begin{split}
W_x(T_y \zeta) = &\limOne_{L \to \infty} Y_{L,x}(T_y\zeta)
\\=&  \limOne_{L \to \infty} \left( \condAv{G_{L+\pnorm{y}{\infty}}\circ T_y}{\etaCond{T_{-y}F(L,k)}}\right.\\ &- \left. \condAv{G_{L+\pnorm{y}{\infty}}\circ T_y}{\etaCond{T_{-y}F(L,k-1)}} \right),
\end{split}
\end{equation}
where we have obtained the right-hand side by writing out the definition of $Y_{L,k}$ and some elementary properties of the conditional expectation.  We then rewrite the right hand side again using the consistency assumption on $G_L$, Assumption~\ref{consistency_assumption} of Proposition~\ref{CLT_prop}, to change coordinates, and obtain
\begin{equation}
\begin{split}
W_x(T_y \zeta)=& \limOne_{L \to \infty}  \condAv{Y_{L+\pnorm{y}{\infty},x-y}}{\etaCond{T_{-y}\Gamma_L}}
\\=& \limOne_{L \to \infty}  \condAv{W_{x-y}}{\etaCond{T_{-y}\Gamma_L}}
=W_{x-y}(\zeta)
\end{split}
\end{equation}
\end{proof}

We are now ready to prove
\begin{lemma}\label{YtobConvLemma}
Let $b^2=\Av W_0^2$; then
\begin{equation}
\left| \frac{1}{L^d} \sum_{k=1}^{L^d} \cFLk{Y_{L,k}^2}{k-1}-b^2\right| \to 0 \label{YtobConv}
\end{equation}
in measure (and therefore also in distribution) as $L \to \infty$.
\end{lemma}
\begin{proof}
 We will do this essentially by showing that $Y_{L,x}$ can be replaced by $W_x$, apart from a boundary term which vanishes in the limit $L \to \infty$.  We can of course write the summand above as
\begin{equation}
\begin{split}
\cFLk{Y^2_{L,k}}{k-1}=&\condAv{W^2_{x_k}}{\eta_{<x_k}=\zeta_{<x_k}} \\&+ \cFLk{W_{x_k}^2}{k-1} \\&- \condAv{W^2_{x_k}}{\eta_{<x_k}=\zeta_{<x_k}} \\&+ \cFLk{Y_{L,k}^2-W_{x_k}^2}{k-1},
\end{split}
\end{equation}
where by $\eta_{<x_k}$ and similar expressions we mean $\eta_x$ for $x<x_k$ in the lexicographic order;
we then deal with the different terms separately.  Letting $f(\eta)=\condAv{W^2_0}{\eta_{<0}=\zeta_{<0}}$, we use translation covariance and the fact that the conditional expectation is a projection in $\mathcal{L}^2$ to obtain %\comment{Is this the $\mathcal{L}^2$-contraction property?}
\begin{equation}
\begin{split}
&\left\| \cFLk{W_{x_k}^2}{k-1} - \condAv{W^2_{x_k}}{\eta_{<x_k}} \right\|_2
%= \left| f(T_{x_k}\eta)-\condAv{f\circ T_{-x_{k-1}}}{T_{-x_{k-1}}\mathcal{F}_L}\right|_2
\\ & \: \: \leq \left\| f - \condAv{f}{\eta_{\Gamma_R}=\zeta_{\Gamma_R}}\right\|_2 =: a_1(R) \label{a1def}
\end{split}
\end{equation}
where $R$ is the largest integer for which $T_{x_k}\Gamma_R \subset \Gamma_L$.  Employing H\"older's inequality followed by a similar step, we have
\begin{equation}
\begin{split}
&\left\| \cFLk{Y_{L,k}^2-W_{x_k}^2}{k-1} \right\|_1
\\ & \:\: \leq \left\|Y_{L,k}+W_{x_k}\right\|_2 \left\| \cFLk{Y_{L,k}-W_{x_k}}{k-1} \right\|_2 \\ & \:\: \leq 2\left\|W_{x_k}\right\|_2 \left\| W_0-\condAv{W_0}{\eta_{\Gamma_R}=\zeta_{\Gamma_R}} \right\|_2 =:a_2(R). \label{a2def}
\end{split}
\end{equation}
The $\mathcal{L}^2$-norm expressions used to define $a_1$ and $a_2$ must vanish as $R \to \infty$ and depend on $L$ only through $R$ since $f$ and $W_0$ are square-integrable.  Since the relevant terms in Equation~\eqref{YtobConv} are an average over $k$ in which the proportion of the terms with arbitrary large $R$ increases without bound as $L$ increases, these terms go to zero in measure.  The proof will be complete if we can show that
\begin{equation}
\frac{1}{L^d} \sum_{k=1}^{L^d} \condAv{W_{x_k}}{\eta_{<x_k}=\zeta_{<x_k}} \to \Av W_0^2
\end{equation}
in measure; which, given the translation covariance of $W_x$ and the fact that $\eta$ is i.i.d., follows immediately from the $\mathcal{L}^2$-ergodic theorem.
\end{proof}

To obtain inequality~\eqref{bBounds}, we note that Lemma~\ref{YBoundLemma} implies a similar bound on $|W_0|$, and therefore that $\Av W_0^2 \leq 2 \Av |\eta_0|^2$, and that
\begin{equation}
\Av W_0^2 \geq \Av \left( \condAv{W_0}{\eta_0=\zeta_0}^2 \right).
\end{equation}
% Certain properties of $G_L$ cited in Proposition~\ref{CLT_prop} carry over to similar properties of $w := \condAv{W_0}{\eta_0}$, which we will then use to obtain bounds on $b$.  Note that $w:\mathbb{R} \to \mathbb{R}$.
%
%\begin{lemma}\label{w_properties}
%\begin{enumerate}
%\item If $G_L(\zeta)$ is nondecreasing in $\zeta_0$, then so is $w$.
%\item $\lip{w} \le \epsilon$.
%\item Derivative integrates to $\epsilon M$
%\item If $B < \infty$, then $w \in C^1(\mathbb{R})$ and $\lip{w'} \leq B$.
%\end{enumerate}
%\end{lemma}
%\begin{proof}
%All of these statements begin with a property of $G_L$ which immediately extends to $Y_{L,k}$ (for all $L$ and all relevant $k$) by Equation~\ref{Ydef}, then to $W_0$ by its definition as a limit of $Y_{L,k}$, and finally to $w$ by the properties of the conditional expectation.
%\end{proof}
By dominated convergence of conditional expectations (applicable by Lemma~\ref{YBoundLemma}) and the definition of $Y$ in Equation~\eqref{Ydef},
\begin{equation}
\condAv{W_0}{\eta_0=\zeta_0} = \lim_{L \to \infty} \condAv{Y_{L,0}}{\eta_0=\zeta_0}=G_1,
\end{equation}
and so
\begin{equation}\label{W0Bounds}
\Av G_1^2 \le \Av W_0^2 \le 2 (\Av |\eta_0|)^2.
\end{equation}
%
%We may now cite some results of Aizenman and Wehr~\cite{AW.CMP} which provide lower bounds on expressions of the form $\Av g^2$ for functions $g$ satisfying these sorts of conditions.  In the following we have in mind that $\nu$ is the distribution of $\eta_0$, $g=\epsilon w$, and $\beta = B/\epsilon$.

All that remains is to show that we have a sufficiently strong control on the large fluctuations of $Y_{L,k}(\eta)$.

\begin{lemma}\label{NoLargeFluct}
For any $a>0$,
\begin{equation}
 \frac{1}{L^d} \sum_{k=1}^{L^d} \cFLk{Y^2_{L,k} I\left[|Y_{L,k}| > a L^{d/2} \right] }{k-1} \to 0
\end{equation}
in probability as $L \to \infty$.
\end{lemma}
The proof is due to Bovier~\cite{Bovier}, and provides a correction of a mistake in~\cite{AW.CMP}.
\begin{proof}
Note that the average of the left hand side above is
\begin{equation}
\begin{split}
&\Av \left[ \frac{1}{L^d} \sum_{k=1}^{L^d} \cFLk{Y^2_{L,k} I\left[|Y_{L,k}| > a L^{d/2} \right] }{k-1} \right]
\\ &\;\; = \frac{1}{L^d} \sum_{k=1}^{L^d} \Av \left( Y^2_{L,k} I\left[|Y_{L,k}| > a L^{d/2} \right] \right)
\\ &\;\; \le \frac{1}{L^d} \sum_{k=1}^{L^d} \left( \Av Y^{2q}_{L,k} \right)^{1/q} \left(P\left[|Y_{L,k}| > a L^{d/2} \right] \right)^{1/p} \label{BovierBound}
\end{split}
\end{equation}
for any $1/p+1/q=1$ by H\"older's inequality.  Chebyshev's inequality (noting $\Av Y_{L,k}=0$) %, and the fact that $Y_{L,k}$ can be expressed as a conditional average of $W_{x_k}$,
 gives
\begin{equation}
P\left[|Y_{L,k}| > a L^{d/2} \right] \le \frac{\Av Y_{L,k}^2}{a^2 L^d}
\end{equation}
and Lemma~\ref{YBoundLemma} together with the fact that $\eta_x$ has finite variance gives a uniform (in $k$) upper bound on $\Av Y_{L,k}^2$, so the right hand side goes to zero; we could conclude that the right hand side of~\eqref{BovierBound} is zero if $\Av Y^{2q}_{L,k}$ is finite for some finite $p$, i.e.\ for $q > 1$.  Lemma~\ref{YBoundLemma} implies that $\Av Y^{2q}_{L,k}$ is finite if $\eta_{x_k}$ has a finite $2q$~moment, and the Lyapunov condition is precisely the fact that this is true for some $q>1$.\footnote{This is the only place where our results require the full Lyapunov condition, and not merely existence of 2 moments.}  We then have
\begin{equation}
\lim_{L \to \infty} \Av \left[ \frac{1}{L^d} \sum_{k=1}^{L^d} \condAv{Y^2_{L,k} I\left[|Y_{L,k}| > a L^{d/2} \right] }{\etaCond{F(L,k-1)}} \right] = 0,
\end{equation}
which, since the quantities inside the average are a uniformly integrable family of nonnegative functions, can only be true if those functions converge in probability to 0.
\end{proof}

We can now apply Theorem~3.2 of~\cite{HallHeyde} and conclude that
\begin{equation}
G_L(\eta)/L^{d/2} \to N(0,b^2),
\end{equation}
and the proof of Proposition~\ref{CLT_prop} is complete.

%\subsection{Some comments on the result obtained and its peculiarities}
%
%The reader may be struck by the lack of an explicit expression for the variance of the limiting distribution in~\ref{CLT_prop}; one can be given, to wit
%\begin{equation}
%b=\sqrt{\Av W_0^2}=\sqrt{\lim_{L \to \infty} \Av \left( \condAv{G_L}{\mathcal{F}_{L,k}}-\condAv{G_L}{\mathcal{F}_{L,k-1}} \right)^2}.
%\end{equation}
%The lower bound in Equation~\eqref{W0Bounds} is however sufficient for present purposes.

%Concentration of measure estimates do something similar to the upper bound in terms of $\Av |\eta|^2$, noting that the Lipschitz constant of $G_L$ is $1$.  Different in that we are using $\| \cdot \|_1$ instead of $\| \cdot \|_2$, a fact I need to explore.  Must also get the constants right.

\section{Bounds on $b^2$}\label{restrictions_sec}

Theorem~\ref{AW_theta} is the source of a perplexing limitation remaining in our result, the exclusion of distributions with isolated point masses (in particular of discrete distributions) from most of the result.  Since the proof of this theorem is rather opaque, it seems worthwhile to give a heuristic discussion which may make the result appear less arbitrary, and clarify some of the issues involved in attempting to obtain a more powerful result.

Consider the simplest nontrivial discrete measure: let $\nu = \tfrac{1}{2} \delta_{1} + \tfrac{1}{2} \delta_{-1}$.  Given $M$ and $\beta$, is there a function $g$ satisfying the apparently relevant properties of $G_1$ ($\lip{g}\le 1$, $\int g' d\nu = M$, $\lip{g'}\le \beta$) with $\int g^2 d\nu=0$?  Quite often the answer is yes.  For example whenever $M \le \min(\beta/3,1)$, the function
\begin{equation}
g_1(x) = \left\{ \begin{array}{ll}
\frac{M}{2} (x^3-x), & -1 \le x \le 1 \\
M(x-1), & x>1 \\
M(x+1), & x < -1
\end{array}\right.
\end{equation}
clearly fits these requirements.  This is of course not a monotone function, and indeed it is more or less obvious that (as Theorem~\ref{AW_theta} states) no monotone, continuously differentiable function will do: if $g'(\pm 1)> 0$, then for such a function $g(\mp 1) \neq 0$.

This would appear to leave little room for improvement in Proposition~\ref{CLT_prop}, but this is not quite the case.  Proposition~\ref{CLT_prop} stipulates that $G_L(\zeta)$ should be monotone in $\zeta_0$ for all $L$ and $\zeta$, but all that is needed is that $G_1$ be monotone, which should be a weaker requirement.  It is, however, not immediately clear that it follows from something as well-known as the FKG inequalities which imply the special case.

\chapter{Free energy fluctuations}\label{G_section}\label{End.proof}

\section{Definition of $G_L$}

In order to prove Proposition~\ref{main_prop}, we need to construct a sequence of functions $G_L$ which represent the effect of the random field $\eta$ on the free energy difference between states with the largest and smallest permissible values of the order parameter.  Since our most robust way of accessing these states is by taking a limit in the uniform field $h$ which couples to $\eta$, and since we are concerned with the thermodynamic limit, it should be plausible that one candidate is described by the formal expression
\begin{equation}\label{formalG}
\lim_{\delta \to 0^+}\lim_{M \to \infty} \frac{1}{2} \condAv{F^{h+\delta}_M (\eta) - F^{h+\delta}_M(r_L(\eta)) - F^{h-\delta}_M (\eta) + F^{h-\delta}_M(r_L(\eta))}{\eta_{\Gamma_L}=\zeta_{\Gamma_L}}
\end{equation}
where $r_L$ is the function which sets the field to $0$ inside $\Gamma_L$; $F_M$ is the free energy of the system on the finite domain $\Gamma_L$ with periodic boundary conditions (defined in Equation~\eqref{Fdef}), and for brevity we have omitted the argument $\upsilon$, or in other words we let $F(\zeta) = \condAv{F(\eta,\ul{\upsilon})}{\eta=\zeta}$; and $\eta_{\Gamma_L}$ is the collection $\eta_x$ where $T_x A_0 \subset \Gamma_L$ and likewise for similar expressions.

 However it is hardly clear that the expression~\eqref{formalG} is well-defined.  We will show that a quite similar quantity is, but first let us turn to a few observations which should help motivate this choice.

For convenience, we let
\begin{equation}\label{GdeltaDef}
\gf{L,M}{\delta}(\zeta) := \frac{1}{2} \condAv{F^{h+\delta}_M (\eta) - F^{h+\delta}_M(r_L(\eta)) - F^{h-\delta}_M (\eta) + F^{h-\delta}_M(r_L(\eta))}{\eta_{\Gamma_L} = \zeta_{\Gamma_L}}
\end{equation}
for $M \ge L$.
Then
\begin{equation}
\begin{split}
\frac{\partial \gf{L,M}{\delta}}{\partial \zeta_x} &= \frac{1}{2} \condAv{\frac{\partial F^{h+\delta}_M}{\partial\eta_x} - \frac{\partial F^{h-\delta}_M}{\partial \eta_x}}{\eta_{\Gamma_L} = \zeta_{\Gamma_L}} \\ &= \frac{1}{2} \condAv{\state{\lop_x}^{h+\delta}_M(\eta) - \state{\lop_x}^{h-\delta}_M(\eta)}{\eta_{\Gamma_L} = \zeta_{\Gamma_L}}.
\end{split}
\end{equation}
This means that $\lip{ \gf{L,M}{\delta} } \le 1$ uniformly in all parameters, which will carry over in the limit $M \to \infty$ to assumption~\ref{uLif} in Proposition~\ref{CLT_prop}.  It also means that
\begin{equation}
\Av \frac{\partial \gf{L,M}{\delta}}{\partial \zeta_x} = \frac{1}{2} \Av \left( \state{\lop_x}^{h+\delta}_M(\eta) - \state{\lop_x}^{h-\delta}_M(\eta)\right)
\end{equation}
which in light of Corollary~\ref{LongLongRange} should mean
\begin{equation}
\Av G_1' = \frac{1}{2} \left( \frac{\partial \F}{\partial h +} - \frac{\partial \F }{\partial h -} \right),
\end{equation}
allowing the desired control on $b$.
The use of a conditional expectation in the definition should take care of assumptions~\ref{local_hypothesis} and~\ref{consistency_assumption}, and we can arrange for the remaining assumption (mean zero) by simply subtracting the mean value.

Let us return to this more carefully:
\begin{proposition}\label{convergence_prop}
Let at least one of the following hold:
\begin{enumerate}
\item $\beta < \infty$
\item The distribution of $\eta_x$ is absolutely continuous with respect to the Lebesgue measure
\end{enumerate}

Then there is a decreasing sequence $\delta_i \to 0$ and an increasing sequence of integers $M_j \to \infty$ such that
\begin{equation} \label{GLDef}
G_L(\zeta) := \lim_{i \to \infty} \lim_{j \to \infty} \left( \gf{L,M_j}{\delta_i} - \Av \left[ \gf{L,M_j}{\delta_i} \right] \right)
\end{equation}
exists for all $\zeta \in \Fields$ and all $L \in \mathbb{N}$.  Furthermore, the family $G_L$ satisfies the hypotheses of Proposition~\ref{CLT_prop}, and $G_1$ has a distributional derivative $G_1'$ satisfying
\begin{equation}
\Av G_1' = \frac{1}{2} \left( \frac{\partial \F}{\partial h +} - \frac{\partial \F }{\partial h -} \right).
\end{equation}
\end{proposition}

I wish to point out that the proof will not assume that $\eta$ and $\ul \upsilon$ are mutually independent, but only that the different $\eta_x$ remain independent when conditioned on $\ul \upsilon$ - this will be important for systems with continuous symmetries, where $\eta_x$ and $\ul \upsilon_x$ will represent different components of a random vector, and may therefore be correlated.

\section{Proof of proposition~\ref{convergence_prop} - finite temperature}\label{convergence.section.1}

In this situation we will proceed by constructing functions $\tau$ which satisfy the conditions of Theorem~\ref{tau_into_CLT}.

The following convergence argument will be used frequently in what follows.
\begin{lemma} \label{compactness_lemma_1}
Let $f_{ij}:\R^N \to \R$ be a family of functions labeled by $i,j \in \N$, each satisfying $\lip{f_{ij}} \le 1$ and $f_{ij}(0) = 0$.  Then there are subsequences $i_k,j_l$ such that
\begin{equation}
f(z) = \lim_{k \to \infty} \lim_{l \to \infty} f_{i_k j_l}(z)
\end{equation}
exists for all $z \in \R^N$.  Furthermore the convergence is uniform on any compact $\Xi \subset \R^N$.
\end{lemma}
\begin{proof}
Note that the condition $\lip{f_{ij}} \le 1$ implies uniform equicontinuity.  On the compact domain $\Xi_n := [-n,n]^N$ we have the uniform bound $|f_{xy}| \le Nn$, and so by the Arzel\`{a}-Ascoli theorem, any infinite collection of these functions has a subsequence which converges uniformly on $\Xi_n$.

We then apply the diagonal subsequence trick as follows: there is a sequence $j^1_l$ so that $f_{1,j^1_l}$ converges uniformly on $\Xi_1$, which has a subsequence $j^2_l$ so that $f_{2,j^2_l}$ and $f_{1,j^1_l}$ converge uniformly on $\Xi_2$ and so on.  Then the diagonal subsequence $j_l=j^l_l$ has the property that for any $k,n$ $f_{n,j_l^l}$ converges uniformly on $\Xi_n$, with limits $f_k:\R^N \to \R$ with the same properties we have used above.  By the same argument, we can now choose a sequence $i^1_k$ so that $f_{i^1_k}$ converges uniformly on $\Xi_1$, a subsequence $i^2_k$ so that $f_{i^2_k}$ converges on $\Xi_2$, etc.  Then with $i_k := i^k_k$, $j_l:=j^l_l$, we have the desired result.
\end{proof}

The same argument also gives
\begin{lemma}\label{compactness_lemma_2}
Let $f_{ij}:\R^N \to \R$ be a family of functions labeled by $i,j \in \N$, each satisfying $\lip{f_{ij}} \le 1$ and $|f_{ij}(z)| \le c < \infty$ for all $z \in \R^N$.  Then there are subsequences $i_k,j_l$ such that
\begin{equation}
f(z) = \lim_{k \to \infty} \lim_{l \to \infty} f_{i_k j_l}(z)
\end{equation}
exists for all $z \in \R^N$.  Furthermore the convergence is uniform on any compact $\Xi \subset \R^N$.
\end{lemma}

Now let
\begin{equation}\label{gf_derivs}
\theta^x_{M,\delta}(\zeta) := \frac{\partial \gf{L,M}{\delta}}{\partial \zeta_x} = \frac{1}{2} \left( \Av \state{\lop_x}^{h+\delta}_M(\zeta,\ul \upsilon) - \Av \state{\lop_x}^{h-\delta}_M(\zeta,\ul \upsilon)\right);
\end{equation}
note that the quantity defined does not depend on $L$, and that $\zeta$ is fixed in the averages which are taken over the other fields $\ul{\upsilon}$.  To simplify the similar expressions appearing below we will write $\state{\cdot}^h_M(\zeta) := \Av\state{\cdot}^h_M(\zeta,\ul \upsilon)$.  Also, evidently
\begin{equation} \label{thetaBound}
| \theta^x_{M,\delta}(\zeta) | \le \| \lop_x \| = 1
\end{equation}
and
\begin{equation}
\begin{split}
\left| \frac{\partial \theta^x_{M,\delta}}{\partial \zeta_y} \right| = \frac{\beta}{2} & \left| \state{\lop_x \lop_y}^{h+\delta}_M(\zeta) - \state{\lop_x \lop_y}^{h-\delta}_M(\zeta) \right. \\
& \left. - \state{\lop_x}^{h+\delta}_M(\zeta)  \state{\lop_y}^{h+\delta}_M(\zeta) + \state{\lop_x}^{h-\delta}_M(\zeta)  \state{\lop_y}^{h-\delta}_M(\zeta)
\right|\leq 2\beta;
\end{split}
\end{equation}
then
\begin{equation} \label{phiDef}
\phi^x_{L,M,\delta}(\zeta) := \condAv{\frac{\partial \gf{L,M}{\delta}}{\partial \eta_x}}{\eta_{\Gamma_L}} = \condAv{\theta^x_{M,\delta}(\eta)}{\eta_{\Gamma_L}}
\end{equation}
obeys the same bounds,
\begin{gather}
| \phi^x_{L,M,\delta}(\zeta)| \le 1 \\
\left| \frac{\partial \theta^x_{M,\delta}}{\partial \zeta_y} \right| \le 2 \beta
\end{gather}

For $\beta < \infty$, this means that for each $L$, $\phi^x_{L,M,\delta}$ is a uniformly equicontinuous family of functions of the $L^d$ variables $\zeta_{\Gamma_L}$.  We can apply Lemma~\ref{compactness_lemma_2} to find a decreasing sequence $\delta_i \to 0$ and an increasing sequence $M_j \to \infty$ (by applying the diagonal subsequence trick, we can choose them to be independent of $x$ and $L$) so that
\begin{equation}
\psi^x_L (\zeta) := \lim_{i \to \infty} \lim_{j \to \infty} \phi^x_{L,M_j,\delta_i}(\zeta)
\end{equation}
exists, and by uniformity of convergence
\begin{equation} \label{GL_exists_as_lim}
G_L(\zeta) := \lim_{i \to \infty} \lim_{j \to \infty} \left( \gf{L,M_j}{\delta_i}(\eta) - \Av \gf{L,M_j}{\delta_i} \right)
\end{equation}
also exists with
\begin{equation}
\frac{\partial G_L}{\partial \zeta_x} =  \psi^x_L(\zeta)
\end{equation}
for all $L$, $\zeta$, $x \in \Gamma_L$.

%\comment{need to take care of $\psi$}If $\beta = \infty$, we note that it is obvious from Equation~\ref{GdeltaDef} that $\gf{L,M}{\delta}(0)=0$; Inequality~\eqref{thetaBound} implies $\lip{\gf{L,M}{\delta}} \le 1$, so we can apply Lemma~\ref{compactness_lemma_1} to obtain~\eqref{GL_exists_as_lim}, with $\| G_L\|\le 1$.

Note that Equation~\eqref{phiDef} implies that, for any $K<L\le M$,
\begin{equation} \label{phiMartin}
\phi^x_{K,M,\delta} = \condAv{\phi^x_{L,M,\delta}}{\eta_{\Gamma_K}},
\end{equation}
and by the conditional form of the dominated convergence theorem this implies that
\begin{equation}
\psi^x_{K} = \condAv{\psi^x_L}{\eta_K},
\end{equation}
which makes $\psi^x_{K}$ a martingale; $|\psi^x_{K}|\le 1$ makes it a uniformly integrable one, and applying the relevant martingale convergence theorem we see that
\begin{equation}
\tau_x := \lim_{L \to \infty} \psi^x_L
\end{equation}
exists as an $\mathcal{L}^1$ limit, with
\begin{equation}
\condAv{\tau_x}{\eta_{\Gamma_L}} = \psi^x_L = \frac{\partial G_L}{\partial \zeta_x}.
\end{equation}
Following through the various limits, we see that $|\tau_x (\zeta)| \le 1$ and $\left| \frac{ \partial \tau_x}{\partial \zeta_y} \right| \le 2 \beta$.

We have defined $\gf{L,M}{\delta}$ in terms of periodic boundary conditions, so Equation~\eqref{phiMartin} also implies
\begin{equation}
\phi^x_{K,M,\delta}\circ T_y = \condAv{\phi^{x-y}_{L,M,\delta}}{\eta_{\Gamma_K}=T_y\zeta_{T_{-y} \Gamma_K}},
\end{equation}
which is translation covariance.  Following through the limits used to define $\psi$ and $\tau$ (thanks to the fact that the sequences involved are independent of $L$ and $x$), we see that this implies
\begin{equation}
\tau_{x+y}(T_y \zeta) = \tau_x(\zeta).
\end{equation}

Finally,
\begin{equation}
\begin{split}
\Av \tau_x = \Av \tau_0 = \Av \psi^0_1 = \lim_{i \to \infty} \lim_{j \to \infty} \Av \phi^0_{1,M_j,\delta_i} = \lim_{i \to \infty} \lim_{j \to \infty} \Av \theta^0_{M_j,\delta_i}
\end{split}
\end{equation}
and applying Theoroms~\ref{LongLongRange} and~\ref{mag_ergodicity} we have
\begin{equation}
\Av \tau_x = \lim_{i \to \infty} \lim_{j \to \infty} \frac{1}{2} \Av \left( \state{\lop_x}^{h+\delta}_M(\eta) - \state{\lop_x}^{h-\delta}_M(\eta)\right) =  \frac{1}{2} \left( \frac{\partial \F}{\partial h +} - \frac{\partial \F }{\partial h -} \right).
\end{equation}

\section{Proof of proposition~\ref{convergence_prop} for absolutely continuous distributions of $\eta$}\label{convergence.section.2}

%\begin{proposition}
%There is a decreasing sequence $\delta_i \to 0$ and an increasing sequence of integers $M_j \to \infty$ such that
%\begin{equation} %\label{GLDef}
%G_L(\zeta) := \lim_{i \to \infty} \lim_{j \to \infty} \left( \gf{L,M_j}{\delta_i} - \Av \left[ \gf{L,M_j}{\delta_i} \right] \right)
%\end{equation}
%exists for all $\zeta \in \Fields$ and all $L \in \mathbb{N}$.  Furthermore, the family $G_L$ satisfies the hypotheses of Proposition\ref{CLT_prop}, and $G_1$ has a distributional derivative $G_1'$ with $\Av G_1' = M$.
%\end{proposition}

First of all, note that it is obvious from Equation~\ref{GdeltaDef} that $\gf{L,M}{\delta}(0)=0$; Inequality~\eqref{thetaBound} implies $\lip{\gf{L,M}{\delta}} \le 1$, so we can apply Lemma~\ref{compactness_lemma_1} to obtain~\eqref{GL_exists_as_lim}, with $\| G_L\|\le 1$.  $\Av G_L =0$ is obvious.  We can apply the diagonal subsequence trick to obtain sequences independent of $L$, which implies that the consistency condition~\ref{consistency_assumption} of Proposition~\ref{CLT_prop} is satisfied.

The hard part is to show $\Av G_1'=M$.  Without uniform equicontinuity of the derivatives, we have no reason to expect that an object like $\phi$ of the previous section will converge uniformly, and without that we have no reason to expect that a pointwise limit will still be a derivative.  However the following theorem allows us to find a particular kind of weak limit which will do the trick:
\begin{theorem}\label{weak_conv_derivatives_thm}
Let $g$ be a measurable function and $g_n$ a sequence of measurable functions such that $\|g_n\|\le 1$, $\|g\|\le 1$, and
\begin{equation}
\lim_{n \to \infty} \int_a^b g_n(x) dx = \int_a^b g(x) dx
\end{equation}
for all $a,b \in \R$.  Then for any signed measure $\nu$ with finite total variation which is absolutely continuous with respect to the Lebesgue measure,
\begin{equation} \label{weak_conv_conclusion}
\lim_{n \to \infty} \int g_n d\nu = \int g d\nu.
\end{equation}
\end{theorem}
\begin{proof}
Without loss of generality we can consider only positive measures (which are all that is necessary for the present work anyway), thanks to Hahn's decomposition theorem~\cite{Billingsley_probability}.  As a preliminary, we see that for any Borel set $A$ contained in a bounded interval $I$
\begin{equation}\label{weak_conv_intervals}
\lim_{n \to \infty} \int_A g_n(x) dx = \int_A g(x) dx
\end{equation}
since for any $\epsilon$ there is a set $E_\epsilon$ which is a finite union of intervals which approximates $A$ in the sense that $\lambda(E_\epsilon \Delta A) \le \epsilon$.  Then also
\begin{equation}
\left|\int_{E_\epsilon} g_n d\lambda - \int_A g_n d\lambda \right| \le \epsilon
\end{equation}
uniformly in $n$, so we can take the limit $\epsilon \to 0$ and exchange the order of the limits to obtain Equation~\eqref{weak_conv_intervals}.

There is a nondecreasing sequence $\sigma_m$ of simple functions so that $\sigma_m \to \frac{d \nu}{d \lambda}$ pointwise, and each $\sigma_m$ has bounded support.  By Equation~\eqref{weak_conv_intervals},
\begin{equation}
\lim_{n \to \infty} \int g_n(x) \sigma_m dx = \int g(x) \sigma_m dx
\end{equation}
for all $m$.  In fact since $\sigma_m$ are a nondecreasing sequence, we can apply the Beppo Levi theorem~\cite{Billingsley_probability} to obtain
\begin{equation}
\lim_{m \to \infty} \int g^\pm_n \sigma_m d\lambda = \int g^\pm \frac{d \nu}{d \lambda} d \lambda = \int g^\pm_n d \nu
\end{equation}
and thus
\begin{equation} \label{weak_conv_1}
\lim_{m \to \infty} \int g_n \sigma_m d\lambda = \int g_n d \nu.
\end{equation}
Furthermore,
\begin{equation}
\begin{split}
\left| \int g_n \sigma_m d\lambda - \int g_n d\nu \right| &=
\left| \int g_n \left(\sigma_m - \frac{d \nu}{d \lambda}\right) d\lambda \right|
\\ &\le\int |g_n| \left(\frac{d \nu}{d \lambda} - \sigma_m \right) d\lambda \le
\int \left(\frac{d \nu}{d \lambda} - \sigma_m \right) d\lambda
\end{split}
\end{equation}
and since this last bound is independent of $n$, the convergence in Equation~\eqref{weak_conv_1} is uniform in $n$.  Taking the limit $n \to \infty$ and exchanging the limits on the left hand side gives Equation~\eqref{weak_conv_conclusion}.
\end{proof}

The set of absolutely continuous finite signed measures is isomorphic to $\Lp{1}{\R}$, the predual of $\Lp{\infty}{\R}$, so the substance of Equation~\eqref{weak_conv_conclusion} is also expressed by saying that $g_n \to g$ in the weak-* topology of $\Lp{\infty}{\R}$.  This is a convenient way of phrasing the following:\footnote{A comparable statement appears in the proof of Rademacher's theorem in~\cite{Heinonen}; thus the proof here is more logically circuitous than necessary, but we hope it is the most intelligible way to convey things to our readers.}
\begin{corollary}\label{weak_conv_corr}
Let $f_n$ be a sequence of functions $\R \to \R$ such that $f_n \to f$ pointwise, and $\lip{f_n} \le 1$.  Then their distributional derivatives converge to the distributional derivative of $f$ ($f_n' \to f'$) in the weak-* topology of $\Lp{\infty}{\R}$.
\end{corollary}
\begin{proof}
Thanks to Rademacher's theorem~\cite{Heinonen}, Lipschitz continuity guarantees that the distributional derivatives, i.e.\ functions satisfying
\begin{equation}
\int_a^b f_n' dx = f_n(b)-f_n(a)
\end{equation}
for any $a,b \in \R$, exist with $\|f_n'\| = \lip{f_n} \le 1$, and the convergence $f_n \to f$ then implies
\begin{equation}
\lim_{n \to \infty} \int_a^b f_n' dx = \int_a^b f' dx.
\end{equation}
This allows us to apply Theorem~\ref{weak_conv_derivatives_thm} and the result follows immediately.
\end{proof}

Applying Corollary~\ref{weak_conv_corr} twice to $\gf{1,M_j}{\delta_i}(\eta) - \Av \gf{1,M_j}{\delta_i}$ and $G_1$, we obtain
\begin{equation}
\Av G_1' = \lim_{i \to \infty} \lim_{j \to \infty} \Av \left( \gf{1,M_j}{\delta_i}\right)'(\eta)
\end{equation}
and by Equation~\eqref{gf_derivs} and Theorems~\ref{LongLongRange} and~\ref{mag_ergodicity} we obtain
\begin{equation}
\Av G_1' =  \frac{1}{2} \left( \frac{\partial \F}{\partial h +} - \frac{\partial \F }{\partial h -} \right)
\end{equation}
and the proof of Proposition~\ref{convergence_prop} is complete.

\section{Proof of Proposition~\ref{main_prop}} \label{wrapup}

We now turn to the boundary estimate~\eqref{simpleUpperBound}.  Let $\Lambda_L$ be the smallest subset of $\Zd$ so that $T_x A_0 \subset \Lambda_L$ (i.e.\ $\lop_x \in \alg_{\Lambda_L}$) for all $x \in \Gamma_L$, and for $\Gamma_M \supset \Lambda_L$ let
\begin{equation}
F^h_{M | L} (\zeta) := - \Av \frac{1}{\beta} \log \Tr \exp \left( - \beta H_{\Lambda_L,0}^{h,\zeta,\ul \upsilon} -\beta H_{\Gamma_M \setminus \Gamma_L *}^{h,\zeta, \ul \upsilon}\right),
\end{equation}
where the subscript $0$ refers to free boundary conditions, and the subscript $*$ refers to periodic boundary conditions on the edge of $\Gamma_M$ and free boundary conditions on the edge of $\Lambda_L$; this lets us write
\begin{equation}
H_{\Gamma_M}^{h,\zeta, \ul \omega} = H_{\Lambda_L,0}^{h,\zeta,\ul \omega} + P_{\Lambda_L} \left( V_{\Lambda_L}^{\zeta,\ul \omega} \right) + H_{\Gamma_M \setminus \Gamma_L *}^{h,\zeta, \ul \omega},
\end{equation}
whence, by Lemma~\ref{Ruelle_lemma}
\begin{equation}
\left| F^h_{M | L} (\zeta) - F^h_{M} (\zeta) \right| \le \Av \left\| V_{\Lambda_L}^{\zeta,\ul \upsilon} \right\|;
\end{equation}
then
\begin{equation}
\begin{split}
& \left| F^{h+\delta}_M (\zeta) - F^{h+\delta}_M(r_L(\zeta)) - F^{h-\delta}_M (\zeta) + F^{h-\delta}_M(r_L(\zeta)) \right| \\ \ \ &
\le \left| F^{h+\delta}_{M | L} (\zeta) - F^{h+\delta}_{M | L}(r_L(\zeta)) - F^{h-\delta}_{M | L} (\zeta) + F^{h-\delta}_{M | L}(r_L(\zeta)) \right| + 4 \Av \left\| V_{\Lambda_L}^{\zeta,\ul \upsilon} \right\|.
\end{split}\label{Boundary_estimate_1}
\end{equation}

Since $H_{\Lambda_L0}^{h,\zeta,\ul{\omega}} $ and $H_{\Gamma_M \setminus \Gamma_L *}^{h,\zeta,\ul{\omega}} $ act on disjoint subsets of the lattice, they commute, and
\begin{equation}
F^h_{M | L} (\zeta) = F^h_{\Lambda_L,0} (\zeta) + F^h_{M \setminus L*} (\zeta),
\end{equation}
where
\begin{equation}
F^h_{M \setminus L} (\zeta) := - \Av \frac{1}{\beta} \log \Tr \exp \left( - \beta H_{\Gamma_M \setminus \Lambda_L *}^{h,\zeta, \ul \upsilon}\right).
\end{equation}
When we use this to expand the right hand side of~\eqref{Boundary_estimate_1}, the $\Gamma_M \setminus \Lambda_L$ terms cancel:
\begin{equation}
\begin{split}
F^{h+\delta}_{M | L} (\eta) - F^{h+\delta}_{M | L}(r_L(\eta)) - F^{h-\delta}_{M | L} (\eta) + F^{h-\delta}_{M | L}(r_L(\eta))
\\= F^{h+\delta}_{\Lambda_L,0}(\eta) - F^{h-\delta}_{\Lambda_L,0}(\eta) - F^{h+\delta}_{\Lambda_L,0}(0) + F^{h-\delta}_{\Lambda_L,0}(0)
\end{split}
\end{equation}
and we can apply Lemma~\ref{Ruelle_lemma} again to bound this quantity, obtaining
\begin{equation}
\left| F^{h+\delta}_{M | L} (\eta) - F^{h+\delta}_{M | L}(r_L(\eta)) - F^{h-\delta}_{M | L} (\eta) + F^{h-\delta}_{M | L}(r_L(\eta)) \right| \le 2 \delta |\Lambda_L| = O(\delta L^d),
\end{equation}
where the last term is the effect of the constant field inside $\Lambda_L$.

Plugging this back into Inequality~\ref{Boundary_estimate_1} gives
\begin{equation}
\left| F^{h+\delta}_M (\eta) - F^{h+\delta}_M(r_L(\eta)) - F^{h-\delta}_M (\eta) + F^{h-\delta}_M(r_L(\eta)) \right|
\le 4 \Av \left\| V_{\Lambda_L}^{\zeta,\ul \upsilon} \right\| + O( \delta L^d),
\end{equation}
which gives
\begin{equation}
|\gf{L,M}{\delta}(\zeta)| \le 2\Av \left\| V_{\Lambda_L}^{\zeta,\ul \upsilon} \right\| + O(\delta L^d)
\end{equation}
and since the $\delta$ term is uniform in $M$,
\begin{equation}\label{G_L_upper_bound}
|G_L(\zeta)| \le 2 \Av \left\| V_{\Lambda_L}^{\zeta,\ul \upsilon} \right\|,
\end{equation}
which we have assumed (Assumption~\ref{short_range_assumption}) to be $O(L^{d-1})+O(L^{d/2})$.

Now we need to show that this is in contradiction with Proposition~\ref{CLT_prop} unless $b=0$.  To demonstrate this absolutely clearly, we will convert these to statements about the moment generating functions of $G_L$, $\Av e^{t G_L}$.  The distribution of a random variable is uniquely characterized by its moment generating function provided this is finite on a sufficient region~\cite{Billingsley_probability}, which in this context is guaranteed by~\eqref{G_L_upper_bound}; and then convergence in distribution is equivalent to pointwise convergence of moment generating functions.  Then the conclusion of Proposition~\ref{CLT_prop} can be restated
\begin{equation}\label{mgf_CLT}
\lim_{L \to \infty} \Av \exp \left( t G_L / L^{d/2} \right) = \exp( t^2 b^2 / 2).
\end{equation}
At the same time, if $|G_L| \le A L^{d/2}$ then
\begin{equation}\label{mgf_bound}
\Av e^{t G_L/L^{d/2}} \le e^{tA}
\end{equation}
for all positive $t$; clearly if $b \ne 0$, this will be incompatible with~\eqref{mgf_CLT} for sufficiently large $t$.  Finally we note that Theorem~\ref{AW_theta} states that $b=0$ implies $M=0$ under any of the cases listed in Proposition~\ref{main_prop}, and the proof is complete.

\section{Systems with continuous symmetry: Proof of Proposition~\ref{continuous_prop}} \label{continuous_wrapup}
%Obtain improved upper bounds for special case of systems with continuous symmetry.  Combine with above to conclude that $d \le 4$ case is complete.

As noted above, the main requirement of the proof of Proposition~\ref{continuous_prop} is based on the improved bound
\begin{equation}
|G_L(\zeta)| \le K L^{d-2}
\end{equation}
which should hold at $\vec{h} = 0$  We first note that Proposition~\ref{convergence_prop} holds for the vector case with the following definitions, corresponding to Equations~\eqref{GdeltaDef} and~\eqref{GLDef}:
\begin{gather}
\gf{L,M}{\delta\hat{e}}(\zeta) = \frac{1}{2} \condAv{
F^{\delta\hat{e}}_M(\vec{\eta}) - F^{\delta\hat{e}}_M(r_L(\vec{\eta}))
- F^{-\delta\hat{e}}_M(\vec{\eta})
+ F^{-\delta\hat{e}}_M(r_L(\vec{\eta}))
}{\hat{e} \cdot \vec{\eta}_L = \zeta_L}
\\
G^{\hat{e}}_L(\zeta) = \lim_{i \to \infty} \lim_{j \to \infty} \left( \gf{L,M_j}{\delta \hat{e}}(\zeta)- \Av \left[  \gf{L,M_j}{\delta \hat{e}}(\hat{e} \cdot \vec{\eta}) \right] \right)
\end{gather}
where $\hat{e}$, an arbitrary unit vector, defines the component of the order parameter being examined.

We can obtain the desired bound by focusing on
\begin{equation} \label{g_mini_def}
g_{L,M}^{\delta \hat{e}} (\zeta) := \condAv{F_M^{\delta \hat{e}}(\vec{\eta}) - F_M^{-\delta \hat{e}}(\vec\eta) }{\vecEtaCond{L}};
\end{equation}
since %\comment{$g_{L,M}^{\delta \hat{e}} (0)$ is equal to its average and cancels out, right?}
\begin{equation}\label{G_and_g}
G^{\hat{e}}_L(\zeta) = \frac{1}{2} \lim_{i \to \infty} \lim_{j \to \infty} \left( g_{L,M}^{\delta \hat{e}} (\zeta) - \Av g_{L,M}^{\delta \hat{e}} (\hat{e} \cdot \vec\eta) \right),
\end{equation}
it is easy to turn uniform bounds on $|g_{L,M}^{\delta \hat{e}} (\vec\zeta)|$ into similar bounds on $|G_L|$.

\begin{lemma}\label{little_g_bound_lemma}
With $g$ defined above, Assumption~\ref{continuous_short_range_assum} implies
\begin{equation}
|g_{L,M}^{\delta \hat{e}} (\vec\zeta)| = O(L^{d-2})
\end{equation}
\end{lemma}
\begin{proof}
Let $\rho$ be the generator (in $so(N)$) of a rotation in a plane containing $\hat{e}$, and for each $x \in \Zd$ let $\rho_x$ be the generator of the corresponding rotation in the single-site algebra $\alg_x$.\footnote{Unlike in~\cite{QIMLetter}, we will use the ``mathematician's'' convention that rotations are given by $e^{\theta \rho}$, so $\rho$ is an antisymmetric matrix and $\rho_x$ is an antihermitian operator.}  We introduce the slowly varying angles
\begin{equation}\label{angles_theta_def}
\theta_x := \piecewise{
0, & x \in \Gamma_L \\
\frac{\|x\|_1-L}{L}\pi, & 0 < d_L(x) < L  \\
\pi, & \|x \|_1 d_L(x) \ge L},
\end{equation}
where $d_L(x)$ is the distance from $x$ to $\Gamma_L$ in the largest-component metric, i.e.\ \begin{equation}
 d_L(x):= \min_{y \in \Gamma_L} \pnorm{x-y}{\infty}.
\end{equation}
We also introduce the associated rotations on fields and on $\alg$ defined by
\begin{gather}
R_x := e^{\theta_x \rho} \\
\left(R_\theta(\vec \zeta)\right)_x \equiv R_x \vec \zeta_x \\
\hat{R}_\theta = \bigotimes_{x \in \Zd} e^{\theta_x \rho_x}.
\end{gather}

$\hat{R}_\theta$ is unitary, and so we can rewrite the free energy $F_M^{-\delta \hat{e}}(\vec\eta) $ appearing in~\eqref{g_mini_def} as
\begin{equation}
F_M^{-\delta \hat{e}}(\vec\zeta) = -\frac{1}{\beta} \log \Tr \exp\left( -\beta \hat{R}_\theta^{-1} H_{\Gamma}^{-\delta \hat{e},\vec{\zeta},\ul{\vec{\omega}}} \hat{R}_\theta  \right);
\end{equation}
we wish to use this to obtain something of the form
\begin{equation}
\begin{split}
&\condAv{F_M^{-\delta \hat{e}}(\vec\eta)}{\vecEtaCond{L}}
\\& \; \; \;= -\frac{1}{\beta} \condAv{\log \Tr \exp\left( -\beta [ H_{\Gamma}^{\delta \hat{e},\vec{\zeta},\ul{\vec{\omega}}} + \Delta H_\theta ] \right)}{\vecEtaCond{L}},
\end{split}
\end{equation}
which by Lemma~\ref{Ruelle_lemma} implies
\begin{equation}
|g_{L,M}^{\delta \hat{e}} (\vec\zeta)| \le \|\Delta H_\theta\|;
\end{equation}
however this will not quite be sufficient, since we are not able to establish suitable control over $\|\Delta H_\theta\|$.  Instead, we will split $F_M^{-\delta \hat{e}}(\vec\zeta)$ in half and rewrite each part separately by applying an opposite rotation, to obtain
\begin{equation}\label{F_bounds_for_g}
\begin{split}
\condAv{F_M^{-\delta \hat{e}}(\vec\eta)}{\vecEtaCond{L}} =  -\frac{1}{2 \beta} \Av & \left[ \log \Tr \exp\left( -\beta [ H_{\Gamma}^{\delta \hat{e},\vec{\zeta},\ul{\vec{\omega}}} + \Delta H_\theta ] \right) \right.
\\ & \left. + \log \Tr \exp\left( -\beta [ H_{\Gamma}^{\delta \hat{e},\vec{\zeta},\ul{\vec{\omega}}} + \Delta H_{-\theta} ] \right) \middle| \vecEtaCond{L}\right].
\end{split}
\end{equation}

Combining the Cauchy-Schwarz inequality, the Golden-Thompson inequality, and Lemma~\ref{Ruelle_lemma}, we quickly derive the general inequality
\begin{equation}
\begin{split}
\log \Tr e^A & - \log \Tr e^{B/2} - \log \Tr e^{C/2}
= \log \left( \frac{ \Tr e^A}{\Tr e^{B/2} \Tr e^{C/2}} \right)
\\ & \le \log \left( \frac{ \Tr e^A}{\Tr e^{B/2} e^{C/2}} \right)
\le \log \Tr e^{A-(B+C)/2} \le \left\| A- \frac{B+C}{2} \right\|
\end{split}
\end{equation}
for arbitrary Hermitian matrices $A,B,C$.  Applying this to Inequality~\eqref{F_bounds_for_g} gives
\begin{equation}\label{g_upper_bound}
g_{L,M}^{\delta \hat{e}} (\vec\zeta) \le \frac{1}{2} \left\| \Delta H_\theta + \Delta H_{-\theta} \right\|.
\end{equation}

Now recall Equation~\eqref{vecHam}:
\begin{equation}
\hat{R}_\theta^-1 H_{\Gamma}^{h,\vec{\zeta},\ul{\vec{\omega}}} \hat{R}_\theta = \hat{R}_\theta^{-1} \left( \sum_X P_\Gamma(\inter_0(X)) + \sum_{x \in \Gamma} (-\delta \hat{e}+\vec\zeta_x) \cdot P_\Gamma(\vec{\lop}_x) + %\sum_{\alpha=1}^{N_\alpha} \sum_{x \in \Gamma} \omega_{\alpha x} \cdot P_\Gamma (\vec{\gamma}_{\alpha x})
\right) \hat{R}_\theta
\end{equation}
Since $\vec{\lop}$
%and $\ul{\vec \gamma}$
are vector operators (recall Equation~\eqref{vector_operator_transformation}),
\begin{equation}
\vec \zeta_x \cdot \left( \hat{R}_\theta^{-1} \vec \lop_x \hat{R}_\theta \right) = \vec \zeta_x \cdot R_\theta(\vec \lop)_x = \left[R^{-1}_\theta(\vec \zeta)_x\right] \cdot \vec \lop_x;
\end{equation}
%and similarly for $\ul{ \vec \gamma}$.
now inside $\Gamma_L$ there is no rotation, and outside we are performing an average with respect to an isotropic distribution, so this term makes no contribution to $\Delta H$.

As for the fixed field terms, we have
\begin{equation}
\hat{e} \cdot \left( \hat{R}_\theta^{-1} \vec \lop_x \hat{R}_\theta \right) = \left( R_x^{-1} \hat e \right) \cdot \vec\lop_x.
\end{equation}
The choices of $\rho$ and $\theta$ were intended precisely to make $R_x \hat e = -\hat e$ for $d_L(x) > L$; and for the remaining $(3L)^d$ sites we have $\left\| \hat{e} \cdot \left( \hat{R}_\theta^{-1} \vec \lop_x \hat{R}_\theta \right) + \hat{e} \cdot \vec \lop_x \right\| \le 2$, so these terms make a contribution to $\Delta H$ which is uniformly bounded in norm by $2(3L)^d \delta$.

We are left with the terms arising from the transformation of the nonrandom interaction. % Noting that $\hat R^{-1}_\theta = \hat R_{-\theta}$,
For any $X$ and any (arbitrarily chosen) $x \in X \cap \Gamma$,
\begin{equation}
\begin{split}
&\hat R_{-\theta} P_\Gamma(\inter_0(X)) \hat R_\theta
\\ & \; \;= \left( \bigotimes_{y \in X \cap \Gamma} e^{-(\theta_y - \theta_x)\rho_y} e^{-\theta_x \rho_y} \right) P_\Gamma(\inter_0(X)) \left( \bigotimes_{z \in X \cap \Gamma} e^{-\theta_x \rho_z} e^{(\theta_z - \theta_x)\rho_z} \right)
\\ & \;\; = \left( \bigotimes_{y \in X \cap \Gamma} e^{-(\theta_y - \theta_x)\rho_y} \right) P_\Gamma(\inter_0(X)) \left( \bigotimes_{z \in X \cap \Gamma} e^{(\theta_z - \theta_x)\rho_z} \right),
\end{split}
\end{equation}
(using the rotation invariance of $\inter_0$).  Expanding the exponentials, we obtain
\begin{equation}
\begin{split}
\hat R_{-\theta} P_\Gamma(\inter_0(X)) \hat R_\theta = P_\Gamma(\inter_0(X)) &+ \sum_{y \in X \cap \Gamma} (\theta_x - \theta_y)\left(\rho_y  P_\Gamma(\inter_0(X)) - P_\Gamma(\inter_0(X)) \rho_y \right) \\ & + O\left( \frac{(\diam X)^2|X|^2}{L^2}\left\| \inter_0(X) \right\| \right),
\end{split}
\end{equation}
where the estimate of the higher order terms uses
\begin{equation}
|\theta_x - \theta_y| \le \frac{\pi \|x-y\|_\infty}{L} \le \frac{\pi \diam X}{L}
\end{equation}
and the observation that the $n$th order term in the expansion is potentially a sum of $|X|^n$ terms, as well as $\left\| P_\Gamma(\inter_0(X)) \right\| \le \left\| \inter_0(X) \right\| $.  The first order terms are odd in $\theta$, and will cancel in $\Delta H_\theta + \Delta H_{-\theta}$, with the leading term being second order.  What appears there is
\begin{equation}
\begin{split}
\sum_{X \cap \Gamma \ne \emptyset} \left(\hat R_{-\theta} P_\Gamma(\inter_0(X)) \hat R_\theta -  P_\Gamma(\inter_0(X)) \right) = O\left( L^d \sum_{X \ni 0} \frac{1}{|X|} \frac{(\diam X)^2|X|^2}{L^2} \left\| \inter_0(X) \right\| \right) \\ = O(L^{d-2}),
\end{split}
\end{equation}
where the last equality invokes Assumption~\ref{continuous_short_range_assum}.

Then we indeed have Equation~\eqref{F_bounds_for_g}, with
\begin{equation}
\|\Delta H_\theta + \Delta H_{-\theta} \| = O(L^{d-2}) + O(\delta L^d)).
\end{equation}

This provides only an upper bound on $g_{L,M}^{\delta \hat{e}} (\vec\zeta)$, rather than a bound on its absolute value.  However it is obvious from the definition~\eqref{g_mini_def} of $g$ that $g_{L,M}^{\delta \hat{e}} (\vec\zeta) = - g_{L,M}^{-\delta \hat{e}} (\vec\zeta)$, so the needed lower bound follows automatically.
\end{proof}

With Equation~\eqref{G_and_g}, Lemma~\ref{little_g_bound_lemma} means that
\begin{equation}
|G^{\hat{e}}_L(\zeta)| = O(L^{d-2})
\end{equation}
as desired.  In $d \le 4$, this means that for sufficiently large $L$ we have $|G_L| \le A L^{d/2}$, and we use the same moment generating function argument as in the previous section% (noting that the assumptions needed to connect the derivatives of $\F$ to the variance of the limiting distribution are now assumptions on the distribution of the component $\hat{e} \cdot \vec{\eta}$)
, we see that $\F (h \hat{e})$ is differentiable at $h=0$ for all $\hat{e}$.

\chapter{Conclusion}

The previous sections have concluded the proof of the rounding effect for quantum lattice systems; that is, that first order phase transitions (and therefore, in light of Corollary~\ref{LongLongRange} and Theorem~\ref{ShortLongRange}, long range order) are impossible in the presence of direct randomness in low dimensions.  This has been done by establishing a unified analysis of free energy fluctuations applicable to both classical and quantum systems.

At the same time, much remains to be said about the character of the ``rounded'' phase transitions, and of the exceptional cases which have appeared in the course of this work. No simple statement is likely to encapsulate the situation in this context; certainly none can be advanced at this time.  Knowledge of this area continues to grow, and some techniques which may be used to shed further light on it are discussed in Appendix~\ref{numerics.chapter}.

\appendix

\chapter{Methods for numerical studies of random field spin systems} \label{numerics.chapter}

\section{The maximum flow representation of the Ising model ground state}

The RFIM at zero temperature has the considerable virtue that for particular finite field configurations the ground state can be computed easily and exactly thanks to a relationship with the maximum network flow problem.

A maximum flow problem is the following.  We are given an undirected graph (that is, a finite collection of vertices (points), some pairs of which are connected by edges), with two special vertices, the source $s$ and the sink $t$; each edge has a capacity, a finite nonnegative number which we can denote by a symmetric matrix $C_{ij}$ whose indices label the vertices.  A flow is an antisymmetric matrix $F_{ij}$ which does not exceed the capacities ($|F_{ij}| \le C_{ij} \forall i,j$) and which is conserved ($\sum_j F_{ij}=0$) except at the source and the sink.  A maximum flow is one which maximizes the total current from the source to the sink, which is given by $\sum_j F_{sj} \equiv \sum_j F_{jt}$.  This problem has been extensively studied by computer scientists, and there are a number of well-studied and efficient algorithms for solving it.  Most standard implementations (for example the Boost Graph Library~\cite{BGL.book}) assume that the capacities are integers, which we shall see is inconvenient for our purposes, but this can be circumvented by rescaling and using very large integers.

The relationship to the Ising model is through the related minimum cut problem: given the same objects as in the maximum flow problem, a cut is a choice of a division of the vertices into two components, one (call it $S$) containing the source and the other ($T$) the sink.  Each cut has a cost, which is the total of the capacities of all edges which connect $S$ to $T$, $\sum_{i \in S} \sum_{j \in T} C_{ij}$, and a minimum cut is a cut which minimizes this cost function.  Given a ferromagnetic Ising model with Hamiltonian
\begin{equation}
  \ham(\underline{\sigma}) = -\frac{1}{2} \sum_{i,j} \left( 1 - \sigma_i \sigma_j \right) - \sum_i \left( h_i \sigma_i - |h_i| \right),
\end{equation}
we make a graph whose vertices are the sites plus a source and sink, with edges of capacity $J_{ij}$ connecting each interacting pair of sites, an edge of capacity $h_i$ connecting each $i$ with positive field to the source and one with capacity $-h_i$ connecting each site with negative field to the sink.  Then a configuration corresponds naturally to a cut with $S$ being the sites with spin $+1$ and the source; and $\ham$ is precisely the cost of this cut, so a minimum cut corresponds to a ground state~\cite{Picard}.

The ``max cut - min flow theorem''~\cite{FF.MinCut,EFS.MinCut} provides a connection between these two problems.  It states that in a maximum flow, the saturated edges (those with $|F_{ij}|=C_{ij}$) divide the graph in such a way as to provide a minimum cut (or several, if they divide the graph into more than two connected components), and that all minimum cuts for a given problem can be obtained in this way.  The basic idea (also used in the popular push-relabel algorithm to solve the problem~\cite{Goldberg85,Goldberg86}) is that if there is a path of unsaturated edges connecting the source to the sink then it is possible to increase the total current by increasing the flow along each edge of that path.

Together with modern algorithms for solving the max flow problem, this allows the zero temperature random field Ising model to be simulated very efficiently~\cite{HartmannNowak,Sourlas99,WuMachta}, avoiding the extremely slow convergence which plagues monte carlo studies of disordered systems at low temperature.  This has been the main method used for numerical studies of the random field Ising model, although in the last few years histogram reweighting methods like the Wang-Landau algorithm~\cite{Fytas08,WuMachta} have made finite temperature simulations practical as well.

\begin{figure}
\begin{center}
\begin{picture}(300,160)
 \put(60,80){\line(1,0){200}}
 \multiput(60,80)(40,0){6}{\circle*{5}}
 \put(100,80){\line(1,-1){60}}
 \put(160,20){\line(1,1){60}}
 \put(160,20){\line(5,3){100}}
 \put(140,80){\line(1,3){20}}
 \put(160,140){\line(1,-3){20}}
 \qbezier(60,80)(60,140)(160,140)
 \qbezier(160,140)(320,140)(260,80)
 \qbezier(160,20)(320,20)(260,80)
 \put(280,40){$J$?}
 \put(280,120){$J$?}
 \put(70,130){$h$}
 \put(140,110){$h$}
 \put(180,110){$h$}
 \put(120,60){$h$}
 \put(190,60){$h$}
 \put(220,60){$h$}
 \put(160,20){\circle*{5}}
 \put(160,140){\circle*{5}}
 \put(160,150){$s$}
 \put(150,15){$t$}
 \multiput(80,85)(40,0){5}{$J$}
 \multiput(20,80)(15,0){3}{\circle*{2}}
\end{picture}
\end{center}
  \caption{Example of the maximum flow graph used to derive Equation~\eqref{ux_recursion} \label{ux.fig}}
\end{figure}
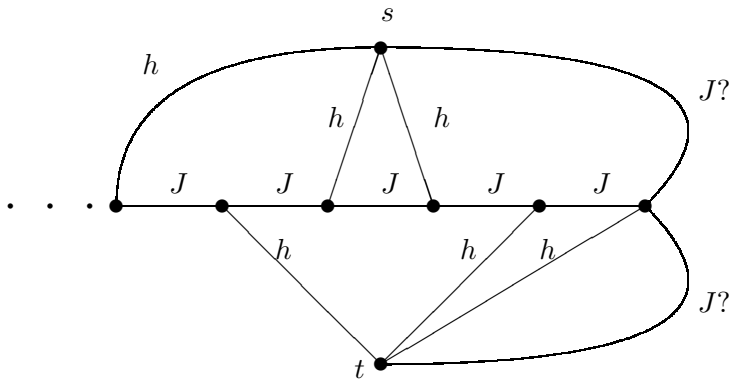

The maximum flow representation also allows the following more general derivation of the exact solution of the one dimensional random field Ising model described in Section~\ref{background.lowDIsing} above.  Construct two different graphs (as shown in Figure~\ref{ux.fig}) by beginning with the construction in the previous section, only for all sites ``before'' $x$; then add an edge with capacity $J$ from site $x-1$ to the sink to make the first graph, and instead connecting to the source to make the second graph.  If the maximum flow for the first graph has nonzero flow through the new edge, take this as $u_x$; otherwise, $u_x$ is minus the flow in the new edge of the second graph.  It is easy to see that this gives a unique set of values which satisfy Equation~\eqref{ux_recursion}, and doing the same with the other half of the system does the same for $v_x$.
%\begin{figure}
%\begin{center}
%\begin{picture}(200,200)
%
%\end{picture}
%\end{center}
%  \caption{Reconstructing the ground state from $u_x$ and $v_x$ by examining flows. \label{ux2gs.fig}}
%\end{figure}
%Now we can obtain a maximum flow for the whole system by combining maximum flows for parts of the system as shown in Figure~\ref{ux2gs.fig}.

We can now obtain the value of $\sigma_x$ as follows.  Suppose $h_x>0$; then we will have $\sigma_x=1$ if there is a maximum flow in which the corresponding edge does not saturate ($F_{sx} < h_x$).  We can obtain a maximum flow for the whole system by pasting together the flow graphs representing the different parts as used above to obtain $u_x$ and $v_x$; then graph can accommodate a flow $F_{sx}$ of up to $-u_x-v_x$, but no larger; so $h_x > -u_x-v_x$ implies $\sigma_x=1$, and $h_x < -u_x-v_x$ implies $\sigma_x=-1$.  The same follows by a similar examination of the cases $h_x <0 $ and $h_x=0$.

\section{Monte Carlo methods for XY and clock models}

\subsection{The lookup table algorithm for the Clock model}

As noted in Section~\ref{3DXY.section}, the clock model has been frequently used as a substitute for the XY model for reasons of computational efficiency. The metropolis algorithm has a high rejection rate at low temperatures.  If one tries to avoid this by using the heat bath algorithm, it appears that one has the choice of sampling the distribution of trial moves with either rejection sampling (which in effect reproduces the same problem) or through an inverse transform method (which is slowed down by the need to evaluate a large number of transcendental functions).  Heat bath updating in the clock model can be implemented much more rapidly by compiling a lookup table, since only a finite number transition probabilities need to be calculated for a given set of parameter values.  A new table must be calculated whenever one changes the random field distribution, coupling constant, or temperature, but can be reused for different configurations of the random field.

To be more precise, if we rewrite the Hamiltonian~\eqref{RFXY_ham} as
\begin{equation}
   \beta \ham = -\sum_x \left(J\sum_{|y-x|=1} \vec{\sigma}_y + \vec{h}_x \right) \cdot \vec{\sigma}_x,
\end{equation}
then range of relative energies involved in rotating a single spin are controlled by an effective field $k_x := J\sum_{|y-x|=1} \vec{\sigma}_y + \vec{h}_x$.  If we denote the number of allowed spin values by $q$, and use a random field taking $nq$ values with the same symmetry, then thanks to the symmetry in interchanging the neighboring spins it takes no more than $nq^{2d}$ values, which are related by a $q$-fold symmetry; for each of these a table of $q-1$ elements needs to be recorded to specify the transition probabilities, giving a table with
\begin{equation}
  T=nq^{2d-1}(q-1)
\end{equation}
elements, usually 32-bit integers (there is a redundancy in this description, and in fact $k_x$ takes no more than $nq\times \binom{2d+q-1}{2d} $ values, but taking full advantage of this complicates the algorithm).  For $q=12$, $d=3$, $n=2$ (the most ambitious case I know to have been implemented~\cite{Fisch10}), this gives a table with about $6.6 \times 10^6$ elements.  Once this table has been calculated, each update step involves only a small number of arithmetical or logical calculations and the generation of a single pseudorandom number.  Assuming that the table is stored in random access memory so that the time required to retrieve a specified element is independent of the table size, we can examine the computing time required by dividing the algorithm used to sample the update distribution into the following steps:
\begin{enumerate}
\item Determine which value of $k$ to use; at most $d+1$ steps of constant complexity
\item Generate a pseudorandom integer $r$
\item Successively look up the probability $p$ of the candidate configurations; if $p\ge r$, choose that configuration; otherwise set $r \to r-p$ and move on.  This is at most $q$ steps of constant complexity.
\end{enumerate}
From this, we can confidently expect that the computer time required for each monte carlo step should not grow faster than linearly in $q$ so long as the table fits in available random access memory.

\subsection{A modified Ziggurat algorithm for the XY model}

I will first describe a slightly modified Ziggurat algorithm for a random variable taking values on $[0,\pi]$ with a decreasing probability density function $p(x)$, before moving on to a further modification to accommodate the situation relevant to the XY model.  In a preparation step, one approximates the graph of $p(x)$ with a collection of $N$ boxes (see Figure~\ref{ziggurat.fig}) of height $h_i$ width $w_i$, and left coordinate $x_i$; there is an easy method for choosing these parameters so that the boxes have equal volume.  To include some boundary cases, we take $x_{N+1}=\pi$ and $h_{N+1}=p(\pi)$.
\begin{figure}
  \includegraphics[keepaspectratio=false,width=\textwidth,height=0.6\textwidth]{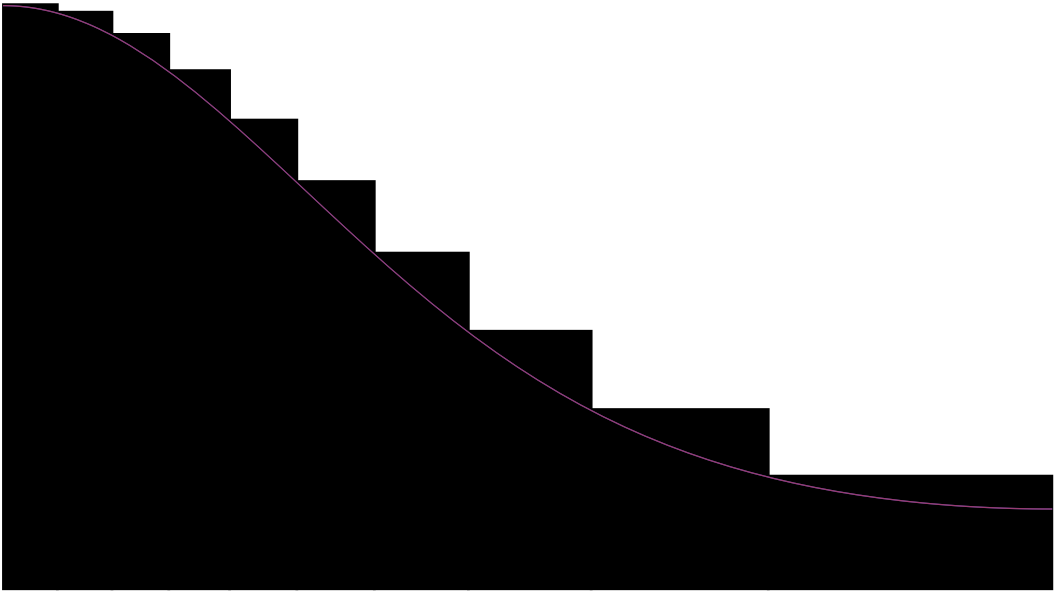}
  \caption{The ziggurat algorithm: a probability density, bounding boxes of equal volume\label{ziggurat.fig}}
\end{figure}

One can then sample the desired distribution in the following way:
\begin{enumerate}
  \item Choose one of the boxes $i$ at random with equal probability
  \item Independently and uniformly choose a random number $x$ from $[x_i,x_{i+1}]$ and $y$ from $[0,h_i]$.  If $y \le h_{i+1}$, return x.
  \item Otherwise, calculate $p(x)$.  If $y \le p(x)$ return $x$, otherwise start over with step 1.
\end{enumerate}
This is a clever way of doing rejection sampling: uniformly select a point in the union of the boxes, accept it if it is under the graph of the desired probability density, otherwise reject it.  When $N$ is reasonably large, to begin with the points generated will be accepted most of the time, and in addition they can usually be accepted without even computing $p(x)$, which can lead to sampling which is even more efficient than an inverse transform method.

In heat bath simulations of the $XY$ model, the key (unnormalized) probability distribution is \begin{equation}\label{pk_def}
  p_k(\theta) = \exp \left( k \cos \theta \right),
\end{equation}
where $k$ is the magnitude of the local field described above, and $\theta$ is the angle between that field and the new spin direction.  This is monotone on $[0,\pi]$ and one can extend it to the full range by adding a step which reflects the spin with probability $0.5$.  The problem with using the Ziggurat method is that the probability distribution depends on a continuous parameter $k$, however it can be generalized in the following manner to accommodate this situation.

\begin{figure}
  \includegraphics[keepaspectratio=false,width=\textwidth,height=0.6\textwidth]{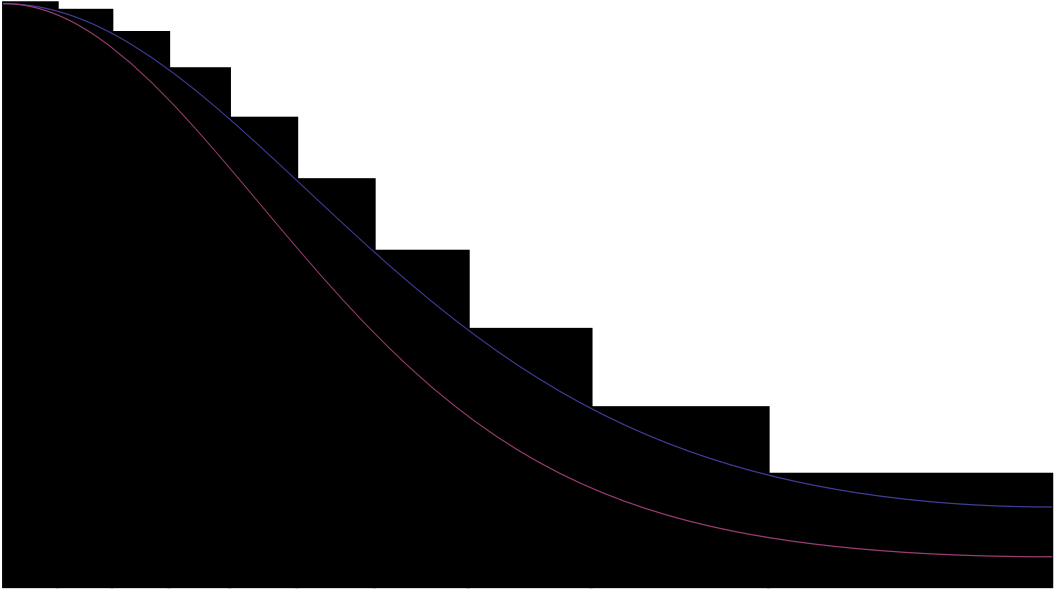}
  \caption{The modified ziggurat algorithm: two probability densities, bounding boxes of equal volume.  The probability densities shown are those of Equation~\eqref{pk_def}, with $k=1$ and $k=1.5$, rescaled for $p_k(0)=1$.\label{ziggurat.fig2}}
\end{figure}
The idea is illustrated in Figure~\ref{ziggurat.fig2}.  For a given range $[k_1,k_2]$, generate a set of boxes as in the original Ziggurat algorithm, but with $\max_{k_1 \le k \le k_2} p_k(x)$.  In addition, for each box we also need to record $t_i=\min_{k_1 \le k \le k_2} p_k(x_{i+1})$.  Then we can sample $p_k$ as follows:
\begin{enumerate}
  \item Choose one of the boxes $i$ at random with equal probability
  \item Independently and uniformly choose a random number $x$ from $[x_i,x_{i+1}]$ and $y$ from $[0,h_i]$.  If $y \le t_i$, return x.
  \item Otherwise, calculate $p_k(x)$.  If $y \le p_k(x)$ return $x$, otherwise start over with step~1.
\end{enumerate}

The efficiency of this method depends on whether we can partition the relevant range of $k$ so that the variation of $p_k$ is small enough that the rejection rate and the frequency with which $p_k$ is calculated do not increase to much.  If we denote the maximum value of the random field by $H$, then $ k$ runs from zero to $ 2 Jd +H$.  $p_k$ is monotone in $k$ (although whether it is increasing or decreasing depends on $x$ and the normalization used) which makes calculating minima and maxima with respect to $k$ very simple.  For the ranges of $k$ relevant to simulating 3-dimensional systems near the apparent critical temperature (roughly $J=2$), it is feasible on a computer with 2GB of RAM to store a table which requires calculating $p_k$ less than one time in 1000.

\chapter{Rounding of First Order Transitions in Low-Dimensional Quantum Systems with Quenched Disorder}\label{QIMLetter}
(With M. Aizenman and J. L. Lebowitz.  Published as~\cite{QIMLetter}.)
\newcommand{\co}{\ensuremath{^c}}
\newcommand{\m}[1]{\ensuremath{m_{#1}(T,h,\epsilon)}}
\newcommand{\Z}{\mathbb Z}
\newcommand{\G}{\widetilde{G}_{\Lambda}}
\newcommand{\Ga}{\widetilde{G}_{\Lambda,\alpha}}
\newcommand{\g}{G_{\Lambda}}
\renewcommand{\gf}{G^\delta_{\Lambda,\Gamma}}
\newcommand{\gfj}{G^\delta_{\Lambda,\Gamma_j}}
\newcommand{\gr}{g_{\Lambda,\Gamma}^{\delta}}
\newcommand{\grj}{g_{\Lambda,\Gamma_j}^{\vec{\delta}}}
\newcommand{\Gf}{\widetilde{G}_{\Lambda,\Gamma,\delta}}
\renewcommand{\lop}{\kappa}
\renewcommand{\Tr}{\textup{Tr }}
\newcommand{\delBy}[2]{\ensuremath{\frac{\partial #1}{\partial #2}}} %\newcommand{\rhoT}[1]{\ensuremath{\tilde{\rho}_{#1}}}
\newcommand{\E}[1]{\textup{Av} \left[#1\right]}
\renewcommand{\state}[1]{\ensuremath{\left\langle #1 \right \rangle}}
\newcommand{\condExp}[2]{\ensuremath{\textup{Av}\left[ #1 \middle| #2 \right] }}
\renewcommand{\co}{\ensuremath{^c}}
\renewcommand{\Zd}{\ensuremath{\mathbb{Z}^d}}
\newcommand{\hstate}[2]{\state{#1}_\Gamma^{\eta,#2}}
\renewcommand{\m}[1]{\ensuremath{m_{#1}(T,h,\epsilon)}}

\section*{Abstract}
We prove that the addition of an arbitrarily small random perturbation to a quantum spin system rounds a first order phase transition in the conjugate order parameter in $d \leq 2$ dimensions, or for cases involving the breaking of a continuous symmetry in $d \leq 4$. This establishes rigorously for quantum systems the existence of the Imry-Ma phenomenon
%\ifthenelse{\boolean{showdetails}}{\cite{YM,AYM}}{},
which for classical systems was proven by Aizenman and Wehr.
%\ifthenelse{\boolean{showdetails}}{\cite{AW,AW.PRL}}{}.
\\
\\
\\
A first order phase transition, in Ehrenfest's  terminology, is one associated with a discontinuity in the density of an extensive quantity.   In  thermodynamic terms this corresponds to a discontinuity in the derivative of the free energy with respect to one of the parameters in the Hamiltonian, more specifically the one conjugate  to the order parameter, e.g.\ the magnetic field in a ferromagnetic spin system.
In what is known as the Imry-Ma phenomenon~\cite{YM,AYM}, any such discontinuity is rounded off in low dimensions when the Hamiltonian of a homogeneous system is modified through the  incorporation of an arbitrarily weak random term, corresponding to quenched local disorder,  in the field conjugate to the order parameter.

This phenomenon has been rigorously established for classical systems~\cite{AW.PRL,AW.CMP}, where it occurs  in  dimensions $d\le 2$, and $d\le 4$ when  the discontinuity is associated with the breaking of a continuous symmetry.   In this letter we prove analogous results for quantum systems at both positive and zero temperatures (ground states).

The existence of this effect was first argued for random fields by Imry and Ma  on the basis of a heuristic analysis of free energy fluctuations.  While the sufficiency of Imry and Ma's reasoning was called into question, the predicted phenomenon was established rigorously through a number of works~\cite{FFS,Imbrie.PRL,Imbrie.CMP,BK.PRL,BK.CMP,AW.PRL,AW.CMP}.  The statement was further extended to different disorder types by Hui and Berker~\cite{HuiBerker.PRL,Berker.PRB.90}.

%Establishing the precise conditions when the Imry-Ma phenomenon occurs in quantum systems is an important open question~\cite{senthil1998prf,Goswami}.   As stressed in~\cite{..}, first order quantum phase transitions (QPT$_1$) attract a great deal of current interest

The general existence of the Imry-Ma phenomenon in quantum systems was not addressed by these rigorous analysis, and in particular the Aizenman-Wehr~\cite{AW.PRL,AW.CMP} proof of the rounding effect applies only for classical systems.   However, as  stressed in~\cite{Goswami},  establishing whether  the Imry-Ma phenomenon extends to first order quantum phase transitions (QPT$_1$)  is an important open problem.   The results presented here answer this question.
We find that the critical dimensions for the phenomenon for  quantum systems are the same as for classical systems, including at zero temperature.
%,
% at both positive and zero temperatures.

%
% In various respects quantum systems are related to classical systems in $d+1$ dimension, and one could wonder whether for the quantum ground state  the critical dimensions are lower.   We show that they are not.

We consider spin systems on the $d$-dimensional lattice $\mathbb{Z}^d$, where the configuration at each site is described by a finite-dimensional Hilbert space, with a Hamiltonian of the form
\begin{equation}
\label{Hamiltonian}
\mathcal{H} = \mathcal{H}_0 - \sum_x \left( h + \epsilon \eta_{x} \right) \lop_{x} \end{equation}
where $\{\lop_x\}$ are translates of some local operator $\lop_0$,
 and $h$ and $\epsilon$ are real parameters.  The quenched disorder is represented by $\{\eta_x\}$, a family of independent, identically distributed random variables.  $\mathcal{H}_0$ may be translation invariant and nonrandom, or it can include additional random terms (although we will not discuss the latter case, our results hold there also).
For convenience we will assume that $\|\lop_x\| = 1$, which can be arranged by rescaling $h$ and $\epsilon$.
We will refer to the $\eta$s as random fields, although in general they may also be associated with some other parameters, e.g.\  random bond strengths.

An example of a system of this type (with $\lop_x=\sigma_x^{(3)}$) is the ferromagnetic transverse-field Ising model with a random longitudinal field~\cite{Senthil} (henceforth QRFIM), with
 \begin{equation} \label{thisHam}
\mathcal{H} = -\sum J_{x-y} \sigma_x^{(3)} \sigma_y^{(3)} - \sum \left[ \lambda \, \sigma_x^{(1)} + (h+ \epsilon \eta_x) \, \sigma_x^{(3)} \right] \end{equation} where $\sigma_x^{(i)}$ $(i=1,2,3)$ are single-site Pauli matrices, and $J_{x-y} > 0$.
The QRFIM~has recently been studied as a model for the behavior of ${\rm LiHo_xY_{1-x}F_4}$ with $x>0.5$ in a strong transverse magnetic field~\cite{Schechter,SchSt.PRB}.

We will examine phase transitions where the order parameter is the volume average of the expectation value of $\lop_x $ with respect to an equilibrium (KMS) state, and show that this quantity cannot be discontinuous in $h$ for low-dimensional systems.  As is well known,
this order parameter is related to the directional derivatives $(\pm)$ of the \emph{free energy density},
\begin{equation}
\m{\pm} \ := \ -\frac{\partial} {\partial h \pm} \F(T, h, \epsilon) \, \label{mDef}
\end{equation}
where, as usual,  at positive temperatures
\begin{equation} \label{Z}
\F(T,h, \epsilon) = \ \lim_{\Gamma \nearrow \Zd}
\frac{-1}{\beta |\Gamma|} \log \Tr e^{-\beta H_{\Gamma}} \, \end{equation}
(with $\beta := 1/k_B T$), and $\F(0,h,\epsilon)$ is the corresponding limit of the ground state energy.
Here $H_{\Gamma}$ is the Hamiltonian of the system restricted to the finite box $\Gamma \subset \Zd$, and $ |\Gamma|$ is the number of sites in that box. It is known under the assumptions enumerated below that for almost all $\eta$ this limit exists and is given by a non-random function of the parameters (see, e.g.~\cite{AW.CMP,vuillermot1977tqr}), which does not depend on the boundary conditions.
By general arguments which are valid for both classical and quantum systems, $\F$ is convex in $h$; therefore the directional derivatives exist, and are equal for all but countably many values of $h$ \cite{Ruelle}.

For typical realizations of the random field, the interval $[\m{-}, \m{+}]$ provides the asymptotic range of values of the order parameter for any sequence of finite volume Gibbs states
 or ground states
(the argument is similar to that found in \cite{AW.CMP} for classical systems).  At a first order phase transition $m_- < m_+$, and there are then at least two distinct infinite volume KMS states~\cite{Ruelle} with different values of the order parameter.
In the QRFIM~the $m_+$ and the $m_-$states  can be obtained through the $+$ or $-$ boundary conditions (i.e.\ the spins $\sigma^{(3)}_x$ are replaced by $\pm 1$ for all $x \notin \Gamma$).  In general, such states  are  obtained by  adding $\pm \delta$ to the uniform field $h$ and letting  $\delta \to 0$ \emph{after} taking the infinite volume limit.

Our discussion is restricted to systems satisfying:
\begin{enumerate}[A.]
\item \label{finiteRange}
The interactions are \emph{short range}, in the sense that
for any finite box $\Lambda \in \Zd$ the Hamiltonian may be decomposed as: $\ham = H_\Lambda+V_\Lambda+H_{\Lambda\co}$, with $H_\Lambda$ acting only in $\Lambda$,  $H_{\Lambda\co} $  only in the complement  $\Lambda\co$, and $V_\Lambda$ of norm bounded by the size of the boundary:
 \begin{equation}
\| V_\Lambda \| \leq C | \partial \Lambda |  \label{finiteRangeIneq} \, .
\end{equation}
\item \label{fluct} The variables $\eta_x$ have an \emph{absolutely continuous}  distribution with respect to the Lebesgue measure  (i.e.\ one  with a probability density with no delta functions), and a  finite  \emph{ $r$th moment}, for some $r>2$.
 \end{enumerate}

Our main results are summarized in the following two statements.  The first applies regardless
of whether the order parameter is related to any symmetry breaking.

\begin{theorem}\label{bigThm}
In dimensions $d\le 2$, any system of the form of \eqref{Hamiltonian} satisfying the above assumptions has $\m{+}=\m{-}$ for all $h$, and $T\ge 0$, provided $\epsilon \neq 0$.    \end{theorem}

The next result is formulated for situations where the the first order phase transition would represent continuous symmetry breaking.  An example is the  $O(N)$ model
with
\begin{equation} \label{HeisHam}
 \mathcal{H}_0 \ =\ -  \sum J_{x-y} \vec{\sigma}_x \cdot \vec{\sigma}_y
\end{equation}
where $\vec{\sigma}$ are the usual quantum spin operators.    More generally, $ \mathcal{H}_0 $ is assumed to be a sum of finite range terms which are invariant under the global action of the rotation group $SO(N)$, and $\vec{\sigma}_x$ is a collection of operators of norm
one which transform as  the components of a vector under rotations.
With the random terms the Hamiltonian is
\begin{equation}
\mathcal{H} =  \mathcal{H}_0  - \sum (\vec{h} + \epsilon \vec{\eta}_x) \cdot \vec{\sigma}_x. \label{ONHam} \end{equation}

\begin{theorem}\label{bigThm2}
For the $SO(N)$-symmetric system described above, with the random fields $\vec{\eta}_x$ having a rotation-invariant distribution, the free energy
is continuously differentiable in $\vec{h}$ at $\vec{h}=0$ whenever $\epsilon \neq 0$, $d \leq 4$, and $N \geq 2$.
\end{theorem}

Before describing the proof, let us comment on   the implications of the statements, and their limitations.
%, and extensions.
\newcounter{ccount}
\begin{list}{\arabic{ccount}. }{\usecounter{ccount} \setlength{\itemindent}{\leftmargin} \setlength{\listparindent}{\parindent} \setlength{\leftmargin}{0pt}}
\item While the statements establish uniqueness of the expectation value of the bulk averages of the observables $\kappa_x$, or $\vec{\sigma}_x$ (in Theorem~\ref{bigThm2}), they do not rule out the possibility of the coexistence of a number of equilibrium states, which differ from each other in some other way than the mean density of $\kappa$, which they share.
More can be said  for models for which it is known by other means   that non-uniquess of state is possible only if there is long range order in $\kappa$.   (Such is the case for  QRFIM,  through its relation to the classical ferromagnetic Ising model in $d+1$ dimensions~\cite{campanino1991lgs}.)

\item The results address only the discontinuity,  or symmetry breaking (as in the QRFIM), but they leave room for other phase transitions, or singular dependence on $h$.    For instance, for  the Ashkin-Teller spin chain for which Goswami et. al.~\cite{Goswami}  report finding the Imry-Ma phenomenon in some range of the parameters but not elsewhere,  the results presented here rule out the persistence of a first-order transition between the paramagnetic and Baxter phases in the full range in the model's parameters.   However, they do not rule out the possibility of  other phase transitions.

\item Randomness which does not couple to the order parameter of the transition need not cause a rounding effect.  For example, in the transverse-field Ising model in a random \emph{transverse} field, where the random field $\eta_x$ couples to $\sigma^{(1)}$, ferromagnetic ordering is known to persist~\cite{Fisher,campanino1991lgs}.  Presumably the same is true for the Baxter phase of the Ashkin-Teller model.  It was  even suggested that there are systems in which the introduction of randomness of this sort may even induce long range order which would not otherwise be present~\cite{Wehr.PRB,Wehr.PRL}, and our results do not contradict this.  In addition we can draw no conclusions about quasi-long-range order, that is power law decay of correlations, including of $\kappa_x$.

\end{list}

Other comments, on the technical assumptions under which the statements hold, are found after the proofs. \\

The proofs of Theorems~\ref{bigThm} and~\ref{bigThm2} are based on the analysis of the differences, between the $m_+$ and the $m_-$-states, in the free energy (at $T=0$, ground state energy) which can be ascribed   to the random field within a finite region $\Lambda$ of diameter $L$.
Putting momentarily aside  the question of existence of limits, a relevant quantity could be provided by:
\begin{equation} \label{gLim}
\G(\eta_\Lambda) := \lim_{\delta \to 0}\lim_{\Gamma \to \Z^d} \condExp{\gf}{\eta_\Lambda} -\E{\gf}
\end{equation}
where
$\gf$ is the difference of free energies
\begin{equation}
\gf(\eta):= \frac{1}{2} \left( F^{\eta,h+\delta}_{\Gamma} - F^{\eta^{(\Lambda)},h+\delta}_{\Gamma}-F^{\eta,h-\delta}_{\Gamma} + F^{\eta^{(\Lambda)},h-\delta}_{\Gamma}\right), \label{gfDef}
\end{equation}
with
\begin{equation}
F^{\eta,h}_{\Gamma}:= \frac{-1}{\beta}\log \Tr \exp (-\beta H^{\eta,h}_{\Gamma}) \, ,
\end{equation}
 $\eta^{(\Lambda)}$ is the random field configuration obtained from $\eta$ by setting it to zero within $\Lambda$, and  $\condExp{\cdot}{\eta_\Lambda}$ is a conditional expectation, i.e.\ an average over the fields outside of $\Lambda$.
 (The modification of the field $h$ by $\pm\delta$ serves to select the desired ($m_\pm$) states).

 Somewhat inconveniently, it is not obvious that  for all models the limits in \eqref{gfDef}
exist.  Nevertheless, one  can prove that for each system of the class considered here there is a sequence of volumes $\Gamma_j \nearrow \Z^d$  for which the limit exists for all $\Lambda$, with convergence uniform in $\eta_\Lambda$.    The proof of this assertion is by a compactness argument, whose details can be found elsewhere~\cite{future}\footnote{Sections~\ref{convergence.section.1} and~\ref{convergence.section.2}.}.

The essence of the proof of  Theorem~\ref{bigThm} is the  contradiction between two estimates:
\begin{enumerate}[i.]
\item
Under Assumption~\ref{finiteRange}, equation~\eqref{finiteRangeIneq}:
\begin{equation}
|\G(\eta)| \leq  4 C | \partial \Lambda | \, . \label{uBound}
\end{equation}
\item % [ii.]
Whenever $m_-< m_+$,
 $\G/\sqrt{|\Lambda|}$ converges in distribution  to a normal  random variable with a positive variance
   (as one would guess by considering the difference in the random field terms between states of  different mean  magnetizations,  neglecting the states' local adjustments to the random fields).
\end{enumerate}

More explicitly,
for the upper bound we note that in the absence of the interaction terms $V_\Lambda$, the right hand side of~\eqref{gfDef} would be zero.  Using~\eqref{finiteRangeIneq}, one gets \eqref{uBound}.

To prove the normal distribution for $\gf$, we apply a theorem of %\citep[Proposition~6.1]{AW.CMP}
\cite{AW.CMP} (Proposition~6.1)
 (as corrected in %\citep[p.~124]{bovier2006smd}
\cite{Bovier}, p. 124).  It implies that for $\Lambda \nearrow \Zd$, under Assumption~\ref{fluct},  $\G/\sqrt{|\Lambda|}$ converges in distribution  to a normal  random variable with variance of the order of
\begin{equation}
b\ =\ \E{\frac{\partial\G}{\partial \eta_x}} = m_+ - m_-.
\end{equation}

 The two statements described above contradict the assumption  that $m_-<  m_+$  in dimensions $d\le2$.   That is so even at the critical dimension,  where $L^{d/2} = L^{d-1}$.  The reason is that the lower bound
implies the existence of arbitrarily large fluctuations on that scale, whereas the upper bound is with a uniform constant.  This proves Theorem~\ref{bigThm}.\\

%
%For the case covered by Theorem~\ref{bigThm2} rigorous soft-mode deformation analysis shows that the second estimate can be replaced by  $C L^{d-2}$.  This  raises  the critical dimension from $d=2$ to  $d=4$.

The above proof is similar to that of the classical results~\cite{AW.CMP,AW.PRL} which this work extends.  However, the discussion  of the free energy fluctuations was based there on the analysis of the Gibbs states, and more specifically of the response to the fluctuating fields of the `metastates' which were  specially constructed for that purpose.
Except for special cases, such as the QRFIM, that argument was not available for quantum systems, where  the equilibrium  expectation values are no longer given by integrals over positive measures.
 The proof of the quantum case is enabled by a more direct analysis of the free energy. \\

Theorem~\ref{bigThm2} is proven by establishing that in the presence of continuous symmetry the   upper bound  \eqref{uBound}, for $\Lambda = [-L,L]^d$,  can be replaced by:
\begin{equation}
|\G(\eta_\Lambda)| \leq K L^{d-2}\label{symUBound} \, .
\end{equation}
Here $\G$  is defined as in~\eqref{gfDef},\eqref{gLim}, but $\vec{h}=\vec{0}$, and $\delta$ is replaced by by $\vec{\delta} := \delta \hat{e}$ with $\hat{e}$ a unit vector.   This change in the upper bound raises  the critical dimension to $d=4$.

To obtain  \eqref{symUBound} we focus on % the quantity
\begin{equation}\label{grDef}
\gr(\vec{\eta}_\Lambda) := \condExp{F_\Gamma^{\vec{\eta},\delta \hat{e}} - F_\Gamma^{\vec{\eta},-\delta \hat{e}}}{\vec{\eta}_\Lambda} \, .
\end{equation}
Since
%\begin{equation}
$
\G(\eta_\Lambda) = \lim_{\delta \to 0}\lim_{\Gamma \to \Z^d}
\frac{1}{2} \left( \gr(\vec{\eta}_\Lambda)  - \gr(0)    \right) \, ,
%\end{equation}
$
any uniform bound on  $|\gr|$ for given $\Lambda$ implies a similar bound on $|\G|$.  The claimed  bound may be obtained though a soft-mode deformation analysis, which we shall make explicit for the case of pair interaction (the general case can be treated by similar estimates).

The free enrgy $F_\Gamma^{\vec{\eta},-\delta \hat{e}}$ in \eqref{grDef} may be rewritten
% the above expression
by rotating both the spins and the field vectors  with respect to an axis perpendicular to $\hat{e}$ at the slowly varying  angles
\begin{equation}
\theta_x :=\ \left\{
\begin{array}{ll}
0, & \; \| x \| \leq L \\
\frac{\|x\|-L}{L}\pi, & L < \|x \| < 2L \\ \pi, & \|x\| \geq 2L
\end{array}
\right.  \, .
\end{equation}
The rotation aligns the external fields in the two terms ($\pm \delta \hat{e}$), except within $\Lambda$ where the effect is negligible when $\delta \to 0$.
%Comparing the resulting Hamiltonian to that of $F_\Gamma^{\vec{\eta},\delta \hat{e}}$ in the original basis, we see that
The effect of the rotation on the random fields is absorbed by rotation invariance of the average.  In the end, the Hamiltonian of the rotated system differs from the Hamiltonian used to define the other free energy by
\begin{equation}\label{twist}
\begin{split}
\Delta H_{\underline{\theta}} := \sum_{\{x,y\} \subset \Gamma} & J_{x-y} \left[ \vec{\sigma}_x \cdot \vec{\sigma}_y \right. \\
- &\left. \vec{\sigma}_x \cdot \left( e^{i (\theta_y-\theta_x) \rho_y} \vec{\sigma}_y e^{- i (\theta_y-\theta_x) \rho_y} \right)
\right]
\end{split}\end{equation}

%
%Thus after the rotation, the Hamiltonians defining the two free energies in~\eqref{grDef} differ consequently only in the bond terms; in the case of pair interactions (by way of illustration) the difference is
%\begin{equation}\label{twist}
%\begin{split}
%\Delta H_{\underline{\theta}} := \sum_{\{x,y\} \subset \Gamma} & J_{x-y} \left[ \vec{\sigma}_x \cdot \vec{\sigma}_y \right. \\
%- &\left. \vec{\sigma}_x \cdot \left( e^{i (\theta_y-\theta_x) \rho_y} \vec{\sigma}_y e^{- i (\theta_y-\theta_x) \rho_y} \right)
%\right]
%\end{split}\end{equation}
% with $\rho_x$ the infinitesimal generator of  rotation at $x$.

When the resulting expression for $F_\Gamma^{\vec{\eta},-\delta \hat{e}}$ in  \eqref{grDef} is expanded in powers of  $\theta_x-\theta_y \approx \pi \|x-y\|/L$,  the zeroth-order term cancels with $F_\Gamma^{\vec{\eta},-\delta \hat{e}}$, and the second and higher order terms yield the claimed bound.  The main difficulty is to eliminate the first order terms, which amount to a sum of $O(L^d)$ quantities each of order $1/L$.
However, the sign of these terms is reversed when the rotation is  in the reversed direction.
To take advantage of this, we combine two expressions for
$\gr(\vec{\eta}_\Lambda) $ with the rotations applied  in opposite directions, yielding:
\begin{equation}
\begin{split}
\gr(\vec{\eta}_\Lambda) &= \rm{Av}
\left[ \log \Tr e^{-\beta H} \right. \\
- \tfrac{1}{2} & \left. \log \Tr e^{-\beta (H + \Delta H_{\underline{\theta}})}
- \tfrac{1}{2}\log \Tr e^{-\beta (H + \Delta H_{-\underline{\theta}})} \middle|
\vec{\eta}_\Lambda \right]
\end{split}
\end{equation}
(where $H \equiv H_\Gamma^{\vec{\eta},\delta \hat{e}}$.)
By known operator inequalities \cite{Ruelle}:
\begin{equation}
\begin{split}
\log \Tr e^{-\beta H} - &\tfrac{1}{2} \log \Tr e^{-\beta (H + \Delta H_{\underline{\theta}})} - \tfrac{1}{2} \log \Tr e^{-\beta (H + \Delta H_{-\underline{\theta}})}
 \\ &\leq \tfrac{1}{2} \| \Delta H_{\underline{\theta}} + \Delta H_{-\underline{\theta}} \|
\end{split}
%\ln \Tr e^A - \tfrac{1}{2} \ln \Tr e^B - \tfrac{1}{2} \ln \Tr e^C
%\leq \ln \Tr e^A - \ln \Tr e^{B/2}e^{C/2}
%\leq \ln \Tr e^A - \ln \Tr e^{(B+C)/2}
%\leq \|A - \tfrac{1}{2}(B+C)\|  \, .
\end{equation}
The right hand side is zero to first order, and one is left with an \emph{upper bound} on~$\gr$ of the desired form.  Repeating this analysis with the roles of the terms exchanged we obtain an identical \emph{lower bound},  and thus  inequality~\eqref{symUBound} follows.\\
% We use this to obtain a lower bound on $F_\Gamma^{\vec{\eta},\delta \hat{e}} - F_\Gamma^{\vec{\eta},\delta \hat{e}}$ in which the leading order terms of~\eqref{twist} cancel, and by repeating the same procedure with the roles of the terms exchanged we obtain an identical lower bound, and this is easily turned into Inequality~\eqref{symUBound}.

%%%%%%

The  above argument is spelled out in detail in \cite{future}\footnote{See also Section~\ref{continuous_wrapup}}.   Let us end with few additional comments on the assumptions.
\newcounter{dcount}
\setcounter{dcount}{\value{ccount}}
\begin{list}{\arabic{ccount}. }{\usecounter{ccount} \setlength{\itemindent}{\leftmargin} \setlength{\listparindent}{\parindent} \setlength{\leftmargin}{0pt}}\setcounter{ccount}{\value{dcount}}
\item For Theorem~\ref{bigThm2}, the assumption that the interaction has a strictly finite range can be weakened to a condition somewhat similar to Assumption~\ref{finiteRange}.  For pair interactions (equation~\eqref{ONHam}) it suffices to assume: \begin{equation}
\sum_{x \in \Zd} |J_x| \ \| x \|^2 <\infty  \label{continuous_short} \, .
\end{equation}
%where $\| x \|$ is the largest-component norm.

%% (This is satisfied by an inverse cube interaction in $d=2$, but not $d \geq 3$.)
%%
%%\item %(Distribution of the random field)
%%The assumption of an absolutely continuous distribution excludes systems where the random parameter takes only a finite number of values, which includes a number of models of interest.  Such an assumption is generally necessary at zero temperature as can be seen by the behavior of the Ising chain in such a random field~\cite{Bleher}, but at finite temperature it can be replaced by the assumtion that the distribution has a continuous part which extends along the entire range of values.  For the QRFIM at finite temperature one need only assume that the random field has more than one possible value (otherwise there is no randomness), and this could be the case more generally.  Some form of the finite moment assumption is probably necessary for the system to even be well-defined, but it is hard to imagine it being violated for any system of physical interest.
%%  &

 \item %(Distribution of the random field)
 The restriction to absolutely continuous distribution excludes a number of models of interest.  Such an assumption is generally necessary at zero temperature, as can be seen by the behavior of the Ising chain in a random field~\cite{Bleher} which takes only a finite number of values.  For positive temperatures it can be replaced by the requirement that the distribution has a continuous part which extends along the entire range of values.  For the QRFIM at finite temperature one need only assume that the random field has more than one possible value, and this may well be the case more generally.

\end{list}

\section*{Acknowledgments}
We thank S. Chakravarty and M. Schechter for useful discussions.   MA wishes to thank for the hospitality accorded him at the  Weizmann Institute, at the  Department of Physics of Complex Systems.
The work of R.L.G. and J.L.L. was supported in part by the
NSF Grant DMR-044-2066 and AFOSR Grant AF-FA9550-04, and the work of M.A. by the NSF Grant DMS-060-2360.
%The work of R.L.G. and J.L.L. was supported by NSF Grant DMR-044-2066 and AFOSR Grant AF-FA9550-04, that of M.A. by NSF Grant DMS-060-2360.

\chapter{Some concepts and results in mathematical probability}\label{prob_appendix}

In this appendix, I will review some of the probabilistic terminology and concepts used in this dissertation for the benefit of those readers who may find them obscure. Proofs and further details can be found in standard probability textbooks such as~\cite{Billingsley_probability} or~\cite{Klenke}.

\section{$\sigma$-algebras and measures}

We begin with some set $X$, for example the set of outcomes of a class of measurements.  Assigning probabilities to members of $X$ is a problem when $X$ is uncountably infinite (e.g.\ when it is a real interval), so we will instead assign probabilities to subsets of $X$ (events).  It turns out to be impossible to consistently do this for all subsets of $X$ whatsoever, but we at least want to be able to look at complements of sets (corresponding to logical not) and countable unions (corresponding to logical and), intersections (logical or) then comes for free. A collection of $\mathcal{X}$ of subsets of $X$ which includes $X$ and is closed under complement and countable union is called a \emph{$\sigma$-algebra}.  The most common example is the Borel algebra on the real line, which is the smallest $\sigma$-algebra which includes all intervals.  The combination of a set and a $\sigma$-algebra on that set is called a \emph{measurable space}, and in this context an element of the $\sigma$-algebra  is called a \emph{measurable set}.

A function $\mu$ assigning a nonnegative real number (possibly $\infty$) to each member of a $\sigma$-algebra is called a \emph{measure} if $\mu(\emptyset)=0$ and
\begin{equation}
  \mu\left( \bigcup_{i=1}^\infty A_i \right)= \sum_{i=1}^\infty \mu(A_i)
\end{equation}
for any countable collection of disjoint (nonoverlapping) sets $A_i \in \mathcal{X}$. $\mu$ is called a \emph{probability measure} if $\mu(X)=1$.  This definition immediately implies a number of the usual rules of probability: probability of mutually exclusive events is additive, the probabilities of a complete set of mutual exclusive events is 1, etc.  A measurable space with an associated measure is a \emph{measure space}, or if that measure is a probability measure it is a \emph{probability space}.

One measure we use frequently is the Lebesgue measure $\lambda$, also known as the uniform measure on the real line.  This has the property that for any finite interval $\lambda(I)$ is the length of $I$, and it is the only measure on the Borel algebra with this property.  Another important Borel measure is the Dirac measure, defined by
\begin{equation}
\delta_x(A) := \piecewise{
   1, & x \in A \\
  0, & x \notin A}.
\end{equation}

It is possible to define integration with respect to a measure, but only with respect to \emph{measurable functions}.  A function $f$ from one measurable space to another is measurable if the preimage of a measurable set is a measurable set (that is, if the set of $x$ which produces some measurable range of outcomes $f(x)$ is guaranteed to be measurable).  This means that knowing which measurable sets contain $x$ is enough information to deduce the value of $f(x)$.  Continuous functions are always measurable with respect to the Borel algebra.  Given a measure, it is possible to define integration of a function $f$ with respect to a measure $mu$, denoted
\begin{equation}
  \int f d\mu
\end{equation}
first for positive measurable functions, then for functions whose positive and negative parts give a finite integral (these are then \emph{integrable functions}).  This has many of the usual properties of an integral.  We can also define integrals over a measurable set $A$, which have the property that
\begin{equation}
  \int_A d \mu = \mu(A).
\end{equation}
We can also write this in terms of the \emph{indicator function}
\begin{equation}
  I[A](x):=\piecewise{
  1, &x \in A \\
  0, & x \notin A
  }
\end{equation}
as
\begin{equation}
  \int_A f d\mu = \int f I[A] d\mu
\end{equation}

Integration with respect to the Lebesgue measure (leaving aside some issues of divergence) coincides with the usual notion of integration on the real line:
\begin{gather}
  \int f d\lambda = \int_{-\infty}^\infty f(x) dx \\
  \int_{[a,b]} f d\lambda = \int_a^b f(x) dx
\end{gather}
whenever the relevant expressions are all well defined.  A set $A$ which is contained in a measurable set $B$ with $\mu(B)=0$ is called a \emph{null set}\footnote{This is not precisely the same as saying that $A$ has measure zero, since $A$ is not assumed to be measurable; however this distinction is rarely important.}; two functions which differ only on a null set have the same integrals.

Given a measure $\mu$ and a nonnegative measurable function $f$, we can define a measure $f \mu$ by
\begin{equation}
  f\mu(A) = \int_A f d \mu.
\end{equation}
If it is possible to write a measure $\nu$ as $ f \mu$, then we saw that $\nu$ is \emph{absolutely continuous} with respect to $\mu$; $f$ is the \emph{Radon-Nikodym derivative} or \emph{density} of $\nu$ with respect to $\mu$.  When I say that a measure is absolutely continuous without specifying another measure, this should be understood to be the Lebesgue measure.

\section{$L^p$ norms and spaces; convergence of measurable functions}

Consider some measure space $(X,\mathcal{X},\mu)$, and the following functional on the measurable functions ($1 \le p \le \infty$):
\begin{equation}
  \pnorm{f}{p} := \left( \int |f|^p d\mu \right)^{1/p}
\end{equation}
This is not quite a norm, since it is unaffected by changing the value of $f$ on a null set.  It is however a norm on the resulting equivalence classes, and so it is frequently referred to as the $L^p$ norm; the space of equivalence classes of functions with $\pnorm{f}{p}$ finite is denoted $\mathcal{L}^p(\mu)$.  This can be extended to $p=\infty$ by setting
\begin{equation}
  \pnorm{f}{\infty} = \inf_{A:\mu(X\setminus A)=0} \sup_{x \in A} |f(x)|.
\end{equation}
For any $1 \le p \le \infty$, the $L^p$ norm provides a notion of convergence like any other norm.  We say that a sequence $f_n$ converges to $f$ in $\Lp{p}{\mu}$ iff
\begin{equation}
  \lim_{n \to \infty} \pnorm{f_n-f}{p} =0.
\end{equation}

One useful property of these norms is \emph{H\:older's inequality}, which states that
\begin{equation}
  \pnorm{fg}{1} \le \pnorm{f}{p} \pnorm{g}{q}
\end{equation}
whenever $1/p + 1/q = 1$ (including $p=1$,$q=\infty$).

There are several other important notions of convergence for measurable functions.  The first is \emph{convergence in measure}: $f_n \to f$ in measure if
\begin{equation}
  \mu \left( \left\{ |f-fn| > \epsilon \right\} \cap A \right) \to 0
\end{equation}
for any $\epsilon > 0$ and any measurable set $A$ with $\mu(A) \le \infty$.

A sequence $f_n$ converges to $f$ \emph{almost everywhere} (in a probability space, \emph{almost certainly} or \emph{almost surely}) if there is a null set $N$ so that
\begin{equation}
  f_n(x) \to f(x) \forall x \notin N.
\end{equation}
Almost everywhere convergence and $L^p$ convergence both imply convergence in measure.

The following results give some relationships between convergence of functions and of integrals.

\begin{sloppypar}
\begin{theorem}[Beppo-Levi]
Let $f_n$ be a sequence of integrable functions with $\int f_n d\mu~<~\infty$, with $f_n \nearrow f$ (monotone convergence) almost everywhere, with $f$ measurable.  Then
\begin{equation}
  \lim_{n \to \infty} \int f_n d\mu = \int f d\mu.
\end{equation}
\end{theorem}

\begin{theorem}[Lebesgue, dominated convergence]\label{domConvThm}
Let $f_n$ be a sequence in $L^1(\mu)$ with $f_n \to f$ in measure for some measurable $f$.  Assume there is a nonnegative $g~\in~L^1(\mu)$ such that $|f_n|\le g$ almost everywhere for all $n$.

Then $f \in L^1(\mu)$ and $f_n \to f$ in $L^1$; in particular
\begin{equation}
  \int f_n d\mu \to \int f d\mu.
\end{equation}
\end{theorem}
\end{sloppypar}

We can also talk about limits of sequences of measures on a metric space (e.g.\ $\R$).  The most common notion is the following: a sequence of finite measures $\mu_n$ on the same measurable space \emph{converges weakly} to $\mu$ if
\begin{equation}
  \int f d\mu_n \to \int f d\mu
\end{equation}
for all bounded, continuous functions $f$.

\section{Random variables, expectations, conditional expectations}

An integrable function on a probability space is called a \emph{random variable}.  The average (or expectation) of a random variable $\phi$ is
\begin{equation}
  \Av \phi = \int \phi d\mu.
\end{equation}
The expectation satisfies \emph{Chebyshev's inequality},
\begin{equation}
\Av I[ |\phi- \Av \phi | \ge \epsilon] \le \epsilon ^{-2} \left( \Av (\phi^2) - (\Av \phi)^2 \right).
\end{equation}

It is possible to say a great deal about a random variable without specifying which probability space it is defined on.  The important fact is that a random variable $\phi$ defines a probability measure, called its \emph{distribution} and denoted by $P_\phi$, on the space in which it takes its values.  This is exactly the probability that $\phi$ will take a value in a particular measurable set.  We can talk about multiple random variables being \emph{identically distributed} if their distributions are the same.  Random variables \emph{converge in distribution} if their distributions converge weakly.  One important distribution is the \emph{normal distribution} with mean $m$ and variance $b^2$, denoted $N(m,b^2)$, which is absolutely continuous with respect to the Lebesgue measure and has density $e^{-x^2/2b^2}/(b \sqrt{2 \pi})$.

Among the important properties of a random variable are its \emph{moments}.  The $k$th moment is $\Av f^k$.  A bounded random variable is characterized by its (integer) moments, or by its exponential moment generating functional $\Av e^{tf}$.

Random variables $\phi$ and $\psi$ defined on the same probability space are \emph{independent} or \emph{independently distributed} if, for any measurable $A$ and $B$, the probability that $\phi \in A$ and $\psi \in B$ is the product of the two separate probabilities.  A similar definition holds for arbitrary collections of random variables.  We frequently speak of collections of random variables being \emph{independently and identically distributed}, abbreviated i.i.d..

Given a random variable $\phi$ and a sub-$\sigma$-algebra $\mathcal{F}$ of the one for which it was defined, the \emph{conditional expectation} of $\phi$ with respect to $\mathcal{F}$ (written $\condAv{\phi}{\mathcal{F}}$) is another random variable which is measurable with respect to $\mathcal{F}$ and which satisfies
\begin{equation}
  \Av \left( \phi I[A] \right) = \Av \left( \condAv{\phi}{\mathcal{F}} I[A] \right)
\end{equation}
for any $A \in \mathcal{F}$.  A version of Theorem~\ref{domConvThm} holds for conditional expectations:
\begin{theorem}
  Let $phi_n$ be a sequence of random variables, and let $\psi$ be a nonnegative random variable on the same space, with $|\phi_n| \le \psi$ and $\phi_n \to \phi$ almost surely.  Then
  \begin{equation}
    \lim_{n \to \infty} \condAv{\phi_n}{\mathcal{F}} = \condAv{\phi}{\mathcal{F}}
  \end{equation}
  almost surely and in $L^1$.
\end{theorem}

There is a great deal to be said about collections of random variables which fail to be independent in particular ways.  To talk about this, we first introduce the idea of a \emph{filtration}, which is a nested sequence of $\sigma$-algebras, $\mathcal{F}_1 \subset \mathcal{F}_2 \subset \cdots$.  A \emph{stochastic process} is a sequence of random variables $\phi_n$ such that each $\phi_n$ is measurable with respect to the corresponding element of the filtration.

A \emph{martingale} is a stochastic process which satisfies
\begin{equation}
  \condAv{\phi_n}{\mathcal{F}_m} = \condAv{\phi_m}{\mathcal{F}_m} \ (\forall n>m).
\end{equation}
There are a number of results about the limits of martingales, which are referred to at various points in the text.

%\section{The $C^*$-algebra formalism of statistical mechanics}
%\section{Concentration of Measure}
%
%At the center of the Aizenman and Wehr proof of the rounding effect is a description of the asymptotic behavior of a set of random quantities defined as functions of a growing collection of independent and identically distributed random variables.  The study of such objects mainly falls under the heading of concentration of measure~\cite{}.  As well as its more foundational uses in probability~\cite{LedouxTalagrand,Ledoux2001concentration,Talagrand1996new} the work done on this subject has found considerable application in the study of spin glasses~\cite{Bovier}, and therefore it seems appropriate to discuss its relationship to the techniques employed in the present work.
%
%The essence of the concentration of measure phenomenon has been summarized by its most noted investigator in the phrase
%
%{\quotation{A random variable that depends (in a ``smooth'' way) on the influence of many independent variables (but not too much on any of them) is essentially constant.\cite{Talagrand1996new}}}
%
%Although this statement in many ways understates what has been achieved, it already contains something important for our purposes: concentration of measure is concerned with smallness of fluctuations, whereas we will need a result which states that the fluctuations are not too small.  Fortunately, we need this only in a fairly specific context.

\chapter{Product Measure Steady States of Generalized Zero Range Processes}\label{GZRP}
(With J. L. Lebowitz, published as~\cite{GZRP})

\section*{Abstract}
We establish necessary and sufficient conditions for the existence of factorizable steady states of the Generalized Zero Range Process on a periodic or infinite lattice.  This process allows transitions from a site $i$ to a site $i+q$ involving (a bounded number of) multiple particles with rates depending on the content of the site $i$, the direction $q$ of movement, and the number of particles moving.  We also show the sufficiency of a similar condition for the continuous time Mass Transport Process, where the mass at each site and the amount transferred in each transition are continuous variables; we conjecture that this is also a necessary condition.

\section{Introduction}
The classical zero range process (ZRP) is a widely studied lattice model with stochastic time evolution \cite{EH2005}.  To define the process consider a cubic box $\Lambda \subset \mathbb{Z}^d$ with periodic boundary conditions, i.e.\ a $d$-dimensional torus.  At each site $i$ of $\Lambda$ there is a random integer-valued variable $n_i \in \{0,1,\ldots\}$, representing the number of particles at site $i$.  The time evolution is specified by a function $\alpha_q(n_i)$ giving the rate at which a particle from a site $i$ containing $n_i$ particles jumps to the site $i+q$, where $q$ runs over a set of neighbors $E$ (the most common choice is $E=\{\pm e_1, \pm e_2, \ldots\}$, but our treatment holds for any finite $E$ which spans $\mathbb{Z}^d$).  The name zero range indicates the fact that the jump rate from $i$ to $i+q$ depends only on the number of particles at $i$.

It is easy to see, when the system is finite, $\alpha_q(n)+\alpha_{-q}(n) \geq \delta >0$ for all $n > 0$ (we always have $\alpha(0) =0$) and $E$ spans $\mathbb{Z}^d$, that all configurations with a given total particle number $N \equiv \sum_{i \in \Lambda} n_i$ are mutually accessible, and hence there is a unique stationary measure $\tilde{P}_\Lambda(\underline{n};N)$ for each $N$.  Normalized superpositions of these measures yield all of the stationary states of this system.  Conversely, given a stationary measure for which there is a nonzero probability of $N$ particles being present in the system one can obtain $\tilde{P}_\Lambda(\underline{n};N)$ by restricting that measure to configurations with $N$ particles.

The ZRP was first introduced in \cite{Spitzer1970}.  It was assumed there and in most subsequent works that the rates $\alpha_q(n)$ are of the form
\begin{equation}
\alpha_q(n)=g_q\alpha(n)
\label{iso}
\end{equation}
with $g_q$ independent of $n$ and $\alpha(n)$ independent of $q$.  In this case the system has the unique steady state given by \cite{Spitzer1970, Andjel1982}
\begin{equation}
\tilde{P}_\Lambda(\underline{n};N)=C_N\delta\left(\sum_{i \in \Lambda}n_i - N\right) \prod_{i \in \Lambda}p(n_i)
\label{pmeasure}
\end{equation}
where
\begin{equation}
p(n)=\frac{c \lambda ^n}{\prod_{k=1}^n \alpha(k)}
\label{singlesite}
\end{equation}
$C_N$ and $c$ are normalization constants given by
\begin{eqnarray}
c = \left(\sum_{n=0}^\infty \frac{\lambda ^n}{\prod_{k=1}^n \alpha(k)}\right)^{-1}\\
C_N = \left(\sum_{\underline{n} : \sum n_i = N} \prod_{i \in \Lambda}p(n_i)\right)^{-1}
\end{eqnarray}
The unique stationary measure $\tilde{P}_\Lambda(\underline{n};N)$ is thus a restriction to configurations with $N$ particles of the product measure $\prod_{i \in \Lambda} p(n_i)$ with single-site distribution $p(n)$.

In the limit $\Lambda \rightarrow \mathbb{Z}^d$ with $N/|\Lambda| \rightarrow \rho$ the only stationary extremal measures, i.e.\ the only stationary measures with a decay of correlations, are product measures with $p(n)$ given in Equation (\ref{singlesite}) as the distribution of single site ocupation numbers \cite{Andjel1982}.  These states are parameterized by $\lambda$, which plays the role of the fugacity in an equilibrium system, with different values of $\lambda$ corresponding to different expected particle densities $\rho$, where
\begin{equation}
\rho=\sum_{n=1}^\infty n p(n) = c \sum_{n=1}^\infty \frac{n \lambda^n}{\prod_{k=1}^n \alpha(k)}
\end{equation}

Recently there has been a revival of interest in the ZRP.  For certain choices of $\alpha(n)$, for example when $\alpha(n) \sim 1 + b/n$ for large $n$, the ZRP on $\mathbb{Z}$ exhibits a transition between a phase where all sites almost certainly contain finite numbers of particles to a `condensed' phase where there is a single site containing an infinite number of particles \cite{Evans2000,GSS2003}.  This condensation has attracted attention as a representative of an interesting class of phase transitions in one dimensional non-equilibrium systems, and has also been applied to models of growing networks \cite{EH2005,KLMST2002}.

Evans, Majumdar and Zia \cite{EMZ2004} have proposed a generalization of the ZRP, called a Mass Transport Model (MTM).  They considered a one dimensional lattice on which there is a continuous `mass' $m_i \geq 0$ at each site, with a parallel update scheme in which at each time step a random mass $\mu_i$, $0 \leq \mu_i \leq m_i$ moves from each site $i \in \mathbb{Z}$ to the neighboring site $i+1$ with a probability density $\phi(\mu,m_i)$.  This process shares many features with the (totally asymmetric) Zero Range Process, in particular the existence of a condensation transition in certain cases \cite{MEZ2005}.  One very significant difference from the ZRP, however, is that the system has a product measure steady state if and only if $\phi(\mu,m)$ satisfies a certain condition.  Taking a limit in which the probability of a transition at any given site at any given time step goes to zero (discussed in \cite{EMZ2004}) gives a stochastic process on continuous time (equivalent to a model with random sequential local updates) which we will call a Mass Transport Process or MTP, in which the rate of transitions from site $i$ to $i+1$ is given by $\alpha(\mu,m)$.  This process has a product measure steady state if and only if there exist functions $g$ and $p$ such that
\begin{equation}
\alpha(\mu,m) = g(\mu)p(m-\mu)/p(m)
\label{mt1d}
\end{equation}

An interesting question, then, is whether similar criteria for the existence of a product measure exist for such processes in higher dimension and with movement in both directions allowed.  This question is already relevant for the ZRP.  A particular case in $d=2$, studied in \cite{vanB, GL}, has
\begin{eqnarray}
\alpha_{\pm 1}(n)=\alpha [ 1 - \delta_{n,0} ]  \\
\alpha_{ \pm 2}(n)=\alpha^{(2)}(n)=\alpha n \label{2drates}
\end{eqnarray}
i.e.\ a constant rate (independent of $n$) per occupied site for moving in the $\pm x$ direction and a rate proportional to $n$ in the $\pm y$ direction.  A treatment of this system based on fluctuating hydrodynamics and computer simulations (originally conducted on a similar but not quite equivalent system, but which we have reproduced on this system) suggests that this particular system has correlations between occupation numbers at different sites a distance $D$ apart decaying according to a dipole power law $D^{-2}$.  This behavior, which is very different from a product measure steady state or its projection (\ref{pmeasure}), is conjectured to be generic for nonequilibrium stationary states of systems with non-equilibrium particle conserving dynamics in $d \geq 2$.

In the present work we prove rigorously that Equation (\ref{iso}) is a necessary and sufficient condition for the existence of product measure steady states for ZRPs, and as a consequence that the system described by (\ref{2drates}) has no product measure steady states.  This condition in turn is a special case of a condition on a class of systems which we call Generalized Zero Range Processes (GZRP), in which we also allow transitions in which more than one particle moves at a time, although we will assume that the number of particles moving in a single transition is bounded.  The rate now depends on the number of particles $\nu$ which move in the transition as well as the number of particles $n$ at the site before the transition, and so the rates are given by a function $\alpha_q(\nu,n)$ with some $\nu_{max}$ such that $\alpha_q(n,\nu)=0$ whenever $\nu > \nu_{max}$.  The classical ZRPs discussed above are a special case with $\nu_{max}=1$ and $\alpha_q(1,n)=\alpha_q(n)$.  We prove that a necessary and sufficient condition for the GZRP to have product measure steady states is
\begin{equation}
\alpha_q(\nu,n)=\frac{g_q(\nu)f(n-\nu)}{f(n)}
\label{condition}
\end{equation}
for some non-negative $g_q(\nu)$ and $f(n)$ with $\sum f(n) < \infty$.  This has a clear similarity to (\ref{mt1d}), and when $\nu_{max}=1$ (\ref{condition}) reduces to (\ref{iso}).

We will prove Equation (\ref{condition}) in the course of finding a weaker result for the continuous-time Mass Transport Process generalized to dimension $d \geq 1$ and to transitions in all directions.  We show that these systems have product measure steady states when $\alpha_q(\mu,m) = g_q(\mu)p(m-\mu)/p(m)$.  This condition is also necessary under certain conditions (generalizing (\ref{mt1d}) to higher dimension) and we conjecture that this is so in all cases.

\section{Factorizability in the Mass Transport Process}
Let $P_\Lambda(\underline{m},t)$ be the time-dependent probability density of finding the system in a particular configuration $\underline{m}$ with mass $m_i$ at site $i \in \Lambda$, $m_i\in (0,\infty)$.  As noted above, we are considering periodic boundary conditions; this case is somewhat simpler than those of other boundary conditions.  The master equation describing the evolution of $P_\Lambda(\underline{m},t)$ is
\begin{equation} \label{master}
\begin{split}
\frac{\partial P_\Lambda(\underline{m},t)}{\partial t} = &\sum_{i \in \Lambda} \left(-\sum_{q \in E} \int_0^{m_i} d\mu \alpha_q (\mu,m_i)P(\underline{m},t) \right. \\
& \left. + \sum_q \int_0^{m_i}d\mu \alpha_{-q}(\mu,m_{i+q}+\mu)P(\underline{m}^{i,q,\mu},t)\right)
\end{split}
\end{equation}
where
\begin{equation}
{m_j^{i,q,\mu}} =
\left\{\begin{array}{ll}m_j, & j \notin \{i,i+q\} \\m_j-\mu, & j=i \\m_j+\mu, & j=i+q\end{array} \right. \label{miqmu}
\end{equation}

A stationary state of the system is a distribution $\tilde{P}_\Lambda(\underline{m})$ such that $\partial P_\Lambda(\underline{m},t)/\partial t = 0$ whenever $P_\Lambda(\underline{m},t)=\tilde{P}_\Lambda(\underline{m})$, or equivalently
\begin{equation}
\sum_{i \in \Lambda} \sum_q \int_0^{m_i} d\mu \alpha_q (\mu,m_i)\tilde{P}_\Lambda(\underline{m}) = \sum_{i \in \Lambda} \sum_q \int_0^{m_i}d\mu \alpha_{-q}(\mu,m_{i+q}+\mu)\tilde{P}_\Lambda(\underline{m}^{i,q,\mu})
\label{stationary}
\end{equation}

We wish to find conditions under which there is a $\tilde{P}_\Lambda$ which is factorizable, that is which takes the form of
\begin{equation}
\hat{P}_\Lambda(\underline{m})=\prod_{i \in \Lambda} p(m_i)
\label{factor}
\end{equation}

Assuming that $\alpha_q(\mu,m) > 0$ for all $m > 0$ and $0 < \mu \leq m$ for at least one $q$ of each pair of opposite directions, the system in a finite torus $\Lambda$ will have a unique steady state corresponding to each value of the total mass $M = \sum m_i$.  Any linear combination of such states is also a solution of (\ref{stationary}).  Given a factorizable steady state, states of definite total mass can be obtained by projecting $\hat{P}$ onto the set of configurations with a particular value of $M$ in analogy with Equation (\ref{pmeasure}).

Assume that there is a factorizable steady state as in (\ref{factor}).  Let $\bar p (s)$ be the Laplace transform of $p(m)$, and let
\begin{equation}
\phi_q(\mu,s) = \left[1/\bar p (s)\right]\int_0^\infty dm e^{-sm}\alpha_q(\mu,m+\mu)p(m+\mu)
\label{phidef}
\end{equation}

Note that, since $\alpha_q(\mu,m)=0$ for $m < \mu$,
\begin{equation}
\begin{split}
\int_0^\infty dm e^{-sm}&\alpha_q(\mu,m)p(m) \\
&= e^{-s\mu}\int_0^\infty dm e^{-sm}\alpha_q(\mu,m+\mu)p(m+\mu) = e^{-s\mu}\phi_q(\mu,s)\bar p (s)
\end{split}
\label{lemma1}
\end{equation}

We also have
\begin{equation}
\begin{split}
\int_0^\infty dm e^{-sm} \int_0^m d\mu \alpha_q(\mu,m)p(m) &= \int_0^\infty d\mu \int_\mu^\infty dm e^{-sm} \alpha_q(\mu,m)p(m)  \\
&= \int_0^\infty d\mu e^{-s\mu} \int_0^\infty dm e^{-sm} \alpha_q(\mu,m+\mu)p(m+\mu) \\
&= \int_0^\infty d\mu e^{-s\mu} \phi_q(\mu,s)\bar p (s)
\end{split}
\label{lemma2}
\end{equation}

Multiplying both sides of (\ref{stationary}) by $\prod_i e^{-s_i m_i}$ and integrating over all $m_i$, we obtain
\begin{equation}
\begin{split}
\sum_{i \in \Lambda,q \in E} \left(\prod_{j \neq i} \bar p (s_j)\right) & \int_0^\infty dm_i \int_0^{m_i}d\mu \alpha_q(\mu,m_i)p(m_i)e^{-s_im_i}  \\
= &\sum_{i \in \Lambda,q \in E} \left(\prod_{j \neq i, i+q} \bar p (s_j)\right) \int_0^\infty dm_i \int_0^\infty dm_{i+q} \int_0^\infty d\mu  \\
& \times \alpha_{-q}(\mu,m_{i+q}+\mu)p(m_i-\mu)p(m_{i+q}+\mu) e^{-s_im_i-s_{i+q}m_{i+q}} \end{split}
\label{intermediate}
\end{equation}

Rewriting (\ref{intermediate}) with the aid of (\ref{lemma1}) and
(\ref{lemma2}) and canceling common factors, we obtain
\begin{equation}
\sum_{i \in \Lambda,q} \int_0^\infty d\mu \phi_q(\mu,s_i) e^{-s_i \mu} = \sum _{i \in \Lambda,q} \int_0^\infty d\mu \phi_{-q}(\mu,s_{i+q})e^{-s_i\mu}
\label{mtcondition}
\end{equation}

Equation (\ref{mtcondition}) will be satisfied if (though not only if)
\begin{equation}
\phi_q(\mu,s)=g_q(\mu)
\label{mtconstant}
\end{equation}

In this case Equation (\ref{phidef}) gives
\begin{equation}
\int_0^\infty dm e^{-sm} \alpha_q(\mu,m+\mu)p(m+\mu) = g_q(\mu) \int_0^\infty dm e^{-sm}p(m)
\end{equation}
which by uniqueness of the Laplace transform gives
\begin{equation}
\alpha_q(\mu,m)=g_q(\mu)\frac{p(m-\mu)}{p(m)} \label{mtsimple}
\end{equation}

Equation (\ref{mtsimple}) is a generalization of the comparable formula for the unidirectional case \cite{EMZ2004}.  In this case and in all other cases where, for each $q \in E$, either $\alpha_q \equiv 0$ or $\alpha_{-q} \equiv 0$ and hence either $\phi_q \equiv 0$ or $\phi_{-q} \equiv 0$, there is in Equation (\ref{mtcondition}) only one term which depends on each pair $m_i,m_{i+q}$, and in order for the equation to be satisfied it must depend on only one of them.  This happens only if (\ref{mtconstant}) holds for that $q$, so in these cases Equation (\ref{mtsimple}) gives the only possible rates for which there is an invariant product measure.

Although in general Equation (\ref{mtconstant}) is not the only way of satisfying Equation (\ref{mtcondition}), solutions of this equation only correspond to realizable dynamics when $p$ and $\alpha_q$ are non-negative and normalizable; the resulting restrictions on $\phi_q$ from Equation (\ref{phidef}) are such that it seems unlikely that there are reasonable (indeed any) rates, other than those in (\ref{mtsimple}), which satisfy all of these conditions.

Dynamics for which the system has a factorizable steady state can be found by beginning with some suitable (positive and normalizable) $p(m)$ and then defining $\alpha_q(\mu,m)$ via (\ref{mtsimple}).  For example let
\begin{equation}
p_c(m) = c e^{-cm} \theta(m)
\end{equation}
\noindent where $\theta$ is the Heaviside step function.  The possible transition rates corresponding to $\tilde P(\underline m ) =\prod p_c(m)$ are
\begin{equation}
\alpha_q(m,\mu) = g_q(\mu) e^{c\mu}\theta(m-\mu) = \tilde{g}_q(\mu) \theta(m-\mu)
\end{equation}
\noindent where $\tilde{g}_q$ are arbitrary non-negative integrable functions, i.e.\ the rates $\alpha_q(\mu,m)$ are independent of $m$ as long as $\mu \leq m$.

\section{Reverse processes}
In this section we will show that Equation (\ref{mtsimple}) is a
necessary condition for the existence of factorizable steady states
of any MTP whose reverse process is also an MTP.  The relevant way
in which the reverse process can fail to be an MTP is that it can
have transition rates which depend on the mass at the target site of
the transition as well as on the mass at the site it is leaving.

In general, given a Markov process with transition rates
$K(\underline{m} \rightarrow \underline{m}')$ and stationary
distribution $\tilde{P}(\underline{m})$, the reverse process is
defined by rates $K^*(\underline{m} \rightarrow \underline{m}')$
given by
\begin{equation}
K^*(\underline{m} \rightarrow \underline{m}') =
\frac{K(\underline{m}' \rightarrow
\underline{m})\tilde{P}(\underline{m}')}{\tilde{P}(\underline{m})}
\label{reverse1}
\end{equation}
This new process is what one obtains by running the original process
backwards.  Consequently the reverse process has the same stationary
distribution as the original process, and when $K$ is translation in
variant so is $K^*$.\footnote{If for some configurations
$\tilde{P}(\underline{m})=0$, then one defines a new process on a
configuration space excluding these configurations (this problem
does not arise in the case under consideration).}

For an MTP defined by rates $\alpha_q(m_i,\mu)$, $K(\underline{m}
\rightarrow \underline{m}')$ is equal to $\alpha_q(m_i,\mu)$ for
configurations $\underline{m}$ and $\underline{m}'$ related by
moving a mass $\mu$ from site $i$ to $i+q$, and to $0$ otherwise.
The reverse process is specified by the rate function
$\alpha_{i,q}^*(\mu,\underline{m})$ which is the rate of transitions
from a configuration $\underline{m}$ in which a mass $\mu$ moves
from site $i$ to site $i+q$; these are the only transitions in this
process.

When $\tilde{P}(\underline{m})$ is a product measure with
single-site weights $p(m)$, (\ref{reverse1}) becomes
\begin{equation}
\alpha_{i,q}^*(\mu,\underline{m}) =
\alpha_{-q}(\mu,m_{i+q}+\mu)\frac{p(m_i-\mu)p(m_{i+q}+\mu)}{p(m_i)p(m_{i+q})}
\label{reverseDef}
\end{equation}
Rewriting, we have
\begin{equation}
\alpha_{i,q}^*(\mu,\underline{m}) \frac{p(m_i)}{p(m_i-\mu)} =
\alpha_{-q}(\mu,m_{i+q}+\mu)\frac{p(m_{i+q}+\mu)}{p(m_{i+q})}
\label{reverse2}
\end{equation}
If $\alpha^*$ defines a mass transport process, then it must be
independent of all $m_j$ for $j \neq i$.  In this case both sides of
(\ref{reverse2}) are equal to some function which depends only on $\mu$ and $q$.  This which can only be true if $\alpha$ satisfies Equation (\ref{mtsimple}), in which case one finds that
\begin{equation}
\alpha^*_{i,q}(\mu,\underline{m}) = \alpha_{-q}(\mu,m_i)
\end{equation}

\section{Factorizability in Generalized Zero Range Processes}
With mass at each site restricted to an integer particle number $n_i$, we can reproduce the analysis in the previous section up to Equation (\ref{mtcondition}).  Denoting the vector of occupation numbers by $\underline{n}$, and the transition rates by $\alpha_q(\nu,n)$, the stationarity condition corresponding to Equation (\ref{stationary}) is
\begin{equation}
\sum_{i \in \Lambda}\sum_{q \in E}\sum_{\nu=1}^{n_i} \left(-\alpha_q(\nu,n_i)\tilde{P}_\Lambda(\underline{n})+\alpha_{-q}(\nu,n_{i+q}+\nu)\tilde{P}_\Lambda(\underline{n}^{i,q,\nu})\right) =0
\label{zrstationary}
\end{equation}

Suppose $\tilde{P}$ is factorizable,
\begin{equation}
\tilde{P}_\Lambda(\underline{n})=\prod_{i \in \Lambda}p(n_i)
\end{equation}
where $p(n)$ is the probability of having $n$ particles at a given site.  Then define the generating function (discrete Laplace transform)
\begin{equation}
\bar{p}(z)=\sum_{n=0}^\infty z^n p(n)
\end{equation}
and let
\begin{equation}
\phi_q(\nu,z)=\frac{\sum_{n=0}^\infty z^n \alpha_q(\nu,n+\nu)p(n+\nu)}{\bar{p}(z)}
\label{zrphidef}
\end{equation}
Note that $\phi_q(\nu,z) \geq 0$ for all $\nu,z \geq 0$.
The counterpart of Equation (\ref{mtcondition}) is then
\begin{equation}
\sum_{i \in \Lambda,q \in E} \sum_{\nu=1}^\infty z_i^\nu \phi_q(\nu,z_i) = \sum_{i \in \Lambda,q \in E} \sum_{\nu=1}^\infty z_i^\nu \phi_{-q}(\nu,z_{i+q})
\label{zrcondit}
\end{equation}

We now exploit the assumption that transitions occur only for $\nu \leq \nu_{max}$.  Then choosing some $j \in \Lambda$ and $\tilde{q} \in E$ and taking the $k$th derivative of the above expression with respect to $z_j$ and $z_{j+\tilde{q}}$ gives
\begin{equation}
\sum_{\nu=k}^{\nu_{max}}\frac{\nu!}{(\nu-k)!}z_j^{\nu-k}\phi_{-\tilde{q}}^{(k)}(\nu,z_{j+\tilde{q}})+
\sum_{\nu=k}^{\nu_{max}}\frac{\nu!}{(\nu-k)!}z_{j+\tilde{q}}^{\nu-k}\phi_{\tilde{q}}^{(k)}(\nu,z_j)=0
\label{dvcondition}
\end{equation}

For $k=\nu_{max}$, we have
\begin{equation}
\phi_{\tilde{q}}^{(\nu_{max})}(\nu_{max},z_j)+\phi_{-\tilde{q}}^{(\nu_{max})}(\nu_{max},z_{j+\tilde{q}})=0 \label{inducStart}
\end{equation}
For (\ref{inducStart}) to hold for all $z_j$ and $z_{j+\tilde{q}}$, both terms on the left-hand-side must be constant, and thus the functions $\phi_{\pm \tilde{q}} (\nu_{max},\cdot)$ are polynomials of degree $\nu_{max}$; being non-negative they must have non-negative leading terms.  Equation (\ref{inducStart}) states that pairs of leading terms of these polynomials must add up to zero and so each must be zero, and therefore the functions $\phi_q (\nu_{max},\cdot)$ are polynomials of degree at most $\nu_{max}-1$ for each $q$.

Now setting $k=\nu_{max}-1$ we find by the same reasoning that the functions $\phi_q(\nu_{max},\cdot)$ are polynomials of degree at most $\nu_{max}-2$.  Proceeding in this manner we find
\begin{equation}
\phi_q(\nu,z) = g_q(\nu)
\label{zrcondition}
\end{equation}
as a necessary as well as a sufficient condition for (\ref{zrcondit}) to be satisfied.  Referring to the definition of $\phi$, this implies that
\begin{equation}
\alpha_q(\nu,n)=g_q(\nu) \frac{p(n-\nu)}{p(n)}
\label{gzrpcond}
\end{equation}
is a necessary and sufficient condition for the existence of a product measure.

In the case where $\nu_{max}= 1$ and $\alpha_q(1,n)=\alpha_q(n)$, this condition becomes
\begin{equation}
\alpha_q(n)=c_q\frac{p(n-1)}{p(n)}
\end{equation}
This is what we referred to above as the classical ZRP, with the well-known stationary measure \cite{Spitzer1970, Evans2000} discussed in the introduction.

\section{GZRPs on infinite lattices}
In order to show that we have really found all of the factorizable steady states of this class of systems, it remains to be established that the conditions obtained also apply to an infinite lattice; that is, that there are not rates for which the resulting GZRP on an infinite lattice has product measure steady states while the GZRPs defined on finite lattices have none, or do not have the same such stationary states.

Let $P(\underline{n})$ be a product measure with single-site distribution $p(n)$ which is stationary for rate functions $\alpha$ on $\mathbb{Z}^d$, and let $\Lambda$ be a finite box in $\mathbb{Z}^d$ such that there is some $i_0 \in \Lambda$ such that $i_0+q \in \Lambda$ for all $q \in E$.  Denote by $\underline{n}_\Lambda$ the configuration of the system inside of $\Lambda$, and let $P(\underline{n}_\Lambda)$ be the marginal distribution of this configuration.  Then we have
\begin{equation}
\begin{split}
\frac{d}{dt} P(\underline{n}_\Lambda) = & - \sum_{i \in \Lambda} \sum_{q \in E} \sum_{\nu=1}^{n_i} \alpha_q(\nu,n_i) P(\underline{n}_\Lambda)
+ \sum_{i \in \Lambda} \sum_{q \in E} \sum_{\nu=1}^{n_{i+q}} \alpha_q(\nu,n_i+\nu)P(\underline{n}_\Lambda^{i,q,\nu}) \\
& - \sum_{i \in \partial \Lambda} \sum_{q \in E : i+q \in \Lambda} \sum_{\nu=1}^\infty\sum_{n=\nu}^\infty \alpha_q(\nu,n) P(\underline{n}_\Lambda) p(n) \\
& + \sum_{i \in \partial \Lambda} \sum_{q \in E : i+q \in \Lambda} \sum_{n=1}^\infty\sum_{\nu=1}^n \alpha_q(\nu,n) P(\underline{n}^{i,q,\nu}_\Lambda) p(n) \\
=&0
\end{split}
\label{generator}
\end{equation}
where $\partial \Lambda = \{i \in \mathbb{Z}^d \setminus \Lambda|(\exists q \in E)(i+q \in \Lambda)\}$ and
\begin{equation}
n^{i,q,\nu}_k = \left\{\begin{array}{ll}
n_k,&k \notin \{i,i+q\} \cap \Lambda \\
n_k+\nu,&k = i \in \Lambda\\
n_k - \nu,&k=i+q \in \Lambda\end{array}\right.
\end{equation}
Equation (\ref{generator}) is very similar to Equation (\ref{zrstationary}), and by repeating the procedure used above with Equation (\ref{generator}) in place of Equation (\ref{zrstationary}), it can easily be seen that $\alpha$ and $p$ must satisfy Equation (\ref{dvcondition}) and so that (\ref{gzrpcond}) is a necessary condition for the existence of a product measure steady state of the process on $\mathbb{Z}^d$ as well as on a finite torus.

\section{Conclusion}

We have shown that there is a straightforward necessary and sufficient condition, Equation (\ref{condition}), for a generalized Zero Range Process to have a product measure steady state.  For Mass Transport Processes, we have found a condition, Equation (\ref{mtcondition}), for the existence of a product measure steady state, which is considerably more opaque than in the GZRP; it is not clear that this is equivalent to the sufficient condition expressed in Equation (\ref{mtsimple}), the counterpart of the condition we have obtained for GZRPs.  We have, however, presented some reasons to believe that it is.

\section*{Acknowledgements}
The work of RLG was supported by Department of Education grant P200A030156-03.  The work of JLL was support by NSF grant DMR 01-279-26, and by AFOSR grant AF 59620-01-1-0154.  We thank P A Ferrari for a fruitful discussion at the Institut des Hautes \'Etudes Scientifiques in Bures-Sur-Yvette, France, where part of this work was done.  We also thank R K P Zia for useful discussions.

\bibliographystyle{spphys}
\bibliography{thesis,zrp}

\begin{vita}

\heading{Curriculum Vitae} \vspace{15pt}
\heading{Rafael L Greenblatt} \vspace{15pt}
\begin{descriptionlist}{xxxxx-xxxxx}
\item[2002-2010] Ph.D. in Physics, Rutgers University
\item[1998-2001] B.A. (High Distinction) in Mathematics and Physics, University of California Berkeley
\item[1996-1998] Attended San Diego City College
\end{descriptionlist}
\medskip
\newcommand{\TA}{Teaching assistant, Department of Physics and Astronomy, Rutgers University}
\newcommand{\GA}{Departmental graduate assistant, Department of Physics and Astronomy, Rutgers University}
\begin{descriptionlist}{xxxxx-xxxxx} %positions held since BS degree
\item[2009] \TA
\item[2008-2009] Part time lecturer, Department of Physics and Astronomy, Rutgers University
\item[2006] \TA
\item[2005] \GA
\item[2003-2005] Department of Education GAANN fellow
\item[2002-2003] \TA
\end{descriptionlist}

Publications:
\begin{itemize}
\item Comment on ``Yang-Lee Zeros for an Urn Model for the Separation of Sand'' (with J.L. Lebowitz), {\it Phys. Rev. Lett.} {\bf 93} 238901, 2004.
\item Product measure steady states of generalized zero range processes (with J.L. Lebowitz), {\it J. Phys. A} {\bf 39} 1565--1574, 2006.
\item Rounding of First Order Transitions in Low-Dimensional Quantum Systems with Quenched Disorder (with M. Aizenman and J.L. Lebowitz), {\it Phys. Rev. Lett.} {\bf 103} 197201, 2009.
\item Phase diagrams of systems with quenched randomness (with M. Aizenman and J.L. Lebowitz), ArXiv:0912.1251, to appear in {\it Physica A}, 2010.
\end{itemize}
\end{vita}

\end{document}